\documentclass[format=acmsmall, review=true]{acmart}
\usepackage{acm-ec-23}
\usepackage{booktabs} 
\usepackage[ruled]{algorithm2e} 

\usepackage{verbatim}
\usepackage{changepage}
\usepackage[font=footnotesize,labelfont=bf]{caption}
\usepackage[font=footnotesize]{subcaption}
\usepackage{algorithmic}
\usepackage{thm-restate}
\usepackage{amsmath,mathtools}
\usepackage[english]{babel}
\usepackage{multirow}
\usepackage{soul}
\usepackage{float}
 \usepackage[english]{babel}

\bibpunct[, ]{(}{)}{,}{a}{}{,}

\usepackage{color}

\usepackage{optidef}

\newif\ifpfapp  
\pfapptrue      
\pfappfalse     

\def\naturals{\mathbb{N}}
\def\reals{\mathbb{R}}

\def\election{\ensuremath\mathcal{E}}
\def\districtplans{\ensuremath\mathcal{D}}
\def\districtplan{\ensuremath D}
\def\voterballots{\ensuremath\mathcal{V}}
\def\voterballot{\ensuremath V}
\def\fairnessfunction{\ensuremath f}

\def\fairnessthreshold{\ensuremath \delta}
\def\pool{\ensuremath \mathcal{P}}
\def\shiftedvotes{\ensuremath x}
\def\poolsize{\ensuremath N}

\def\Set#1{\left\{ #1 \right\}}
\def\Abs#1{\left| #1 \right|}

\def\Paren#1{\left( #1 \right)}		
\def\Brack#1{\left[ #1 \right]}		

\def\govmap{\texttt{GOV2021}}

\setcitestyle{authoryear}

\title[Votemandering]{Votemandering: Strategies and Fairness in Political Redistricting}

\author{Sanyukta P. Deshpande, Ian G. Ludden, Sheldon H. Jacobson }

\begin{abstract}
Gerrymandering, the deliberate manipulation of electoral district boundaries for political advantage, is a persistent issue in U.S. redistricting cycles. 
In this work, we introduce and analyze Votemandering, a strategic blend of gerrymandering and targeted political campaigning devised to gain more seats by circumventing fairness measures. Votemandering leverages accurate demographic and socio-political data, bolstered by advancements in technology and data analytics, to influence voter decisions in pursuit of subtle gerrymandering strategies. We formulate votemandering as a Mixed Integer Program (MIP) that performs fairness-constrained gerrymandering over multiple election rounds. 
To combat votemandering, we present a computationally efficient heuristic for creating and testing district maps that more robustly preserve voter preferences.  
We analyze the influence of various redistricting constraints and parameters on votemandering efficacy.  We explore the interconnectedness of gerrymandering, substantial campaign budgets, and strategic campaigning, illustrating their collective potential to generate biased electoral maps. A case study of Wisconsin State Senate redistricting substantiates our findings on real data, demonstrating how major parties can secure additional seats through votemandering. Our findings underscore the practical implications of these manipulations, stressing the need for informed policy and regulation to safeguard democratic processes.
\end{abstract}

\begin{document}

\begin{titlepage}

\maketitle

\end{titlepage}
\section{Introduction}
\emph{Partisan gerrymandering} is the manipulation of voting district lines for political gain. There is numerous evidence of gerrymandering in the US electoral history, giving unfair political advantage to various parties in power \cite{bickerstaff2020election}. In an effort to detect and quantify gerrymandering, political scientists have devised various \emph{fairness measures}, some of which incorporate historical voting data. If a proposed district plan has a fairness measure outside of a typical range, or, more robustly, is an outlier with respect to a fairness measure over a sample of feasible plans, this anomaly offers evidence of partisan gerrymandering. However, federal courts have refrained from endorsing proposed fairness measures as gerrymandering litmus tests, indicating a need for further research on the robustness and trade-offs of such tests (Rucho v. Common Cause \citeyear{rucho}).

In addition to gerrymandering, political parties seek to enhance their political representation through huge campaign budgets \citep{horncastle2020scale,evers2021most}. 
Although campaigning alone cannot change the party inclinations of voters, it supports the Get Out The Vote (GOTV) cause, increasing voter turnout \citep{karp2008getting, imai2011estimation}.
Recent GOTV campaigns carefully target specific audiences for maximum impact, 
leveraging advanced machine learning algorithms that use voter data (collected through geographical surveys and the available telemetric data) to deliver information about the political inclination of the audience \citep{zarouali2020using}.
Once the targets are clear, personalized campaigns are delivered through direct messages or via social media advertisements. Such campaign efforts have been used in both the 2016 and 2020 U.S. presidential elections, where clear evidence of the effectiveness of the advertisements as well as research scrutinizing the implications of presenting the social choice surfaced \citep{liberini2020politics, brodnax2022home}. 
The implications of such precise campaign efforts become critical, 
as historical election data (influenced by the campaigns) 
are often used to judge the fairness of proposed maps. 
A question of interest is then studying how smart campaign strategies can simultaneously affect immediate elections and future redistricting, and help in securing even higher political representation. 

 This paper aims to investigate the robustness of fairness measures to strategic campaigning and traditional gerrymandering, which we term \emph{votemandering}. Votemandering is based on the idea that a party can strategically campaign in an election to alter the voting data and then draw a new district plan that appears fair for a fixed fairness measure, but gives them an unfair advantage in the next election. The focus is on identifying patterns of selective and disproportionate amendments to the representation of social choice through voting, to circumvent fairness measures for redistricting. This manipulation can be critical, particularly when slight deviations in election data can lead to significantly different fairness measure evaluations.  On this background, the paper seeks to address the following research questions:
\begin{itemize}
\item How vulnerable are popular fairness measures to votemandering?
\item How might voter turnout levels and political geography exacerbate or thwart votemandering?
\item Can careful combinations of fairness measures and legal constraints promote district plans that are more robust to votemandering?
\end{itemize}
These questions delve into both social choice theory and practical public policy considerations. A significant concern arises when technological advancements enable greater access to detailed data on voters' preferences and increase the capacity to influence decisions, thus allowing strategic actors to target specific communities in ways that undermine fairness and equity.



We next present a brief  description of the problem: 
Consider two election rounds with a redistricting cycle falling in between. 
The majority party in the state legislature, referred to as the "majority party," campaigns in the first election, winning the maximum number of seats while simultaneously ensuring they can draw a desired district plan for the second election, which appears fair. Fairness is measured by a metric that uses past election data, such as the efficiency gap (EG), which is influenced by campaigning. Assuming complete information about the opponent party's Get Out The Vote (GOTV) campaign, the goal of the majority party is to maximize the number of seats won in both rounds. We refer to this as votemandering and formulate an optimization framework that identifies the best campaign strategies with the combined objective of securing maximum wins in both rounds and drawing a desired map that remains valid for many years.  Motivated by practical and often legal constraints on redistricting in the US, we also analyze the special case of imposing proximity constraints for the proposed maps, i.e., making the least changes to the original plan while proposing a new plan, calling it \emph{local votemandering}.
 Through this research, we aim to shed light on the unreliability in the process of redistricting (and detecting gerrymandering), and further point at measures that ensure more robust maps in general. 
\\
Key takeaways from this work include:
\begin{enumerate}
    \item We demonstrate that fairness measures can be susceptible to data manipulation, leading to an indirect form of gerrymandering called votemandering. Therefore, the quality of a fairness measure can also be defined by its robustness against strategic amendments to the vote-share data. We formally model this phenomenon and discuss the case of the efficiency gap.
    \item  We show the fragility of district maps concerning votemandering and establish sufficient conditions for a party to benefit from it. We show how 
    campaign budgets and access to opponents' campaign information facilitate; 
    high voter turnout and stricter compactness bounds curtail; 
    and voter clustering patterns have little effect on votemandering. 
    \item 
   We lay the groundwork for creating and evaluating district plans that strongly preserve social choice, providing computationally efficient votemandering solutions. Our work is applicable to real-world data, as demonstrated by the case studies. 
\end{enumerate}

The remainder of the paper is structured as follows. 
Section \ref{sec: litreview} summarizes literature from various disciplines that connect methodologically or philosophically. 
Section \ref{sec:methodology} formally defines votemandering, expounding the model and methods. 
Section \ref{sec:results} proves the efficacy and computability of votemandering specific to the efficiency gap and further explores its sensitivity to state-specific factors such as voter distribution and nonpartisan redistricting constraints. 
Section \ref{sec:local_vm} defines and analyzes local votemandering, a variation with the constraint that the new district plan is close to the original. 
Section \ref{sec:case_study} applies votemandering to  
the  case of state senate redistricting in Wisconsin, demonstrating votemandering strategies for both major parties. 
Finally, Section \ref{sec:conclusion} concludes and outlines 
directions for future work.

\section{Related Literature}
\label{sec: litreview}
This paper connects to a rich body of work from the perspectives of social choice, game theory, optimization, and statistics.

\paragraph{Social Choice Theory}
Social choice theory studies and evaluates the translation of individual preferences or votes to collective societal decisions \citep{sen1986social}. In our work, we examine the impact of strategic campaigning on political redistricting, which may be easily translated to a form of strategic voting aimed at manipulating social choice. The pure form of strategic voting has been studied for decades, although the focus has been more on various voting mechanisms and their evaluation using strategy-proofness, Pareto efficiency, independence of irrelevant alternatives, etc \citep{lackner2018approval, myatt2007theory}. As famously shown by \citet{gibbard1973manipulation} and \citet{satterthwaite1975strategy}, no voting system for more than two players is strategy-proof. \citet{bartholdi1989computational} came up with a voting rule where it is NP-complete for manipulative voters to perform strategic voting, and also noted that many voting rules including the plurality rule can be manipulated with only polynomial computational effort. As we see within our framework, finding optimal strategies for votemandering is hard, but good solutions can be achieved with little computational effort. 

\paragraph{Manipulations within Plurality Voting}
Within the domain of plurality voting, such problems have also been studied from a computational theory point of view, while making a few abstractions on the redistricting part. \citet{cohen2018gerrymandering} study the problem of gerrymandering over graphs and show that the problem of dividing a social network into connected components is NP-complete, and \citet{ito2021algorithms} build over their settings. \citet{lewenberg2017divide, eiben2020manipulating} have studied another variant involving geographic manipulation of borders and location of districts. In another interesting work by \citet{stewart2019information}, information gerrymandering has been studied where the structure of the influence network manipulates the voting outcomes, along with newly placed zealots. \citet{lev2019reverse} study reverse gerrymandering in multi-group decision-making systems, where agents move across units to maximize their influence. The game of allocating optimal resources for campaigning has been modeled as the classic Colonel Blotto game, and its complexity, as well as equilibria, are studied \citep{behnezhad2017faster,behnezhad2018battlefields, macdonell2015waging}, although without examining the subsequent consequences on redistricting.

\paragraph{Quantifying District Plan Fairness}
Lately, with a lot of research being done on finding ways to fairly draw the district boundaries and on knowing if a particular map is gerrymandered \citep{swamy2022multiobjective, landau2009fair, chikina2017assessing, benade2021you}, there has been a growing interest in defining measures to judge the fairness of a map. With multiple redistricting processes reaching the Supreme Court \citep{royden2017extreme}, and the latter relying on ongoing research for the mathematical analysis (Pennsylvania Case \citeyear{penncase}), we ask if there are any strategies for fooling the measures while drawing the politically motivated map boundaries. We study a different form of strategic voting, where the strategies are implemented by the political parties, although carried out through a section of voters. In our work, we introduce a new criterion for the evaluation of voting mechanisms as well as the fairness of the district maps, stressing on the fact that the representation of social choice through voting is inherently connected to redistricting. 

\paragraph{Fooling Fairness Measures}
The idea of fooling the measures that are actually designed for achieving fairness is not new. Starting with \citet{adsul2010nash}, there has been a lot of work in the field of fair division in algorithmic game theory \citep{branzei2017nash, babaioff2021competitive}. By manipulating the preference data of buyers, the fairness criteria of allocation results in higher utility for the strategic players. In the field of redistricting, our work is philosophically the closest to \citet{brubach2020meddling}, where the authors study the effects of fairness measurements on voting strategies. Using the outlier detection method, the work heuristically studies the game of strategic voting where loyal voters alter their votes as directed by their political party. Building on this work and also addressing some open questions raised, we demonstrate our results using indirect manipulation of voter turnout through selective campaigning, and we use a popular fairness measure called the efficiency gap. 

\paragraph{The Efficiency Gap and its Shortcomings}
\citet{stephanopoulos2015partisan} introduce 
the \emph{efficiency gap} (EG) fairness measure 
to quantify partisan gerrymandering. 
EG is a fairly straightforward measure 
that computes the difference between the wasted votes of two major parties and labels a map as  unfair if a party disproportionately wastes more votes than the other. It has been widely used because of its simplicity,  intuition, and the use of actual voter preference data from the elections (Gill v Whitford \citeyear{gill}, Missouri Constitution \citeyear{missouri}). 

With the widespread use of EG, there has also been growing literature on the shortcomings of EG, typically focusing on its implications and the nature of it being a single-dimensional number trying to capture the complex forms of gerrymandering \citep{bernstein2017formula, kean2018flaw}. In particular, \citet{chambers2017flaws} have majorly criticized the philosophical implications like the possible increase in polarization, the problematic ranking of maps and  technical implications like discouragement in proportional representation. Moreover, the EG has also been criticized for favoring uncompetitive elections and voter suppression \citep{plener2018quantifying} and volatility in competitive elections leading to a high number of asymmetric wasted votes as well as for terming 3 to 1 victories as neutral \citep{bernstein2017formula}. \citet{tam2017measuring} have observed the problematic variations of EG implications across states with comparable vote shares; and further have talked about the limited number of values EG can take for any fixed vote shares. 
\citet{alexeev2018impossibility} show that 
sometimes only bizarrely shaped districts satisfy both population balance and EG constraints. 
Other philosophical shortcomings include the definition and weights of wasted votes (winner's surplus wasted votes, loser's all votes are weighted the same), incorrect reporting of the social choice, and bias to the winning party \citep{nagle2019criteria, barton2018improving}. Numerous updates are proposed to the current computation of the EG \citep{barton2018improving, tam2017measuring}, while also criticizing the implication of wasted votes being improperly biased towards districts with higher voting turnout \citep{wallin2017equal}. We note that the main criticism offered by our work is fundamentally independent of the previous work done on evaluating the EG, and our main focus is on the sensitivity of EG and its susceptibility to getting fooled in the broader context of votemandering.

\section{Methodology}
\label{sec:methodology}
In this section, we formally discuss the votemandering model and our methodology. 
Section \ref{sec: highleveloverview} sets the premise with a high-level description of the problem and Section \ref{subsec:vm_model} formally expounds the model.
Section \ref{subsec:vm_opt} presents votemandering as an optimization problem, applicable to a general fairness metric using past-election data. 
Finally, Section \ref{subsec:methods} outlines a two-stage heuristic approach to solving the votemandering optimization problem and describes the specific case of EG.
\subsection{High-level Votemandering Model} \label{sec: highleveloverview}
We begin by defining a function, $\election:\districtplans \times \voterballots \rightarrow \naturals$, to determine state-wide election winners. This function maps a district plan, $\districtplan \in \districtplans$, and a set of voter ballots, $\voterballot \in \voterballots$, to the number of districts won by party $A$ in the election. The election function, $\election$, represents a specific electoral system, such as single-member districts with first-past-the-post voting. Although the voting data, $\voterballot$, may be influenced by stochastic processes like migration and political dialogue, $\election$ is deterministic.

In this framework, partisan gerrymandering involves replacing $\districtplan$ with $\tilde{\districtplan}$ to win more districts, i.e., $\election(\tilde{\districtplan}, \voterballot) > \election(\districtplan, \voterballot)$. Similarly, election campaigning alters $\voterballot$ to $\tilde{\voterballot}$ to secure more districts: $\election(\districtplan, \tilde{\voterballot}) > \election(\districtplan, \voterballot)$.
Note that election campaigning is generally considered fair within the confines of the Federal Election Campaign Act.

Existing approaches to limit partisan gerrymandering involve calculating a fairness measure, $\fairnessfunction : \districtplans \times \voterballots \rightarrow \reals$, and rejecting a district plan $\districtplan$ if and only if $\fairnessfunction(\districtplan, \voterballot_{0}) > \fairnessthreshold$. Here, $\voterballot_{0}$ represents historical voting data, and $\fairnessthreshold$ is a predetermined threshold. This fairness constraint aims to reduce the strategic impact of gerrymandering on election outcomes. However, as noted by \citet{brubach2020meddling}, partisan agents may manipulate voting data in one election to make a future gerrymandered district plan appear fair. Let $\districtplan_0$ represent the current district plan. The manipulative partisan agent, party $A$, attempts to solve the optimization problem:
\begin{maxi}
{\tilde{\districtplan} \in \districtplans, \tilde{\voterballot} \in \voterballots}{\election(\districtplan_0, \tilde{\voterballot}) + \election(\tilde{\districtplan}, \voterballot_0)\label{opt:general_votemander}}{}{}
\addConstraint{\fairnessfunction(\tilde{\districtplan}, \tilde{\voterballot})}{\le \fairnessthreshold.}
\end{maxi}


Previous research on elections and redistricting has focused on the effects of either $\tilde{\districtplan}$, $\tilde{\voterballot}$, or $\fairnessfunction$. In contrast, this paper investigates the efficacy of votemandering, which combines gerrymandering and past or present campaigning, primarily in opposition to a specific partisan bias measure, such as the efficiency gap (EG). The votemandering framework assumes translation of campaign budgets to improved voter turnout and an access to other party's budget allocation information, although it is fairly robust to overcome small uncertainties within the data, as discussed in Section \ref{sec: experimentalresults}.

\subsection{Model Details and Terminology}
\label{subsec:vm_model}
Consider two political parties: 
the (state legislative) majority party, $A$, and the minority party, $B$. 
Party $A$ is assumed to be in-charge of the redistricting process, in line with requirements of majority of the states in the US \citep{brennan}. Suppose parties $A$ and $B$ compete in two rounds of elections 
with a redistricting cycle in between. 
By examining this narrow time window, 
our model studies only short-term implications of campaigning, 
affecting the round-1 election and the subsequent map-drawing process. 

Recall the high-level votemandering optimization problem \eqref{opt:general_votemander}. Set the electoral system, $\election$, as single-member districts with first-past-the-post voting. Function $\fairnessfunction$ represents the fairness measure, such as the EG. We distinguish between \emph{plan} and \emph{map}, with the former indicating unit-to-district assignments and the latter encompassing both a district plan and unit-level voter data. 

Figure \ref{fig: model} illustrates the stages of votemandering. In round-1, elections use the existing district plan, $\districtplan_0$, with voter ballots $\tilde{\voterballot}$ resulting from GOTV campaign efforts. We refer to $\voterballot_0$ as the \textit{original data} and $\tilde{\voterballot}$ as the \textit{new data}. We label $(\districtplan_0,\voterballot_0)$ the \emph{initial map} and $(\districtplan_0, \tilde{\voterballot})$ the \emph{campaigned map}. Following round-1, party $A$ creates a new district plan, $\tilde{\districtplan}$, satisfying fairness constraints using voter data from the round-1 elections, i.e., the new data. Round-2 elections employ the new plan, $\tilde{\districtplan}$, but with the original data, $\voterballot_0$. We designate $(\tilde{\districtplan}, \tilde{\voterballot})$ as the \emph{votemandered map} and $(\tilde{\districtplan}, \voterballot_0)$ as the \emph{target map}. The reversion to $\voterballot_0$ in round-2 implicitly assumes party $A$ can precisely match party $B$'s GOTV budget allocation, negating any increases in voter turnout. We do not model campaign budget strategies in round-2 to avoid added complexity and, more importantly, to concentrate on showcasing the ability to manipulate vote shares for generating a desired map while still appearing to uphold fairness.

\begin{figure}
    \centering
    \includegraphics[width=8.3cm]{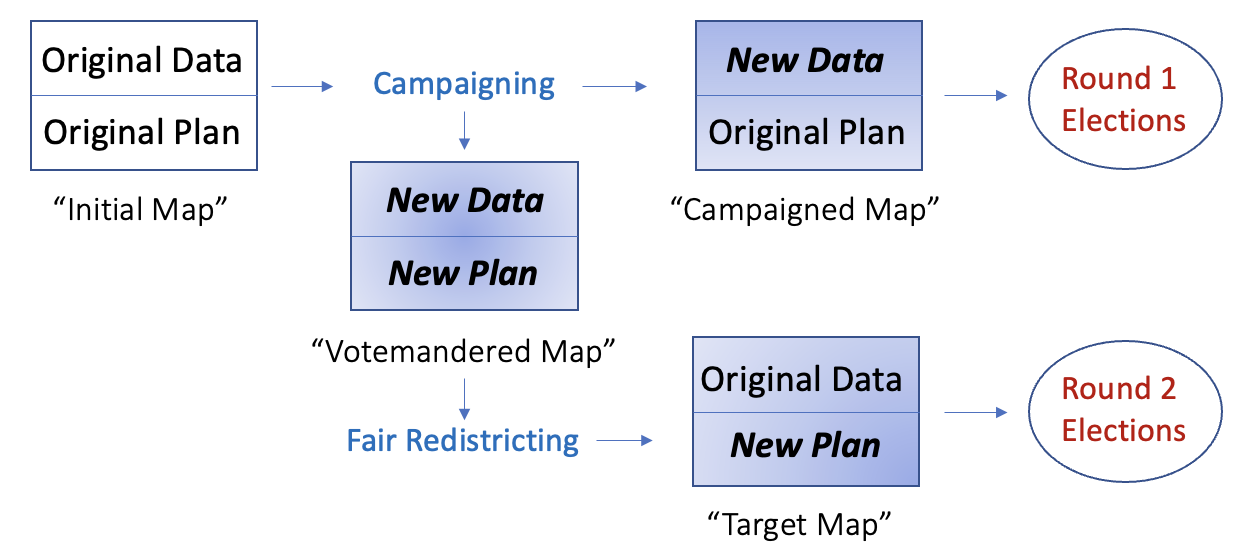}
    \caption{The model highlighting various stages of votemandering}
    \label{fig: model}
\end{figure}

Party $A$'s strategic GOTV campaign in round-1 influences their seat-share in both election rounds: directly through wins in the campaigned map, $\election(\districtplan_0,\tilde{\voterballot})$, and indirectly through wins in the target map, $\election(\tilde{\districtplan}, \voterballot_0)$. To examine the strategies of the majority party, we fix party $B$'s budget allocation across all units and consider party $A$'s optimization problem \eqref{opt:general_votemander} of maximizing their total number of seats. Wins in round-2 are critical because the target map will remain in effect until the next redistricting phase. To cover a complete redistricting cycle (such as a 10-year period in U.S. elections), the model can be extended by adjusting the weight of round-2 wins accordingly. Furthermore, the model accommodates the inclusion of aggregated historical data from multiple elections by appropriately adjusting the weight attributed to campaign influence. 

It is essential to emphasize that votemandering is fundamentally different from both strategic campaigning and gerrymandering due to its interactions between stages. As Section \ref{sec:case_study} demonstrates,  even modest budget allocations can lead to significant votemandering outcomes, setting it apart from traditional campaigning by incorporating additional gerrymandering tactics. An example in Appendix \ref{app: example} illustrates the votemandering process.

\subsection{ Optimization Framework for Votemandering}
\label{subsec:vm_opt}
The votemandering model motivates an optimization framework 
for exploring potential campaign and redistricting strategies for party $A$. Table \ref{tab: notation} in Appendix \ref{app:notation} lists the notation built.

\subsubsection{State Characteristics.}
\label{subsubsec:state_chars}
Let $K$ denote the set of units in a state 
with $n$ districts. 
Each district designates one unit as its \emph{center}. 
The district assignment of each unit $j \in K$ in each round $r \in \Set{1,2}$ 
is represented by the indicator variables 
    \begin{equation*}
    z_{ij}^{r}=
    \begin{cases}
      1, & \text{if unit $j$ is assigned to  the district centered at unit $i$ in round $r$} \\
      0, & \text{otherwise.}
    \end{cases}
    \end{equation*}
Moreover, $z_{ii}^{r} = 1$ if $i$ is a district center in round $r$. 
The original district plan, $\districtplan_0$, determines 
the values of $z_{ij}^1$, and all $z_{ij}^{2}$ are decision variables. 
The following constraints enforce the proper formation of districts in round-2.
\begin{align}
    & \sum_{k \in K} z_{ki}^{2} = 1 & \forall i \in K \label{eq: binary1}\\
    & \sum_{k \in K} z_{kk}^{2} = n  \label{eq: binary2}
\end{align}
Constraint \eqref{eq: binary1} ensures every unit is assigned to some district, and 
constraint \eqref{eq: binary2} ensures exactly $n$ units are chosen as district centers. 


\subsubsection{Budget and Campaigning.}
\label{subsubsec:budget_campaigning}
Assume complete information about unit populations and the corresponding party affiliations, i.e., the maximum number of voters 
for each party in each unit. 
The maximum vote counts for party $A$ ($B$) are given by $v_{init, k}^A$ ($v_{init,k}^B$) in unit $k \in K$, with total unit population $p_k = v_{init,k}^A+v_{init,k}^B$. 
Let $\alpha \in \Brack{0,1}$ denote the fractional baseline voter turnout, 
assumed constant across all units. 
The number of party $A$ votes is the sum of $\alpha v_{init, k}$ and 
the votes through GOTV campaigning in unit $k$. 
The vote shares $v_{init,k}$ and $\alpha$ are fixed for all rounds, and the actual voting turnout varies depending on  campaigning.

Let $\mathcal{B}^A$, $\mathcal{B}^B$ 
denote the parties' GOTV campaign budgets 
in terms of the total number of their supporters 
they can convince to show up to the polls. 
Budget allocations in unit $k$ by party $P \in \Set{A,B}$ 
may push their actual number of votes above 
the baseline turnout $\alpha v_{init,k}^{P}$, 
but their total number of votes cannot exceed $v_{init,k}^{P}$. 
(Hence if $\alpha = 1$, then GOTV budget allocations have no effect.)
This constraint is implemented by defining $v_k^{A}$ as the actual voter turnout for party $A$ in unit $k$, with budget $b_k^A$ spent satisfying
\begin{align}
   & v_k^{A} = \alpha v_{init,k}^{A} + b_k^A && \forall k \in K \label{eq: budgetcapA}\\
   & v_k^{B} = \alpha v_{init,k}^{B} + b_k^B && \forall k \in K \label{eq: budgetcapB}\\
   & b_k^A  \leq (1-\alpha) v_{init,k}^{A} && \forall k \in K \label{eq: budgetcap2}\\
   & \sum_{k \in K} b_{k} \leq \mathcal{B}^A \label{eq: budgetconstraint}
\end{align}
By assumption, 
party $B$'s GOTV campaign allocation, 
and therefore the values $v_k^B$, are known to party $A$. 


\subsubsection{Winning Districts.}
\label{subsubsec:winning_districts}
A party must win more than half of the votes in a district to secure a win. We use indicator variables  $\hat{s}_i^1 $ and $\hat{s}_i^{2}$ and the big-$M$ method for incorporating the wins in campaigned and target maps respectively
 \begin{align}
& 1-M(1-\hat{s}_i^{1}) \leq \sum_{k \in K} z_{ik}^{1} \left(v_k^A- v_k^B\right) \leq M \hat{s}_i^{1} & \forall i \in K \label{eq: round1wins}\\
&  1-M(1-\hat{s}_i^{2}) \leq \sum_{k \in K} z_{ik}^{2} \left( v_{init,k}^A- v_{init,k}^B\right) \leq M \hat{s}_i^{2} & \forall i \in K \label{eq: round2wins}
 \end{align}
Note that constraints \eqref{eq: round1wins} and \eqref{eq: round2wins} are both linear: $z_{ik}^{1}$ show the unit to district assignments in the initial map and are given, although variables $v_k^A$ depend on the budget spent. For \eqref{eq: round2wins}, we know the values of $v_{init,k}^A, v_{init, k}^B$, but variables $z_{ik}^{2}$ depend on the plan that we make for round-2.

\subsubsection{The Votemandering MIP}
\label{subsubsec:vm_mip}
The objective function of \eqref{opt:general_votemander} 
is now represented by the sum of individual district wins in both rounds, i.e., $\hat{s}_i^{1}$ and $\hat{s}_i^{2}$ for every unit $i$. 
As described in \eqref{opt:general_votemander}, a fairness measure constraint $\fairnessfunction(\tilde{\districtplan}, \tilde{\voterballot}) \le \fairnessthreshold $ is implemented, here precisely represented as a  function of the first round variables (updated vote shares $v_{k}^A, v_k^B$) as well as the second-round assignment variables ($z_{ij}^{2}$). 
Using our notation, this constraint refers to the fairness constraint on the votemandered map. Furthermore, the round-2 plan, i.e., $z_{ij}^{2}$ also needs to satisfy the contiguity, population, and/or compactness constraints for making districts. 
We omit these nonpartisan constraints for brevity and 
refer the reader to \cite{swamy2022multiobjective} for implementation details. 
 
Finally, given $v_{init,k}^{A}, v_{init, k}^{B}$, $z_{ij}^{1}$,  $v_{k}^B$,   $\mathcal{B}^A$ and $\fairnessthreshold$, 
a mixed-integer program (MIP) formulation of party $A$'s optimization problem is 
\small
\begin{maxi}[2]
{\Set{b_k^A}_k, \Set{z_{ik}^{2}}_{i,k}}{\sum_{i\in K }\hat{s}_i^{1} & +\sum_{i\in K} \hat{s}_i^{2}}{\label{opt:votemander_mip}}{}
\addConstraint{\textup{constraints } \eqref{eq: binary1}-\eqref{eq: round2wins}}{}
\addConstraint{z_{ik}^{2}, \hat{s}_i^{1}, \hat{s}_i^{2}}{\in \Set{0,1}\quad \forall i,k\in K}
\addConstraint{b_i^A}{\ge 0 \quad \forall i\in K}
\addConstraint{\fairnessfunction \Paren{z_{ik}^{2}, v_{k}^{A}, v_{k}^{B}}}{\le \fairnessthreshold}
\addConstraint{\Set{z_{ik}^{2}} \text{ satisfy nonpartisan constraints.}}{}
\end{maxi}
\normalsize

This concludes the description of the optimization problem \eqref{opt:votemander_mip}.
It is evident that the problem is computationally challenging due to the complex map-making constraints. For most fairness constraints, an exact approach to solving this optimization problem is only feasible for very small-sized grids (on the order of $3 \times 4$). 

\subsection{A Sampling-Based Votemandering Heuristic}
\label{subsec:methods}
The complexity of the optimization problem \eqref{opt:votemander_mip} arises from the interplay between the four votemandering stages. Campaigning decisions depend on the first and third stages (the initial and votemandered maps), whereas the objective depends on the second and fourth stages (the campaigned and target maps). Additionally, \eqref{opt:votemander_mip} accounts for the budget, voter turnout, and feasible map-making constraints, making it difficult to determine the best direction to improve the objective and find an optimal solution.

To address this complexity, the problem is split into two parts leading to an efficient heuristic approach: find a promising target map, then increase round-1 wins while maintaining the apparent fairness of the votemandered map. For a fixed target map defined by $z_{ik}^2$ variables, solving \eqref{opt:votemander_mip} reduces to finding an optimal budget allocation $b_k^A$ while maintaining feasibility (if possible).

A brute-force method of checking all possible new plans $\tilde{\districtplan} \in \districtplans$ is computationally infeasible due to the size of $\districtplans$, i.e., the combinatorial explosion of possible redistricting plans. Instead, sampling is used to reduce the new plan search space from the set $\districtplans$ of all district plans satisfying nonpartisan redistricting constraints to a smaller pool, $\pool \subset \districtplans$. To quickly sample a small but diverse pool $\pool$, we implement the popular recombination Markov chain \citep{deford2021recombination}.

The proposed algorithm considers each candidate plan in $\pool$ according to a priority order, stopping when a pool-optimal plan, $D^*$, is found. Note that an optimal solution within the pool may not be unique, and experiments suggest a large number of pool-optimal plans exist. The number of wins for party $A$ in $D^*$ with the original data, $\election(\tilde{\districtplan}, \voterballot_0)$, is a valid lower bound on the global optimum across all of $\districtplans$. Although recombination sampling may miss optimal new plans, this two-stage heuristic is tractable for standard-sized instances and provides practical solutions that effectively utilize votemandering strategies, showing improvements in the number of seats won.

Let $\pool \subset \districtplans$ be a pool of $\poolsize$ candidate new plans, i.e., 
$\pool \equiv \Set{\districtplan_1, \districtplan_2, \ldots, \districtplan_{\poolsize}}$. 
The choice of new plan $\districtplan_i$ 
combined with the original voter data $\voterballot_0$ 
determines the number of wins in the target map, 
$\election\Paren{\districtplan_i, \voterballot_0}$. 
Hence the best new plan for $A$ is determined by 
finding, for each plan $\districtplan_i \in \pool$, 
the maximum number of round-1 wins for $A$ (via spending budget $\mathcal{B}^A$ )
such that the votemandered map with plan $\districtplan_i$ fools the fairness constraint. 
By decoupling the round-1 and round-2 contributions to 
the objective function of \eqref{opt:votemander_mip}, 
this heuristic efficiently returns the optimal new plan from the pool. 
Algorithm \ref{algo: globalvotemandering} provides a high-level description 
of the heuristic. 


 \begin{algorithm}
    \small
    \caption{Votemandering Heuristic: Select Optimal Plan from a Pool}\label{algo: globalvotemandering}
    \begin{algorithmic}[1]
        \STATE  Input: Pool $\pool=\Set{\districtplan_1,\districtplan_2,\ldots,\districtplan_\poolsize}$ of candidate new plans
        \STATE Sort $\pool$ in decreasing order of $\election\Paren{\districtplan_i, \voterballot_0}$, relabeling from $\districtplan_1$ to $\districtplan_N$
        \STATE $s_{\text{max}}^{1} \gets $ maximum number of round-1 wins for party $A$ by spending campaign budget $\mathcal{B}^A$\label{line:find_s1_max} 
        \STATE $\texttt{best\_plan} \gets \texttt{NULL}$
        \STATE $\texttt{best\_obj} \gets -\infty$
        \FORALL{$\districtplan_i \in \pool$}
            \IF{$s_{\text{max}}^{1} + \election\Paren{\districtplan_i,\voterballot_0} < \texttt{best\_obj}$}\label{line:termination}
                \STATE \textbf{break}
            \ENDIF
            \STATE $\texttt{obj} \gets $ solve \eqref{opt:votemander_mip} with $z_{ik}^{2}$ variables fixed to encode $\districtplan_i$, returning $-\infty$ if infeasible\label{line:fairness_step}
            \IF{\texttt{obj} $>$ \texttt{best\_obj}}
                \STATE $\texttt{best\_plan} \gets \districtplan_i$
                \STATE $\texttt{best\_obj} \gets \texttt{obj}$
            \ENDIF
        \ENDFOR
    \STATE Output: \texttt{best\_plan}
    \end{algorithmic}
\end{algorithm}
\normalsize
\begin{proposition} \label{thm: algo}
Algorithm \ref{algo: globalvotemandering} 
returns a district plan in $\pool$ 
that, when used for the votemandered and target maps, 
maximizes the total number of wins for party $A$ 
across the two election rounds. 
\end{proposition}
\begin{proof}
See Appendix \ref{app:sec_methodologyproof}. 
\end{proof}

The main computational effort in Algorithm \ref{algo: globalvotemandering} 
occurs in Line \ref{line:fairness_step}. 
With the target map fully determined, 
the objective of \eqref{opt:votemander_mip} simplifies 
to maximize the campaigned map (round-1) wins 
over all possible budget allocations $\Set{b_k^A}_k \in K$ 
while maintaining the fairness of the votemandered map. 
We call Line \ref{line:fairness_step} the \emph{fairness step}, 
because the goal is to maximize wins conditioned on 
plan $\districtplan_i$ appearing fair as the votemandered map. Henceforward, we use a specific fairness measure, the efficiency gap (EG), which we formally define in Section \ref{subsec: egdefn}.
We next expound on the simplified version of \eqref{opt:votemander_mip} 
solved with the fairness step specific to EG. 

\subsubsection{Additional Notation.} Let $\mathcal{I} = \{I_1, ..I_{n}\}$  be the set of districts $I_i$ in the original plan (round-1) such that each $I_i$ is a set of units from $K$. Sets $I_i$ satisfy $I_i \cap I_j = \emptyset$ as no unit can belong to two districts in any round.  
Let ($V_{init,I}^{A}$, $V_{init,I}^{B}$) and  ($V_{I}^{A}$, $V_{I}^{B}$) denote pre-campaigning and post-campaigning votes, 
respectively, in district $I$. 
Similarly, $\mathcal{J}= \{J_1, ..J_{n}\}$ is the set of districts in the new plan (round-2), with ($V_{init,J}^{A}$, $V_{init,J}^{B}$) and  ($V_{J}^{A}$, $V_{J}^{B}$) denoting pre and post-campaigning votes in district $J \in \mathcal{J}$. Let $\hat{x}_{I_i}$, indexed using sets  $I_i \in \mathcal{I}$ and $\hat{y}_{J_j}$, indexed using sets  $J_j \in \mathcal{J}$ be the indicator variables denoting the wins in the campaigned map (round-1) and the votemandered map (round-2), respectively. Note that we resort to the set notation ($\hat{x}_{I}$ using $I, J$) unlike that of the original optimization problem (i.e., $\hat{s}_i^{1}$), as we now have assignments of both initial and target maps, allowing lesser notation. 

\subsubsection{Incorporating EG into the Fairness Step.} \label{subsec: egdefn}
Using the definition for EG, the difference between wasted votes for each district $I \in \mathcal{I}$ (denoted henceforth by $\mathcal{W}(B-A)$) is given by
\begin{equation}
   \mathcal{W} (B-A)(I)= 
\begin{cases}
    \frac{3V_{I}^{B}  - V_{I}^{A}}{2},              & \text{if }   V_{I}^{A} >  V_{I}^{B} \text{ ($A$ wins})\\
    \frac{V_{I}^{B} - 3V_{I}^{A} }{2},              & \text{if }  V_{I}^{A} < V_{I}^{B} \text{ ($B$ wins})\\
\end{cases}
\end{equation}
 Using this definition, we further write the constraint of EG less than a particular constant, say $8\%$ \citep{stephanopoulos2015partisan}.  
\begin{equation}\text{efficiency gap of the state } = \Abs{\sum_{I\in \mathcal{I}} \mathcal{W} (B-A)(I) / \left( \sum_{I\in \mathcal{I}}  V_I^A+V_I^B \right)}  \leq 0.08  \label{defn: eg}\end{equation}
Next, we describe the MIP we use to ensure the fairness of the proposed map in round-2, i.e., the votemandered map. Letting $\tau_{J}, \ \forall J \in \mathcal{J}$ denote the difference between wasted votes in $J$'s district, i.e., $\mathcal{W} (B-A)(I)$, we can write: 
\small
{\allowdisplaybreaks
\begin{align}
      \max \  \sum_{I\in \mathcal{I}} \hat{x}_{I} \notag \\
        s.t. \ \ 
        & \text{Constraints \eqref{eq: budgetcap2}, \eqref{eq: budgetconstraint}} \notag\\
        & V_{I}^{A}  = \sum_{k \in I} \alpha v_{init,k}^{A} + b_k^A  & \forall I \in \mathcal{I} \notag\\
         &   V_{J}^{A}   =  \sum_{k \in J} \alpha v_{init,k}^{A} + b_k^A & \forall J \in \mathcal{J}\notag\\
       &  1-M(1-\hat{x}_{I}) \leq  (V_I^{A}- V_I^{B}) \leq M \hat{x}_{I} & \forall I \in \mathcal{I}
       \label{eq: seat1MIP} \\
       & 1-M(1-\hat{y}_{J}) \leq  (V_J^{A}- V_J^{B}) \leq M \hat{y}_{J} & \forall J \in \mathcal{J}
        \label{eq: seat2MIP} \\
       &  0 \leq - \tau_{J}+ \left( \frac{3 V_J^B -  V_J^A}{2} \right) \leq M (1-\hat{y}_{J}) & \forall J \in \mathcal{J} \label{eq: tau1} \\
       &  0 \leq  \tau_{J} - \left( \frac{ V_J^B - 3V_J^A}{2} \right) \leq M \hat{y}_{J} & \forall J \in \mathcal{J} \label{eq: tau2} \\
       & -0.08  \leq \sum_{J} \tau_{J}/\left( \sum_{J\in \mathcal{J}}  V_J^A+V_J^B \right) \leq 0.08 \label{eq: egMIP} \\
      & \hat{x}_{I}, \hat{y}_{J} \in \{0,1\},  \ \ \ b_k^A, \tau_{J} \geq 0 & \forall k \in K, J \in \mathcal{J} 
    \label{prog: step2}
\end{align}}
\normalsize



Algorithm \ref{algo: globalvotemandering} finds an optimal solution within a pool of target maps,  but one may question about the probability that a globally optimal solution exists within a pool generated by running a recombination chain for $\poolsize$ steps. However, given the hardness of finding a globally optimal solution, it is unlikely that a bound on this probability can be determined. In practice, the algorithm is computationally efficient as shown in Section \ref{sec:results} (Theorem \ref{thm: polytime}), and the returns diminish as the size of the pool increases. It is important to note that the primary goal of this paper is to establish the mechanism of votemandering and study its dependence on various crucial factors that affect redistricting, as opposed to finding the optimal votemandering strategies.

As the algorithm works given any inputs of the initial map and campaign budget, it establishes a framework that can be used to test the robustness of any district plan or pool of maps against votemandering. This framework is used in later sections to compare the effects of various state characteristics and external redistricting conditions on votemandering. An ideal map would have a lower objective when tested against a standard pool of target maps. The higher the budget required to votemander, the better the robustness.

\section{Results and Analysis}
\label{sec:results}
This section presents the efficacy and efficiency of votemandering under various conditions. Using EG as our fairness measure, we begin by examining the impact of campaigning on votemandering objective and fairness in Section \ref{sec: vmability}. We show that under certain general conditions, votemandering can always occur. In Section \ref{sec: vmefficiency}, we establish the polynomial-time complexity of Algorithm \ref{algo: globalvotemandering}. Finally, in Section \ref{sec: experimentalresults}, we experimentally analyze the dependence of various factors on votemandering, such as the budget of Party $A$ and Party $B$, compactness, voter turnout, and the concentration index Moran's I.
\subsection{Sufficient Conditions for Votemandering} \label{sec: vmability}
As a build-up to this question, we analyze the strategy space of party $A$: it can add new votes via campaigning in round-1; effectively gerrymander to shift votes from a winning (W) district to a losing (L) district or vice versa. We discuss their key implications on fairness in Lemma \ref{lem: tabvotes}.
\begin{restatable}{lemma}{lem}
     \label{lem: tabvotes}
    The actions of campaigning and vote shifts have an impact on the difference between wasted votes, i.e., $\mathcal{W}(B-A)$, as given in Table \ref{tab: deltaWcases}.
\begin{table}
\small
\centering
\caption{Change in the difference between wasted votes (EG) as new votes are added/shifted}
\label{tab: deltaWcases}
\begin{tabular}{|c|c|c|}
\hline
  & Action                                             & Impact on $\mathcal{W}(B-A)$ \\ \hline
1 & Wasting an additional vote on a losing district    & $-3/2$         \\ \hline
2 & Wasting an additional vote on a winning district   & $-1/2$         \\ \hline
3 & Winning a district $I$ through campaigning              & $(3V^A_I+V^B_I)/2$   \\ \hline
4 & Shift $\shiftedvotes$ votes from a winning to a losing district & $-\shiftedvotes$                   \\ \hline
5 & Shift $\shiftedvotes$ votes from a losing to a winning district & $\shiftedvotes$                    \\ \hline
\end{tabular}
\normalsize

\end{table}
\end{restatable} 
\noindent
\begin{proof} See Appendix \ref{app:sec_results-1}.
\end{proof}

Next, we discuss the total impact on the change in wasted votes $\mathcal{W}(B-A)$, and thereby, the efficiency gap, as a new district plan gets drawn over the same vote data. As the total number of votes does not change in this case, this change can be tracked just through a reshuffle of units into winning and losing districts. For district assignment $\mathcal{I}$ in round-1, the difference between the wasted votes $  \mathcal{W}(B-A)_{\mathcal{I}}$ is expressed as:
\begin{align}
    \mathcal{W}(B-A)_{\mathcal{I}} & = \alpha \left( \sum_{I \in \mathcal{I}(W)} \frac{3V_{init,I}^{B}-V_{init,I}^{A}}{2} \right) + \alpha\left(\sum_{i\in \mathcal{I}(L)}\frac{V_{init,I}^{B} - 3V_{init,I}^{A}}{2} \right)  \notag \\
    & = \alpha \left( \sum_{I \in \mathcal{I}} \frac{V_{init,I}^{B}-V_{init,I}^{A}}{2} \right) + \alpha\left(\sum_{i\in \mathcal{I}(W)}V_{init,I}^{B} - \sum_{I\in \mathcal{I}(L)}V_{init,I}^{A} \right)
\end{align}
where $\mathcal{I}(W)$ and $\mathcal{I}(L)$ are the sets of  winning and losing districts, respectively. Then, after reshuffling to district assignment $\mathcal{J}$ in round-2, the change in $\mathcal{W}$ (defined by $\Delta \mathcal{W} (B-A)_{\mathcal{I} \rightarrow \mathcal{J}}$) and the final $\mathcal{W}$ is given as: 
\begin{align}
    & \Delta \mathcal{W} (B-A)_{\mathcal{I} \rightarrow \mathcal{J}} =   \alpha\left( \sum_{j\in \mathcal{J}(W)}V_{init,j}^{B} - \sum_{j\in \mathcal{J}(L)}V_{init,j}^{A} \right) -\alpha \left( \sum_{i\in \mathcal{I}(W)}V_{init,i}^{B} - \sum_{i\in \mathcal{I}(L)}V_{init,i}^{A} \right) \notag \\
    & \mathcal{W}(B-A)_{\mathcal{J}} = \mathcal{W}(B-A)_{\mathcal{I}} +   \Delta \mathcal{W} (B-A)_{\mathcal{I} \rightarrow \mathcal{J}} + \text{[Any wasted votes through campaign]} \label{eq: ItoJEG}
\end{align}
To conclude,  Table \ref{tab: deltaWcases} and Eq. \eqref{eq: ItoJEG} show that a campaign budget can be allotted (thereby updating $V_{init,I}^A, V_{init,I}^B$ to $V_I^A, V_I^B$) to achieve fairness of a target map, given that the allocation also satisfies the budget and voter-turnout constraints. Thus, votemandering can potentially include at least two (interdependent) ways: (1) Fixing plan $\mathcal{J}$ and allotting appropriate budget to satisfy the fairness bound, whilst benefiting from campaigning in round-1, and (2) Designing a target map with $\mathcal{J}$ that leads to a higher number of wins in round-2, maintaining fairness.  To measure the efficacy of votemandering, we define \textit{votemandering bonus} i.e., $\Delta$ which measures the gain in the number of wins after enabling votemandering. For a target plan $\tilde{\districtplan}$ and campaigning resulting with $\tilde{\voterballot}$,
\begin{equation}
\text{Votemandering bonus } \Delta = \election(\districtplan_0,\tilde{\districtplan}) + \election(\tilde{\districtplan},\voterballot_0) - 2 \election\Paren{\districtplan_0,\voterballot_0}
\end{equation}

Using this definition, a positive votemandering bonus would indicate that we have successfully votemandered. We now characterize the sufficient conditions for successful votemandering using the second way of improving the objective. Intuitively, we need a target map better than the initial map, and a baseline voter turnout to allow GOTV efforts to take place.
\begin{restatable}{theorem}{thmone} \label{thm: existence}
For any vote-share distribution and a corresponding fair initial map with assignment $\mathcal{I}$, the existence of strategies leading to a positive votemandering bonus is guaranteed if 
\begin{enumerate}
    \item A feasible, contiguous map with assignment $\mathcal{J}$ exists with a higher number of wins than the initial map.
    \item The voter turnout $\alpha$ satisfies:
    \[\alpha  \leq 1- \left( \frac{2 \Delta \mathcal{W} (B-A)_{\mathcal{I} \rightarrow \mathcal{J}}}{\sum_{J\in \mathcal{J}(W)} \sum_{j \in J} v_j^A} \right)\]
    where $\Delta \mathcal{W} (B-A)_{\mathcal{I} \rightarrow \mathcal{J}}$ is the change in the difference between wasted votes from assignment $I$ to $J$, and $v_j^A$ are party $A$'s votes in unit $j$.
\end{enumerate}
\end{restatable}
\noindent
\begin{proof}
We prove this by achieving fairness of the map with assignment $\mathcal{J}$ using a strategic allocation of budget, thereby establishing the existence of strategies. We primarily satisfy Eq. \eqref{eq: ItoJEG}. We continue with the proof details in Appendix \ref{app:sec_results-2}.
\end{proof}

In practice, it is generally much easier to votemander (as we demonstrate in Section \ref{sec: experimentalresults} and through case studies in Section \ref{sec:case_study}), except under highly specific conditions such as near $100\%$ voter turnout and nearly all voters favoring a single party. The first way of votemandering, as discussed above, also allows for a positive bonus to be achieved through an increase in wins in the first round, as long as the voter turnout allows for such campaigning to occur in a fair way. Its campaigning effects on fairness, as translated from the additional number of wins, can be dissolved in the votemandered map through reorganization of the campaigned map. While this first way is easier to see in practice, its dependence on the specificity of $\mathcal{J}$ makes it difficult to establish sufficient conditions for votemandering as it demands map making, given an initial assignment $\mathcal{I}$. We explore specific votemandering strategies using both ways in detail in Section \ref{sec:local_vm}.

\subsection{Efficiency of Votemandering Heuristic} \label{sec: vmefficiency}
We now establish the polynomial time complexity of the votemandering heuristic. Recall that it takes a pool of maps $\pool$ as an input, and outputs the target map in $\pool$  maximizing the votemandering objective with respect to the given initial map, i.e., it finds the target map with the maximum votemandering bonus. Proposition \ref{thm: algo} confirms the correctness of Algorithm \ref{algo: globalvotemandering} in its convergence to the optimal target map. Theorem \ref{thm: polytime} now shows that this may be achieved efficiently. 
\begin{restatable}{theorem}{thmtwo} \label{thm: polytime}
Let $\pool$ be a pool of $\poolsize$ candidate target district plans such that $\pool$ is a subset of feasible, but not necessarily fair, $n$-districts plans. A plan in $\pool$ that maximizes the votemandering bonus may be found in poly($\poolsize$, $n$) time.
\end{restatable} 

\begin{proof}
Following Algorithm \ref{algo: globalvotemandering}, see that the only complicated part is the fairness step \eqref{line:fairness_step} in the MIP in Section \ref{subsec: egdefn}. We show that each target map can be checked in polynomial time, enabling us to move through the pool quickly until convergence. We sketch the proof here and provide details in Appendix \ref{app:sec_results-3}.
\begin{enumerate}
    \item For each district in round-1, we decompose its space into \emph{pieces} that each belongs to a district in round-2. For each such piece, we define its capacity = min(its voter-turnout capacity, budget needed to win the round-2 district it is part of (only if part of a losing district)).
    \item Given party $B$'s investment and the original vote shares of $A$ and $B$, we next find the win/lose ($W/L$) status of the districts in the votemandered map and compose a linear program 
    to find the maximum wins in the campaigned map while constraining on the status of the districts, implemented by the variables for investment in the pieces.
    \item If the $(W,L)$ constraint for district $k$ (in the votemandered map) is tight in the optimal solution, do: i) mark $k$'s status as a win and update the fairness of the votemandered map, ii) add a constraint for allocating the budget needed to win $k$. We then solve the updated linear program, and if the objective increases, we repeat all the steps with updated $W/L$ status and constraints until the objective stops increasing. This converges in polynomial time since there is a predefined number of districts with an $L$ status, bounded above by $n$.
\end{enumerate} 
These three steps suffice to prove Theorem \ref{thm: polytime}.
\end{proof}
\subsection{Analysis of factors impacting votemandering} \label{sec: experimentalresults}

Although the conditions for votemandering may be easily satisfied in practice, the required budget to ensure a positive bonus can vary. The efficacy of votemandering depends on several factors, including the initial vote share distribution across the state, the initial district assignment plan, and the available campaign budget for parties $A$ and $B$. Furthermore, it is influenced by various externally imposed constraints on the redistricting process, such as EG, compactness, the number of majority-minority districts, proportionality, etc. As a result, we opt for a randomized approach, i.e., Algorithm \ref{algo: globalvotemandering}, for a pool of maps to examine the dependence on these factors and, in turn, demonstrate the efficacy of votemandering under various conditions.
The pool of plans is randomly generated using recombination and contains plans that all satisfy the externally imposed constraints. We fix a randomly generated vote share distribution across a grid with $20 \times 20$ units such that each unit $i$ has a population $p_i$ uniformly chosen between $350-400$ and vote shares $(v_{init,i}^A/p_i, v_{init,i}^B/p_i)$ between $20-80\%$ for each party. Each feasible map from the pool provides a unit-to-district assignment, mapping the 400 units to 10 districts, with district populations allowed to deviate $1\%$ from the average district population.

\subsubsection{Impact of increasing budget.} \label{subsec: budget}

\begin{figure}
    \centering
    \begin{subfigure}[b]{0.45\textwidth}
         \centering
         \includegraphics[width=5 cm]{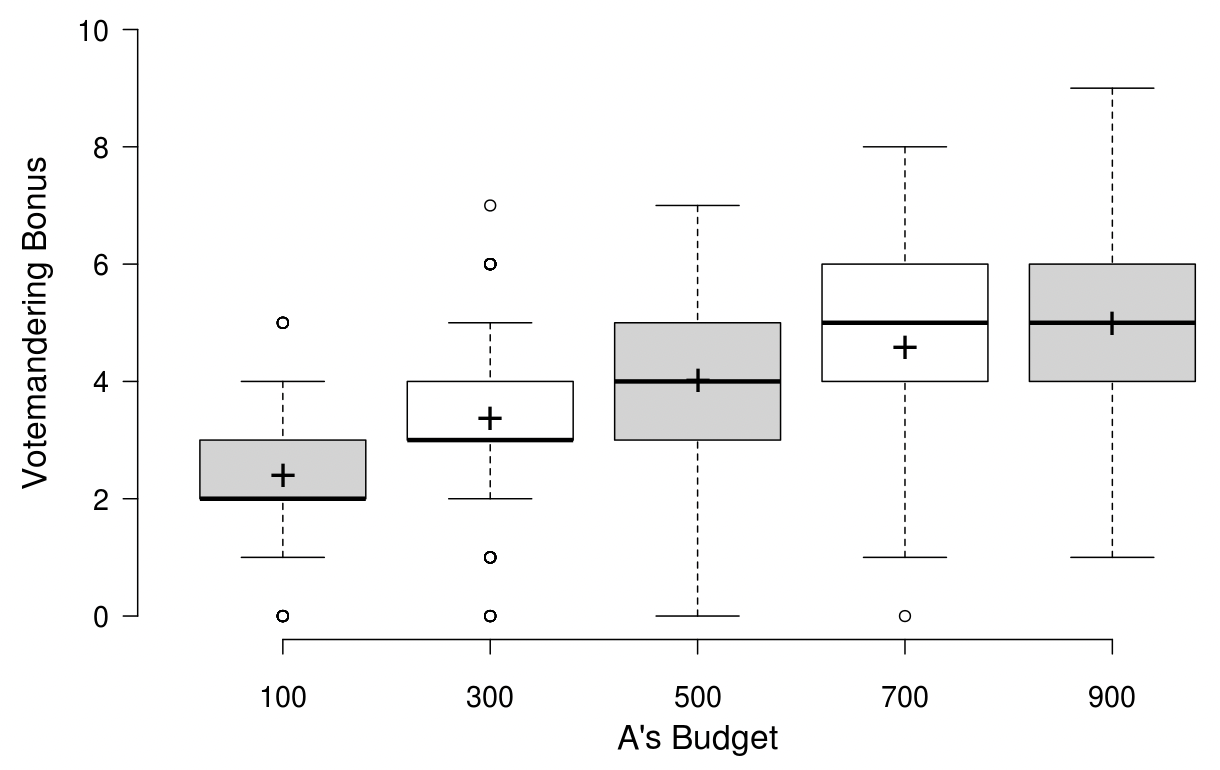}
         \caption{Increasing budget of party A}
         \label{fig: seatsvsbudgetA}
     \end{subfigure}
     \hfill
     \begin{subfigure}[b]{0.45\textwidth}
         \centering
         \includegraphics[width=5 cm]{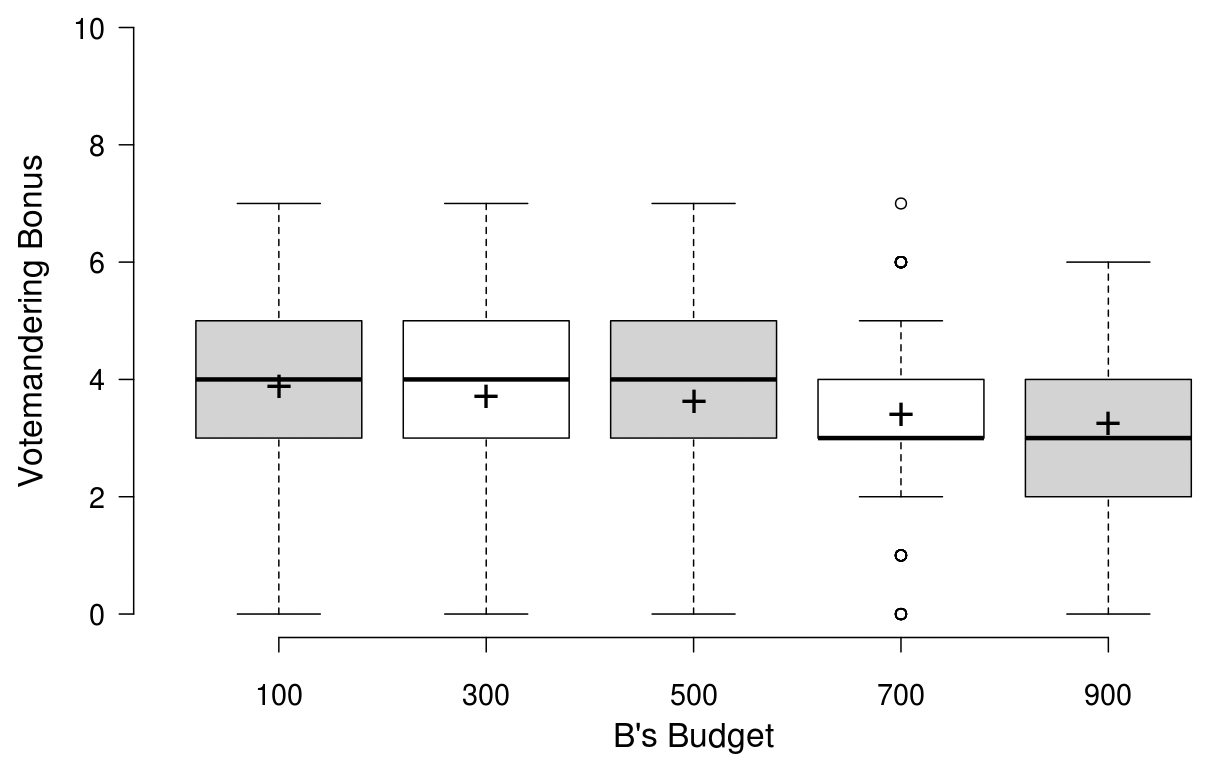}
         \caption{Increasing budget of party B}
         \label{fig: seatsvsbudgetB}
     \end{subfigure}
    \caption{Party A's votemandering bonus, with increase in both budgets}
    \label{fig: seatsvsbudget}
\end{figure} 

We present our first key result through Figure \ref{fig: seatsvsbudget}, which tracks how increasing the campaign budget strengthens the ability to votemander. Recall that the budget equals the number of votes that can be influenced above the baseline voter turnout, with an upper bound given by total party affiliation shares. Figure \ref{fig: seatsvsbudget} plots Party $A$'s votemandering bonus as the parties increase their budget uniformly. In experiments for Figure \ref{fig: seatsvsbudgetA}, Party $B$'s budget $\mathcal{B}^B$ is fixed at 400, and Party $A$'s budget $\mathcal{B}^A$ is varied, whereas for Figure \ref{fig: seatsvsbudgetB}, $\mathcal{B}^B$ is varied, and $\mathcal{B}^A$ is fixed at 400. In both experiments, we plot the majority party $A$'s votemandering bonus coming from its strategic investment. Recall that we do not assume any campaigning strategies from Party $B$. Given any budget allocation of $B$, if $A$ has access to the allocation information, then the algorithm finds the best strategies for $A$. Here, we let $B$ invest most straightforwardly, making its budget investment proportional to each unit's population, allowing fractional investments.

Figure \ref{fig: seatsvsbudgetA} shows a steady increase in the bonus through the means and medians shifting upwards with the increase in $\mathcal{B}^A$. The bonus for $\mathcal{B}^A=100$ indicates the objective that can be attained by accessing Party $B$'s budget investment information while $A$ puts in a little campaigning effort itself. Note that the increase in bonus in Figure \ref{fig: seatsvsbudgetA} is not linear. Improving the allowed budget has diminishing returns in the form of objectives. This is expected since the objective, and therefore the bonus, is capped by the total number of seats available in both rounds.

Most interestingly, the bonus has a counterintuitive relation with increasing $\mathcal{B}^B$ as shown in Figure \ref{fig: seatsvsbudgetB}. As opposed to a clear steady increase in \ref{fig: seatsvsbudgetA}, increasing $\mathcal{B}^B$ does not ensure a steady decrease in the bonus. Although increasing $\mathcal{B}^B$ may impact Party $A$'s chances of winning in the first round, $A$ may use this information to create better districts in the second round. This is achieved by letting Party $B$ win some districts in the votemandered map, only to lose those in the target map as the campaign effects diminish. Because of this trade-off, the decreasing trend is not obvious: Party $A$'s votemandering bonus remains largely unaffected until $\mathcal{B}^B$ reaches 500.  

\subsubsection{Impact of increasing compactness.}\label{subsec: compactness}

\begin{figure}
    \centering
    \begin{subfigure}[b]{0.45\textwidth}
         \centering
    \includegraphics[width=5cm]{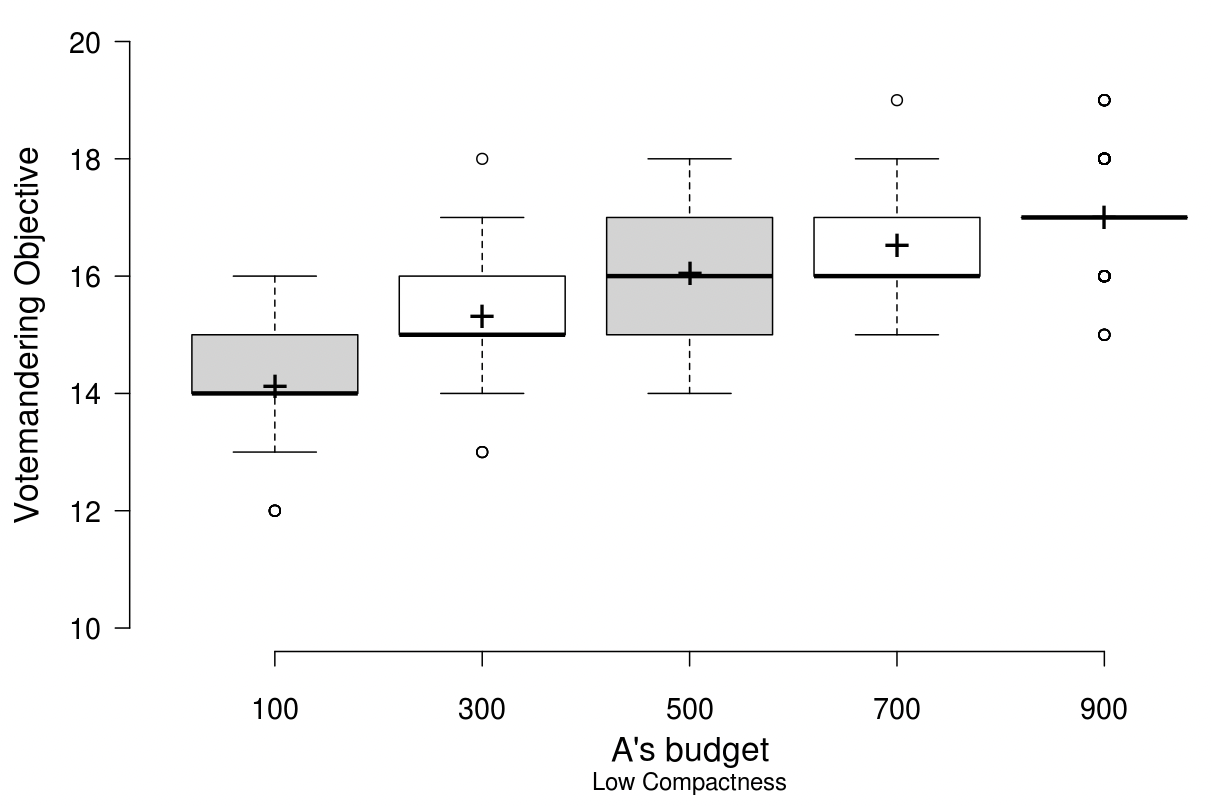}
         \caption{Lower compactness}
         \label{fig: lowc}
     \end{subfigure}
     \hfill
     \begin{subfigure}[b]{0.45\textwidth}
         \centering
    \includegraphics[width=5cm]{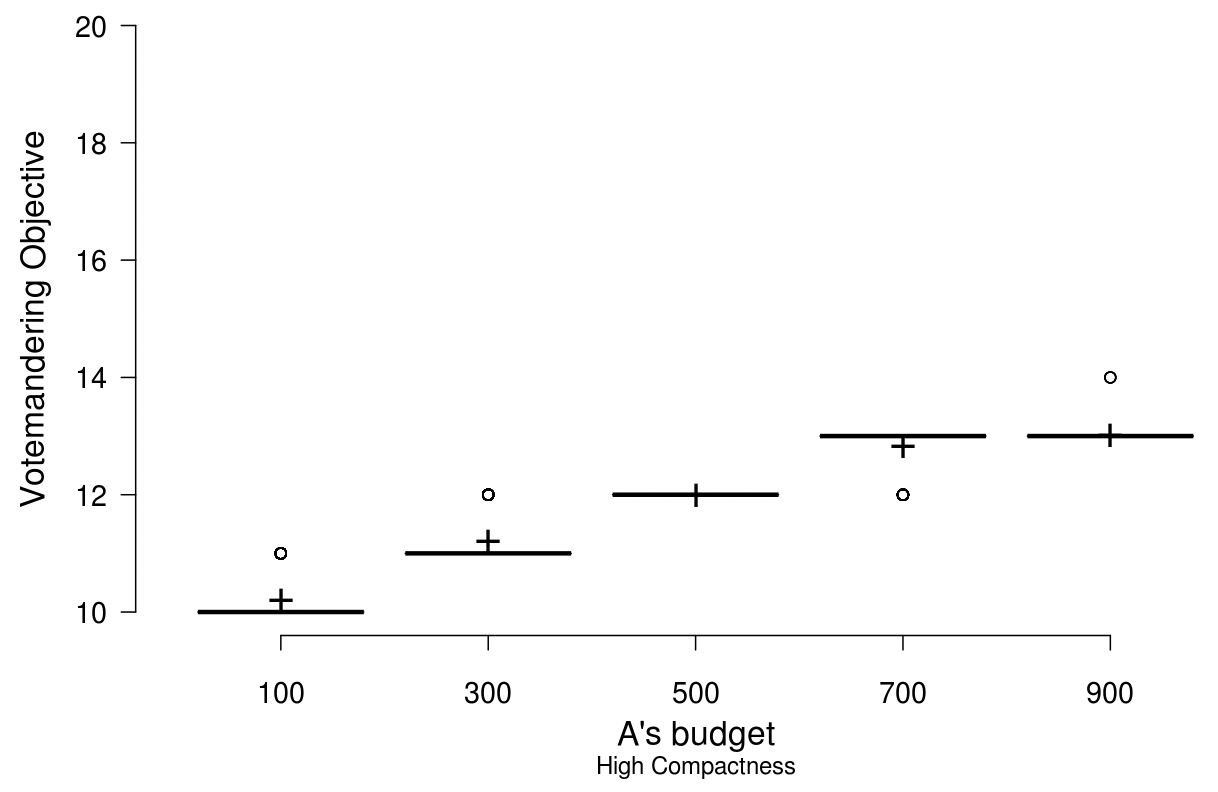}
         \caption{Higher compactness}
         \label{fig: highc}
     \end{subfigure}
   \caption{Party A's votemandering objective with lower and higher compactness bounds}
    \label{fig: seatsvscompactness}
\end{figure} 

The metric of compactness is generally not perceived as a fairness notion and is usually imposed to be in line with the (older) belief of creating districts that minimize the physical distance between units within a district. The salamander-shaped district in the first gerrymander suggests rejecting freehand-shaped districts and asking for compactness as a proxy for partisan neutrality \citep{polsby1991third}. However, compactness is often deemed orthogonal to fairness measures \citep{gurnee2021fairmandering}. Contrary to this belief, we demonstrate that imposing tighter compactness bounds limits the ability of votemandering, leading to better (robust) maps in general. We achieve this by comparing votemandering objectives on two separate pools of maps, generated through recombination: one with looser and one with tighter compactness constraints. For ease of handling, compactness is expressed through the number of cut edges, as done in the foundational work on recombination \citep{deford2021recombination}. The number of cut edges is defined as the number of edges in a state's unit adjacency graph with endpoints belonging to different districts. For instance, a $20\times 20$ grid graph with each district composed of two adjacent columns—making 10 districts overall—will have $9\times20=180$ cut edges. For showing the effects of compactness, the first pool has the maximum number of cut edges equal to $2 \times 180 = 360$, and the second pool has a bound of $0.75\times 180 = 135$ cut edges.

The results are given in Figure \ref{fig: seatsvscompactness}, which show that more compact plans lead to a lower number of seats achieved through votemandering. Recall Lemma \ref{lem: tabvotes}, which shows that investing in a losing district is 3 times more beneficial in achieving fairness, while winning through campaigning marks the investment as a winning-district specific. Then, campaigning in targeted units is usually followed by their re-assignments to losing districts, as a votemandering strategy to achieve fairness benefits. Compactness limits this scope of targeting units for the campaign and subsequent reassigning, by disallowing arbitrary shapes. We elaborate more on the intuition behind this phenomenon, as we discuss the local votemandering strategies in Section \ref{sec: localcharacteristics}. Note that Figure \ref{fig: seatsvscompactness} compares votemandering objectives, as opposed to bonuses shown in Figure \ref{fig: seatsvsbudget}, as here two different pools are used, which also significantly affects the distribution of wins in the initial maps, and thereby the bonuses.

\subsubsection{Impact of voter turnout.} \label{subsec: turnout}

\begin{figure}
    \centering
    \begin{subfigure}[b]{0.45\textwidth}
         \centering
    \includegraphics[width=5cm]{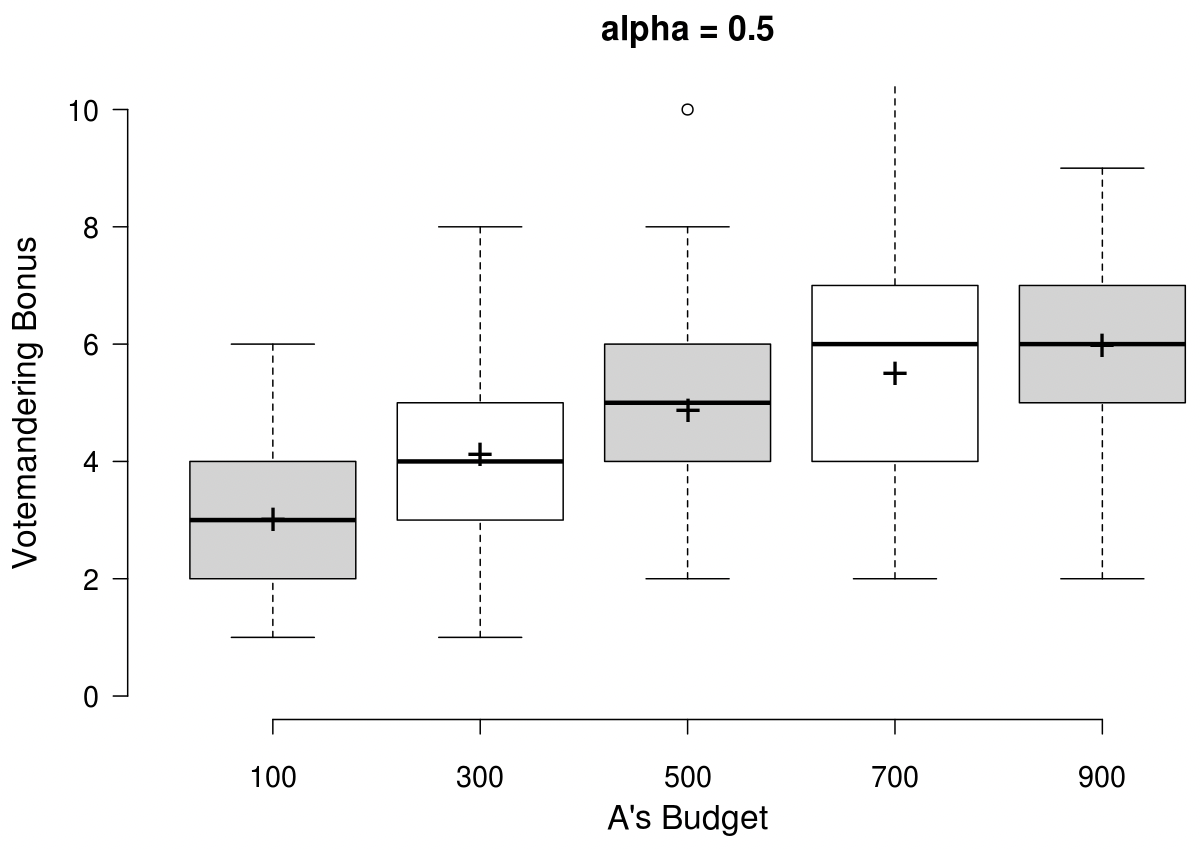}
         \caption{$\alpha$ = 0.5}
         \label{fig: seatsvsalpha10.5}
     \end{subfigure}
     \hfill
     \begin{subfigure}[b]{0.45\textwidth}
         \centering
     \includegraphics[width=5cm]{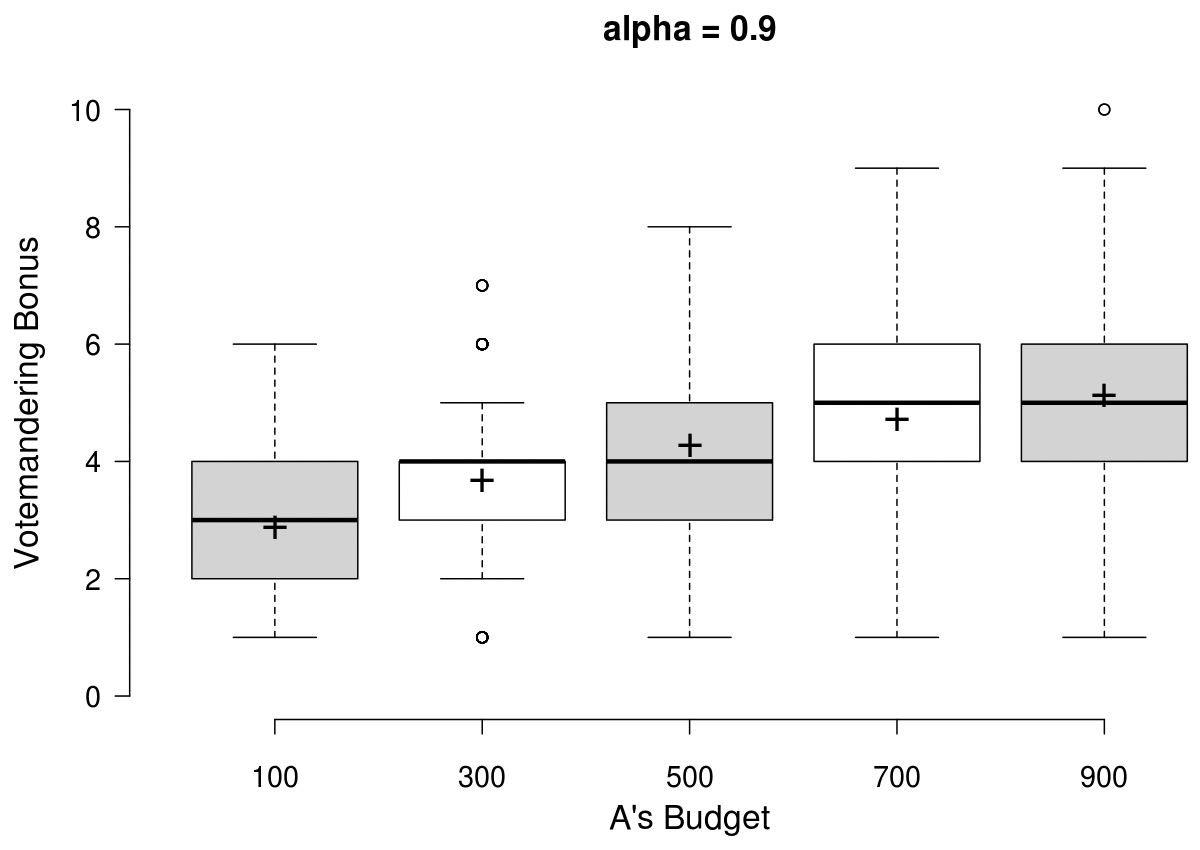}
         \caption{$\alpha$ = 0.9}
         \label{fig: seatsvsalpha10.9}
     \end{subfigure}
   \caption{Increase in A's votemandering bonus with voter turnout}
    \label{fig: seatsvsalpha1}
\end{figure} 

Intuitively, the ability to votemander is a function of how efficiently and thus also, how disproportionately we can allocate budget across the units. The parameter $\alpha$ captures the natural voter turnout and the selective campaigning by a party strategically adds more party votes, bounded above by the natural vote share of that party in every unit. Hence, it is straightforward to see that an increase in $\alpha$ would provide less flexibility to votemander, and will more accurately represent the true vote share, putting less weight on the campaigned votes. We rigorously show this in effect in Figure \ref{fig: seatsvsalpha1} where we plot the objective with respect to increasing $\alpha$. 
With both the means and medians shifting downwards with increasing $\alpha$, this demonstrates that a higher voter turnout supports a better representation of social choice through not only higher volume and election credibility, but also through disallowing political parties to votemander.

\subsubsection{Impact of spatial autocorrelation of voters (Moran's I)} 
\begin{figure}
    \centering
    \includegraphics[width=7cm]{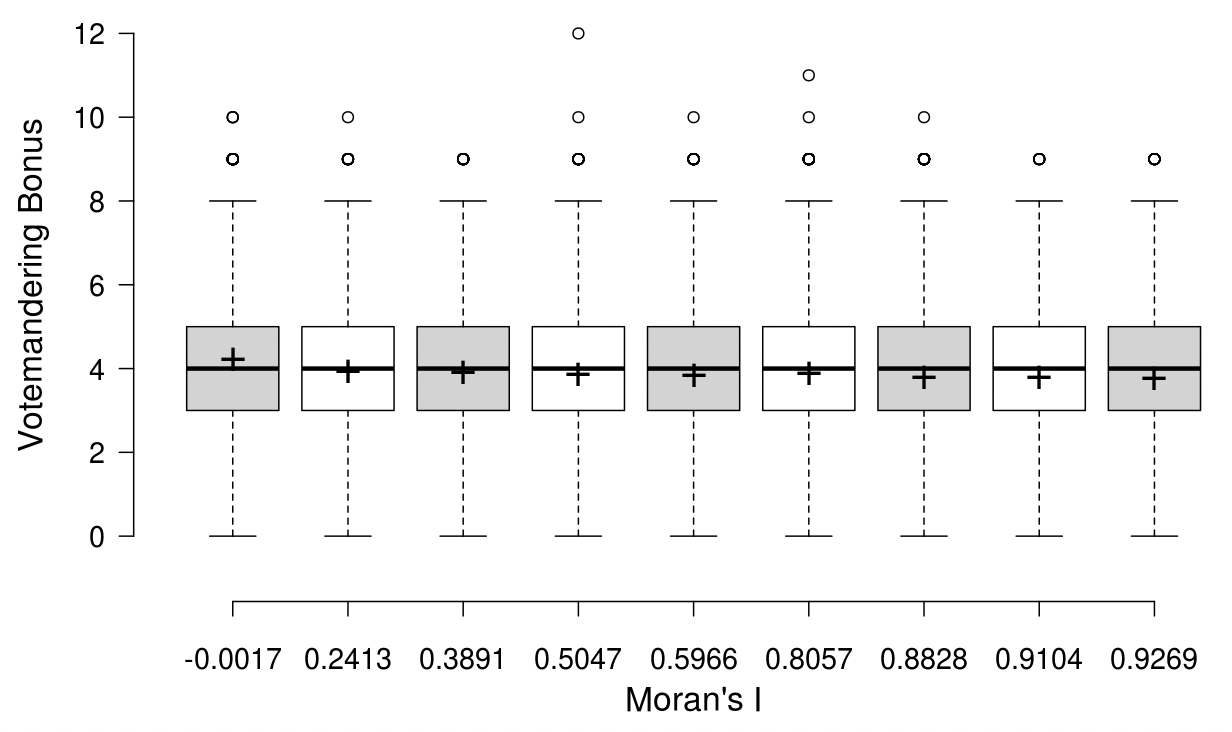}
    \caption{Votemandering bonus with increasing Moran's I}
    \label{fig: moranI}
\end{figure}
One may question if the current experimental setting of a geographically uniformly spread voter population is reasonable, as we often see clusters of societies divided across political and geographic lines. In this experiment, we show that clustering of this data does not have a very significant effect on the votemandering bonus. To demonstrate this, we use a popular spacial auto-correlation metric called \emph{Moran's I} \citep{duchin2021political}. This is a measure of the overall clustering of the spatial data. For  $(v_1,..v_{|K|})$ as the vector of vote shares, $\bar v$ as the average vote share, $y_{ij}$ as a binary variable indicating adjacency of units $i,j$ and $Y = \sum_{i,j}^{|K|} y_{ij}$ as the number of total adjacencies, Moran's I is defined as 
\begin{equation}
I = \frac {|K|} {Y} \frac {\sum_{i=1}^{|K|} \sum_{j=1}^{|K|} y_{ij} (v_i-\bar v) (v_j-\bar v)} {\sum_{i=1}^K (v_i-\bar v)^2}    
\end{equation}
Moran's I is usually used to measure the segregation of geospatial data and it varies between [$-1,1$] with -1 indicating anti-segregation, 0 with no segregation, and 1 with extreme segregation.  The randomly generated voter patterns used in Section \ref{subsec: budget}, \ref{subsec: compactness}, and \ref{subsec: turnout} produce Moran's I values in the range $(-0.01, 0.01)$. For the current set of experiments, we generate voter distributions with varying Moran's I values, group those into nine bins, and plot the votemandering bonuses in Figure \ref{fig: moranI}. 
Note that the overall party vote-shares are kept around the same value while increasing the clustering of data. 
The insensitivity of the votemandering bonus to variations in Moran's I 
supports the experimental design decision of keeping Moran's I near zero 
when varying the other redistricting factors 
in Sections \ref{subsec: budget}, \ref{subsec: compactness}, and 
\ref{subsec: turnout}. 

In summary,  campaign budgets affect the votemandering bonus with diminishing marginal improvements,  high voter turnout and stricter compactness bounds curtail votemandering, 
and voter clustering patterns have little effect on votemandering. Moreover, these factors together help establish the robustness of votemandering objectives to smaller uncertainties in voter-data. For the strategist party, assuming base voter inclinations at the lower end of their confidence intervals is sufficient to devise votemandering strategies. While access to $B$'s budget allocation information is essential, its slow rate of effect helps if the information is stochastic. In further continuation with the discussion in Section  \ref{subsec: budget}, it is worth emphasizing that the fusion of gerrymandering and strategic campaigning separates the effects of pure campaigning strategies from votemandering strategies, enabling a more robust dependence of budgets and voter inclinations on votemandering.

\section{Local Votemandering} 
\label{sec:local_vm}
The votemandering methods in Section \ref{sec:methodology} allow district lines in the target map to deviate significantly from those in the initial map. However, in practice, considerations are made for maintaining the original community boundaries (e.g., retaining  majority/minority districts) as well as structural boundaries (e.g., disallowing county splits) while drawing new district plans on the ground.
Some states also demand that 
redistricting plans remain close to the existing plan. 
For example, 
Nebraska requires the new plan to 
``preserve the cores of prior districts'' 
\citep{nebraska}, and 
the Wisconsin Supreme Court issued 
a similar ``least-change'' order 
for the 2020 cycle \citep{wiscsc}.

Inspired by these requirements, Section \ref{sec: localmethods} introduces a \emph{local votemandering heuristic} that conducts a local search within smaller map sections. This approach generates new plans that satisfy local proximity requirements while maintaining global fairness and budget constraints. The heuristic aims to maximize the votemandering objective, offering insights into crucial votemandering strategies as detailed in Section \ref{sec: localstrategies}. Although local votemandering provides a lower bound on global votemandering performance, it is a faster method for achieving a positive votemandering bonus. Additionally, it is parallelizable and thus scales better with increasing state sizes. Section \ref{sec: localcharacteristics} further explores the relationship between local and global votemandering and emphasizes the significance of key strategies in understanding this connection.

\subsection{Local Votemandering Methods} \label{sec: localmethods}

The local votemandering heuristic generates new target plans by applying small changes, or \textit{local boundary perturbations}, to existing plans between pairs of neighboring districts. These perturbations involve exchanging units between two districts while satisfying external proximity and redistricting requirements. Unlike the top-down approach in Section \ref{sec:methodology} (hereafter referred to as the global votemandering heuristic), this bottom-up approach strategically employs local perturbations to increment the votemandering bonus by one. A district adjacency graph is formed, with edges that have a potential positive bonus—containing at least one district with initial status $L$—assigned weights representing (budget, fairness) costs associated with the perturbations.

The ultimate objective is to maximize the overall bonus by finding a maximum-sized matching that adheres to budget and fairness constraints. The resulting target plan is created by applying perturbations corresponding to the maximum matching. Matchings are utilized because they enable mutually exclusive perturbations between district pairs.

The new target plan remains similar to the initial plan in terms of unit-to-district assignments, with alterations only involving districts that expect positive bonuses. Since matchings are used, the difference is quantified as, at most, $\frac{n}{2}$ independent recombination steps away from the initial map, where $n$ is the number of districts. The heuristic considers perturbations only between district pairs to avoid the complexity of perturbations within an arbitrary number of districts, to create plans closer to the initial plan, and to maintain the highest improvement ratio over the number of perturbed districts. Although it is technically feasible to generalize the heuristic to consider perturbing three or more districts simultaneously, the current approach focuses on pairs.

Generating bonuses within submaps—regions restricted to two adjacent districts—exposes the functionality of votemandering's key strategies, as illustrated in Figure \ref{fig: threestrategies}. Recall the two general votemandering ways discussed in Section \ref{sec: vmability}.  Stemming parallel to those, these strategies encompass securing a first-round win in the campaigned map and either of the second-round wins in the target map: with or without winning in the votemandered map. Depending on the initial map, any of the key strategies could be more cost-effective in terms of budget and/or fairness (or be infeasible). Given cost information about edges, we first formulate an optimization framework that produces the maximum-matching solution and subsequently discuss the formation and cost computation of the key strategies in Section \ref{sec: localstrategies}.



\begin{figure}
    \centering
    \begin{subfigure}[b]{0.45\textwidth}
         \centering
   \includegraphics[width=5cm]{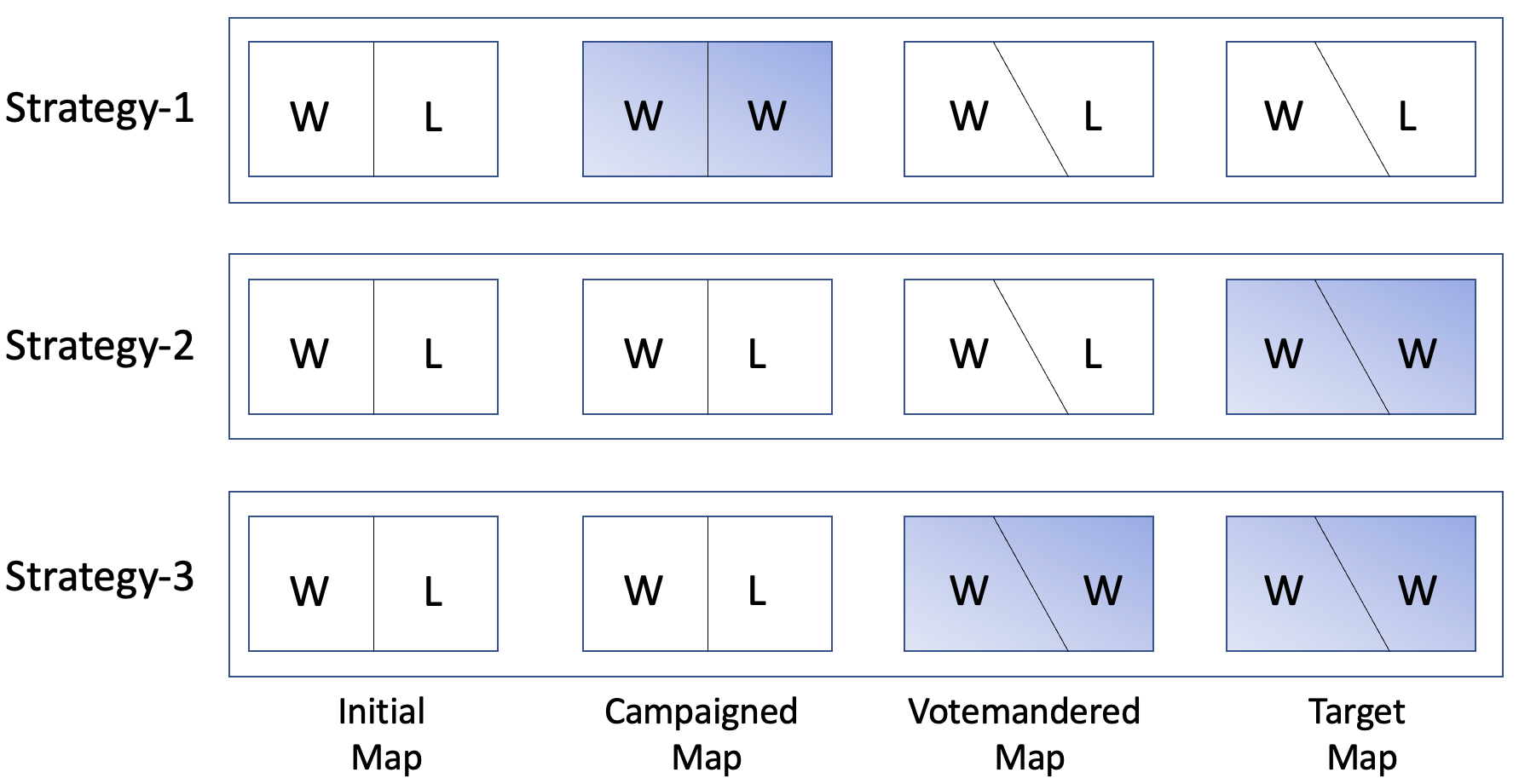}
         \caption{Between neighboring districts with $W,L$ status}
     \end{subfigure}
     \hfill
     \begin{subfigure}[b]{0.45\textwidth}
         \centering
     \includegraphics[width=5cm]{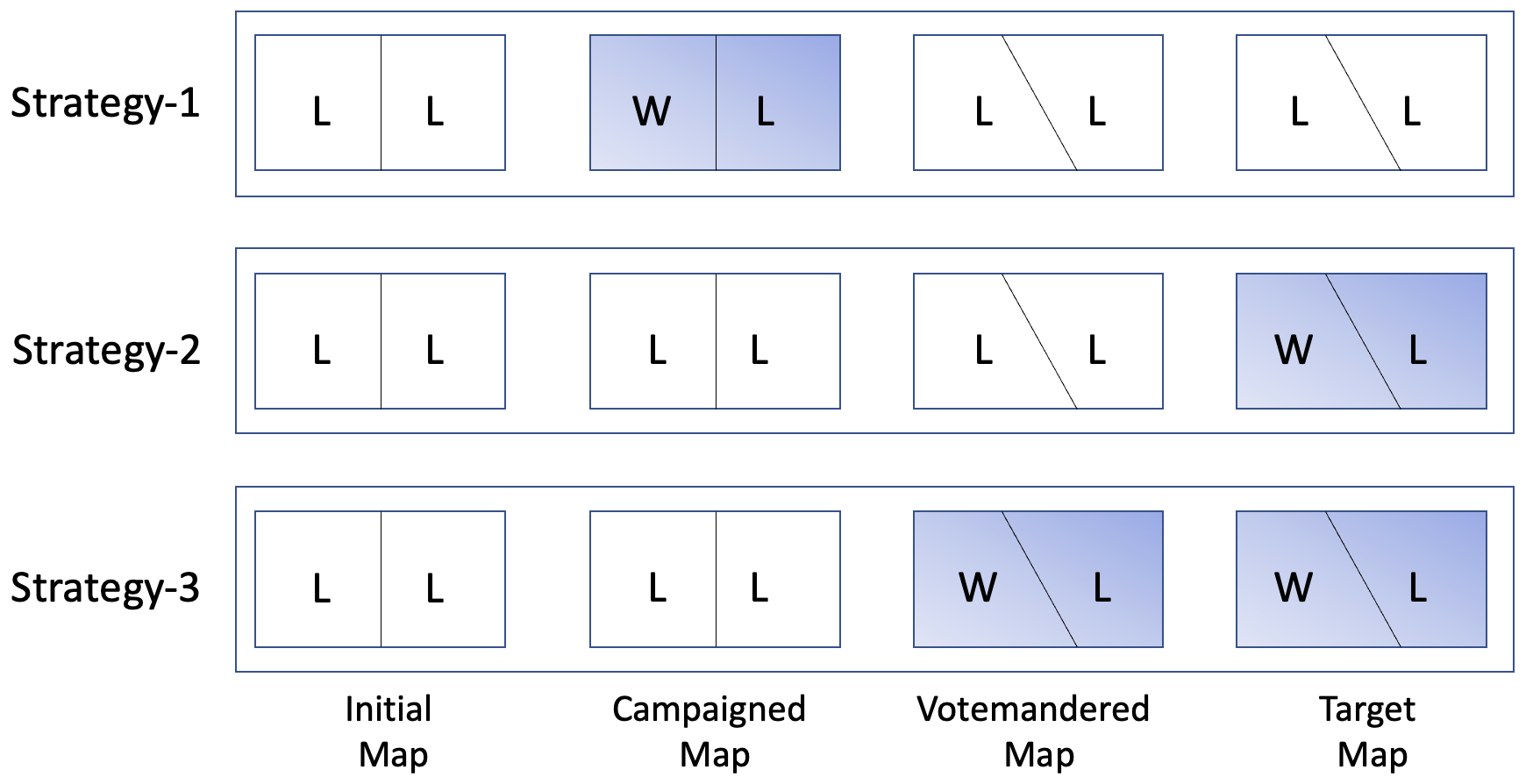}
         \caption{Between neighboring districts with $L,L$ status}
     \end{subfigure}
  \caption{The three key strategies to locally improve the votemandering bonus. The shaded districts show an improvement in the number of wins of the strategist party.}
    \label{fig: threestrategies}
\end{figure} 

\noindent
\subsubsection{An optimization framework for local votemandering.}
As illustrated above, the objective is to select the maximum number of mutually exclusive edges while adhering to the fairness and budget constraints. An edge in the optimal solution may correspond to any of the three strategies depicted in Figure \ref{fig: threestrategies}. Without loss of generality, we assume that the budget is never sufficient to win in all districts in round-1; otherwise, finding efficient votemandering strategies becomes trivial.

Let $n$ represent the set of nodes in the district adjacency graph, and let $E$ denote the set of edges. Let $E_{n_i}$ represent a set of edges (including those corresponding to all three strategies) incident on node $n_i\in n$. Let $b_e, f_e$ denote budget and fairness costs associated with an edge $e$, and let $b$ represent the budget spent to satisfy fairness constraints. Since we do not have a budget to win all districts, $b$ can always be spent on losing districts, contributing $3b/2$ to the fairness cost, as Lemma \ref{lem: tabvotes} demonstrates. The MIP formulation for the local votemandering heuristic can then be expressed as:

{\allowdisplaybreaks \begin{align}
      \max \  \sum_{e \in E} x_e \notag \\
        s.t. \ \ 
       &  \sum_{e\in E} b_e x_e + b \leq \mathcal{B}^A
        \notag \\
         & \sum_{e \in E} f_e x_e - \frac{3}{2}b \leq \ \text{ Fairness cost } \notag \\
        & \sum_{e \in E_{n_i}} x_e \leq 1 & \forall n_i \in n \notag \\
       &  x_{e} \in \{0,1\} &  \forall e \in E \label{opt: local}
\end{align}}

Note that \eqref{opt: local} is equivalent to finding a maximum cardinality matching with two knapsack constraints. Although \eqref{opt: local} is computationally challenging, it significantly simplifies the variable space by discarding unit-specific variables in \eqref{opt:votemander_mip}. In practice, the adjacency graph instance has $|n|$ nodes and only those edges with a positive bonus, making the instance sparse and the problem tractable. We provide an example of this heuristic in Appendix \ref{app: 4a}.


\subsection{Taxonomy of Local Votemandering Strategies} \label{sec: localstrategies}
As shown in \eqref{opt:general_votemander}, securing a win in round-1 is only possible through budget investment, while winning in round-2 can only be achieved through target plan design, which considers the original vote shares. The fairness constraint is relevant because the new plan must satisfy it on the votemandered map, factoring in the vote shares after investments. A strategist party employs the strategies in Figure \ref{fig: threestrategies} to generate bonuses using its budget and redistricting abilities:

\begin{enumerate}
\item Strategy-1: Secure a votemandering bonus with an extra win only in the campaigned map, using $A$'s budget investment. The edge weight vector is of the form: [significant budget, insignificant fairness cost]. The implementation bottleneck is the budget needed for winning (i.e., the vote margin for $L$ district) and finding perturbations that maintain the same $W/L$ status in the votemandered map.
\item Strategy-2: Secure an extra win only in the target map using $B$'s budget investment. The edge weight vector is: [no budget, insignificant fairness cost]. The bottleneck is identifying perturbations that enable $B$ to win in the votemandered map with a margin smaller than its investment, allowing $A$ to win in the target map with original vote shares.
\item Strategy-3: Use boundary perturbations to secure an extra win in the target map while also achieving an additional win in the votemandered map. The edge weight vector is: [insignificant budget, significant fairness cost]. The bottleneck is the fairness cost resulting from the extra win in the votemandered map.
\end{enumerate}

Other strategies, where $A$ wins in a different set of maps than shown in Figure \ref{fig: threestrategies}, may be considered. However, these are expensive and dominated by the key strategies when feasible. For example, winning in both campaigned and votemandered maps incurs significantly higher fairness costs than strategy-1 while providing the same bonus. Winning in both campaigned and target maps combines strategies 1 and 2, making it less likely than either. Therefore, for practical purposes, we only illustrate the three key strategies here and note the generalization of considering additional edge types in \eqref{opt: local}.

Consequently, up to three edges may exist between any two districts, each associated with a specific weight vector. Strategy-2 is particularly important because it only uses information about party $B$'s budget and virtually imposes no cost on party $A$. More details and exact calculations of edge weights can be found in Appendix \ref{app: 4b}.

To determine the optimal edge weights as inputs for \eqref{opt: local}, i.e., costs $b_e, f_e$ for all three strategies, optimal boundary perturbations must be specified while adhering to redistricting constraints such as compactness, contiguity, proximity, and population balance. To achieve this, we employ the randomized recombination technique once again. As with the global heuristic, we separate the problems of finding local perturbations and optimizing costs by creating a pool of plausible submaps for each edge.

For each strategy, we can now examine this pool and select the best edge—i.e., the perturbations leading to a new pair of districts with minimal costs. Finally, given a district adjacency graph, we add up to three edges for each pair of neighboring districts, with costs equal to those of the best pairs from the corresponding pool.


\subsection{Connection to Global Votemandering} \label{sec: localcharacteristics}

As discussed in Section 3, the votemandering bonus $\Delta$ describes how effectively we can votemander. Eq. \eqref{eq: VMbonus} further divides $\Delta$ into three sub-parts.

\begin{align}
\Delta & = \election(\districtplan_0,\tilde{\voterballot}) + \election(\tilde{\districtplan},\voterballot_0) - 2 \election\Paren{\districtplan_0,\voterballot_0} \notag \\
& = [\election(\districtplan_0,\tilde{\voterballot}) - \election\Paren{\districtplan_0,\voterballot_0}]  + [\election(\tilde{\districtplan},\voterballot_0) - \election(\tilde{\districtplan},\tilde{\voterballot}] +
[\election(\tilde{\districtplan},\tilde{\voterballot}) - \election\Paren{\districtplan_0,\voterballot_0}] \label{eq: VMbonus}
\end{align}

The first part of Eq. \eqref{eq: VMbonus}, i.e., the bonus resulting from the difference between the number of wins in campaigned and initial maps, depends on the optimal budget allocation of $A$. The second part, i.e., the difference between the target and votemandered maps, depends on $B$'s budget allocation. The third part, i.e., the difference between votemandered and initial maps, is essentially bounded as only a few discrete EG values are acceptable for the maps, meaning that the votemandered and initial maps can only have a small difference in their number of wins \citep{tam2017measuring}.

Comparing $\Delta$ for global maxima to when the space is restricted for local votemandering, the bonus breakdown highlights the exact three areas where the local heuristic operates (approximating the optimal) through its key strategies. This also suggests that the efficiency of votemandering can be explained using the generalized versions of key strategies, i.e., without necessarily restricting to two districts. Considering the efficiency of local votemandering, note that its produced optimal plan is also achievable by the global heuristic if the plan is present in its pool of maps. In fact, relaxing its more-than-sufficient matching constraints also indicates a possibility of improvement. Despite this, the deeper and targeted local search between every pair of neighboring districts better explores the existence of key strategies and may produce outlier plans more effectively. Moreover, it works precisely by exploiting the information of the initial map. As seen in Section \ref{sec:case_study}, the local votemandering heuristic works efficiently, even outperforming the global heuristic in one case.

Eq. \eqref{eq: VMbonus} also explains why increasing compactness leads to a lower objective and, consequently, a decrease in the ability to votemander. Using the notion of cut edges, higher compactness means fewer cut edges, i.e., a smaller shared boundary between two districts. The efficiency of votemandering depends on the ease of unit exchange across borders via strategies 1 and 2, and more generally, via using Lemma \ref{lem: tabvotes} implications to invest to win in a campaigned map and re-structure to shift this investment to a losing district in the votemandered map. This implies a relationship between the number of cut edges and the ease of votemandering via its key strategies: With a mandated fewer number of cut edges between any two districts, it is more challenging to find units that can be exchanged, resulting in a positive votemandering bonus.

In conclusion, while the global heuristic bypasses the matching constraints, the local heuristic provides an efficient search at the district level, generating plans that closely resemble the original ones. The local search in submaps is entirely parallelizable, potentially yielding much faster results compared to the global heuristic as instance sizes grow. In practice, this runs very quickly. After performing the local search for all edges, the corresponding adjacency graph of the initial map incorporates budget and fairness bounds as variable inputs in the optimization program \eqref{opt: local}. This allows for a more targeted and efficient approach to votemandering, taking advantage of the districts in the initial map.

\section{Case Study: Wisconsin State Senate Redistricting after 2020}
\label{sec:case_study}
This section demonstrates the existence of practical votemandering strategies using Wisconsin state senate redistricting after the 2020 census.  The redistricting cycle was delayed due to lawsuits, prompting the Wisconsin state legislature and the governor's People's Maps Commission to propose state and congressional district plans. The Wisconsin Supreme Court eventually approved the state senate and house maps drawn by the legislature \citep{ballotpedia}.

Wisconsin's balanced partisan composition provides a suitable environment for exploring votemandering's practical potential. The state senate has 33 seats, with approximately half up for election every two years. The Republican party (R) controls the state senate with a 21-11 majority (excluding one vacancy) and had a $51.1\%$ statewide senate election vote share in previous elections. Conversely, the Democratic party (D) won the 2020 presidential election in Wisconsin by a $0.63\%$ margin, and the state's governor is a Democrat. The Democratic governor's veto power over redistricting proposals from the Republican state legislature ensures that both parties influence the final state senate map. Additionally, the Wisconsin state constitution mandates that districts be compact, contiguous, and "bounded by county, precinct, town, or ward lines where possible" \citep{wiscconsti}.

Throughout this section, we consider the 2021 governor's office final state senate plan (\govmap) as the initial map. Deemed fair by the Princeton Gerrymandering Project with a score of A (\citeyear{wiscprinceton}), \govmap\ serves as a reasonable starting point for our votemandering case study, despite not being enacted.

The \govmap\ plan incorporates state senate election data from 2018 and 2020 (for 17 and 16 seats, respectively). According to the 2020 census, the ideal population for each senate district is 178,598. We assume a statewide voter turnout of $65\%$, the average from the 2018 and 2020 elections. Both parties receive a budget to influence 13,734 voters, constituting a $1\%$ total investment of the expected votes cast. The votemandering party strategizes its campaign investment, while the other party is assumed to allocate its investment proportionally to unit populations.

Section \ref{subsec:wi_global} investigates global votemandering effects for each major party, revealing that although both parties can benefit, the Democratic party achieves a more substantial votemandering bonus given \govmap\ as the initial map. Section \ref{subsec:wi_local} illustrates local votemandering strategies that better align with the Wisconsin Supreme Court's goal of minimizing changes to the previous district plan. Both parties continue to benefit within this restricted strategy space, and the votemandering bonus even increases for one party.

\subsection{Wisconsin Global Votemandering}
\label{subsec:wi_global}

\begin{figure}
    \centering
    \begin{subfigure}[b]{0.48\textwidth}
         \centering
        \includegraphics[width=4.5cm]{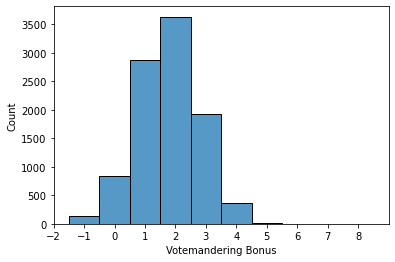}
         \caption{Republican Party Votemandering}
         \label{fig: rephist}
     \end{subfigure}
     \begin{subfigure}[b]{0.48\textwidth}
         \centering
        \includegraphics[width=4.5cm]{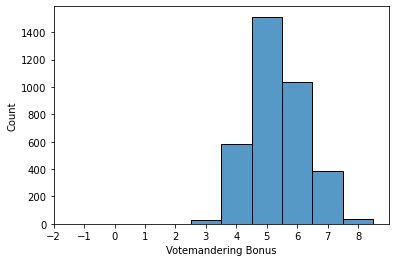}
         \caption{Democratic Party Votemandering}
         \label{fig: demhist}
     \end{subfigure}
    \caption{Distribution of Votemandering Objectives}
    \label{fig: histograms}
\end{figure} 
\begin{figure}
    \centering
   \includegraphics[width=5cm]{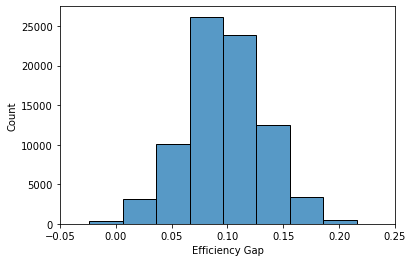}
    \caption{Distribution of Efficiency Gap measure over the pool of maps}
    \label{fig: EGhistogram}
\end{figure}

The global votemandering heuristic expounded in Section \ref{sec:methodology} 
is applied to Wisconsin state senate redistricting, using a pool of
80,000
candidate target maps generated via recombination. 
Figure \ref{fig: histograms} depicts 
the pool's distribution of votemandering bonus 
from each party's perspective. 
Each histogram shows the number of maps 
yielding a specific votemandering bonus 
for the votemandering party 
when fixed as the target map. 
The EG of \govmap, $0.1409$, 
indicates more wasted votes for the Democratic party, 
giving the Democratic party more room to gain seats 
while maintaining a safe EG. 
Hence the Democratic party tends to achieve 
a greater votemandering bonus
than does the Republican party, i.e., 
the distribution in Figure \ref{fig: demhist} 
is shifted to the right of 
the distribution in Figure \ref{fig: rephist}. 
Note the relatively high EG of \govmap\  
does not necessarily imply the map is a partisan gerrymander. 
Figure \ref{fig: EGhistogram} shows the distribution of EG for the pool of maps. 
Due to the spatial distribution of voters in Wisconsin, 
recombination tends to generate maps with positive (i.e., Republican-leaning) 
EG values.
Based on the EG value of \govmap\ and 
the pool's EG distribution, 
we impose a fairness bound of $0.0 \le \text{EG} \le 0.15$ 
throughout the case study. 

\begin{table}[ht]
\centering
\small
\caption{Wisconsin Global Votemandering Characteristics: Republican and Democratic Parties}
\label{tab: casestudy}
\begin{tabular}{|l|cc|cc|}
\hline
\multicolumn{1}{|c|}{Majority party} & \multicolumn{2}{c|}{Republican Votemandering}        & \multicolumn{2}{c|}{Democratic Votemandering}        \\ \hline
\multicolumn{1}{|c|}{}               & \multicolumn{1}{c|}{Number of Wins} & Efficiency Gap & \multicolumn{1}{c|}{Number of Wins} & Efficiency Gap \\ \hline
Initial Map                          & \multicolumn{1}{c|}{21}             & 0.1409         & \multicolumn{1}{c|}{12}             & 0.1409         \\ \hline
Campaigned Map                       & \multicolumn{1}{c|}{24}             & 0.2314         & \multicolumn{1}{c|}{15}             & 0.0468           \\ \hline
Votemandered Map                     & \multicolumn{1}{c|}{21}             & 0.1285         & \multicolumn{1}{c|}{16}             & 0.0096     \\ \hline
Target Map                           & \multicolumn{1}{c|}{23}             & 0.1915       & \multicolumn{1}{c|}{17}             & -0.0195         \\ \hline
\end{tabular}
\normalsize
\end{table}

To show an illustration of how these votemandering objectives are attained, we next describe optimal target maps with both parties as strategists. Table \ref{tab: casestudy} shows the characteristics of votemandering stages  specific to our examples. Figures \ref{fig: WI1}-\ref{fig: WI2} visually show the initial and target maps, strategic investment of the Republican party, and the various stages of votemandering on the real state data of Wisconsin. Figures \ref{fig: WIglobal1D}-\ref{fig: WIglobal2D} show the same for the Democratic party.

For the Republican votemandering, the objective is  47 seats across two rounds, against 42 with no strategic investment. The  Democratic votemandering bonus is 8, showing a larger improvement for the Democratic party as explained above. The proposed Republican map would show an EG of 0.1285 with the previous election data (apparently fairer than the initial map), with the actual value being 0.1915. This is reflected in the difference of 2 seats between the votemandered and target map. Most interestingly, although the investment leads to 3 new seats in the campaigned map, the new plan is such that it completely negates this effect, making the votemandered map win 21 seats. 

In these cases, both parties choose outlier maps as the target maps and make the votemandered maps fair. The strategic investment for both cases is done across the map by smartly choosing units: the bottom-left part in Figure \ref{fig: investmentglobD} is largely uniform (to make a difference only in round-1), the top-half part in both Figure \ref{fig: investmentglobR} and \ref{fig: investmentglobD} is non-uniform (to only affect units that remain part of losing districts in the target map) and finally, the bottom-half part in Figure
\ref{fig: investmentglobR} is sporadic, with the right corner serving as an investment made to satisfy the fairness bound.
\begin{figure}[!ht]
    \centering
    \begin{subfigure}[b]{0.32\textwidth}
         \centering
            \includegraphics[width=3cm]{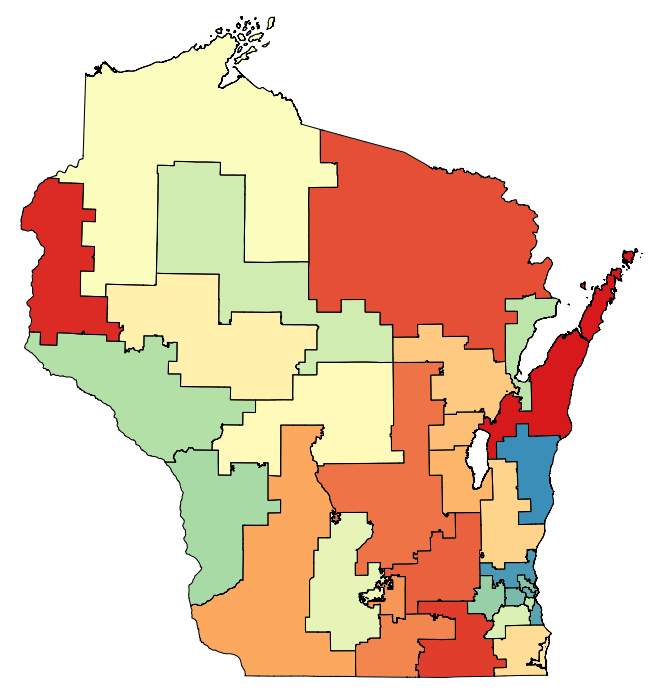}
         \caption{Initial Map}
         \label{fig: initmapglobR}
     \end{subfigure}
     \begin{subfigure}[b]{0.32\textwidth}
         \centering
        \includegraphics[width=3cm]{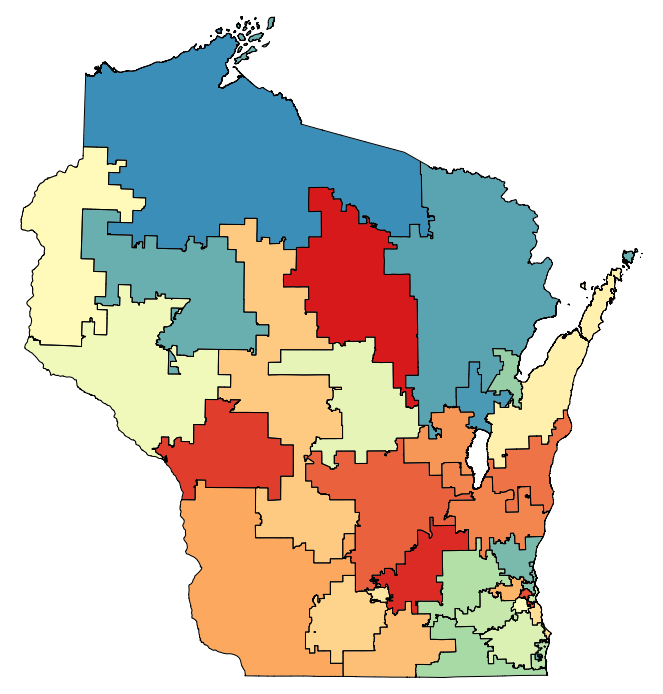}
         \caption{Target Map}
         \label{fig: targetmapglobR}
     \end{subfigure}
     \begin{subfigure}[b]{0.32\textwidth}
         \centering
        \includegraphics[width=3cm]{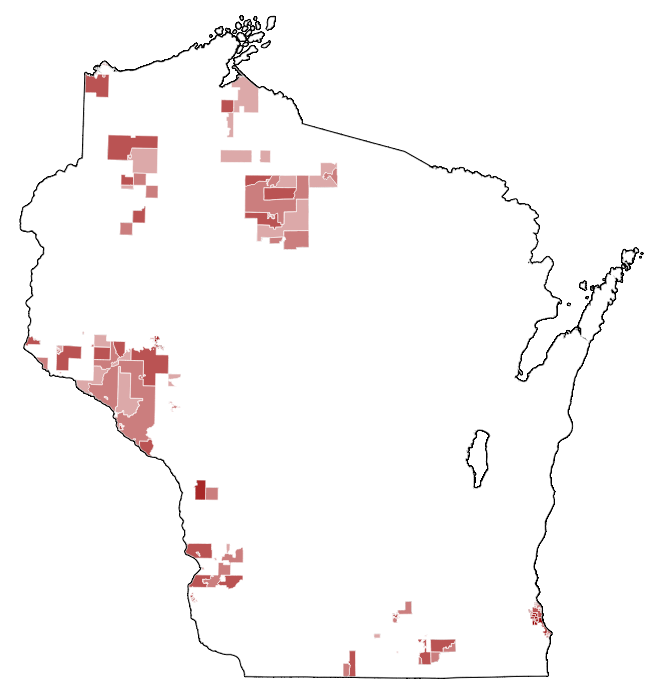}
         \caption{Investment}
         \label{fig: investmentglobR}
     \end{subfigure}
   \caption{Republican Global Votemandering: The initial map, the chosen target map, and the strategic investment of budget (with intensity  indicated by the darker color)}
    \label{fig: WI1}
\end{figure} 
\begin{figure}[!ht]
    \centering
    \begin{subfigure}[b]{0.24\textwidth}
         \centering
             \includegraphics[width=2.2cm]{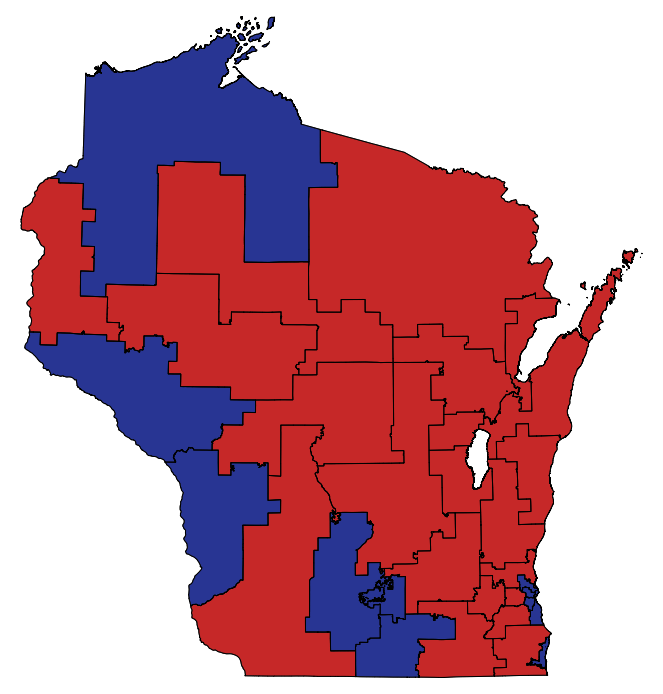}
         \caption{Initial Map}
         \label{fig: initmapglobwinsR}
     \end{subfigure}
     \begin{subfigure}[b]{0.24\textwidth}
         \centering
       \includegraphics[width=2.2cm]{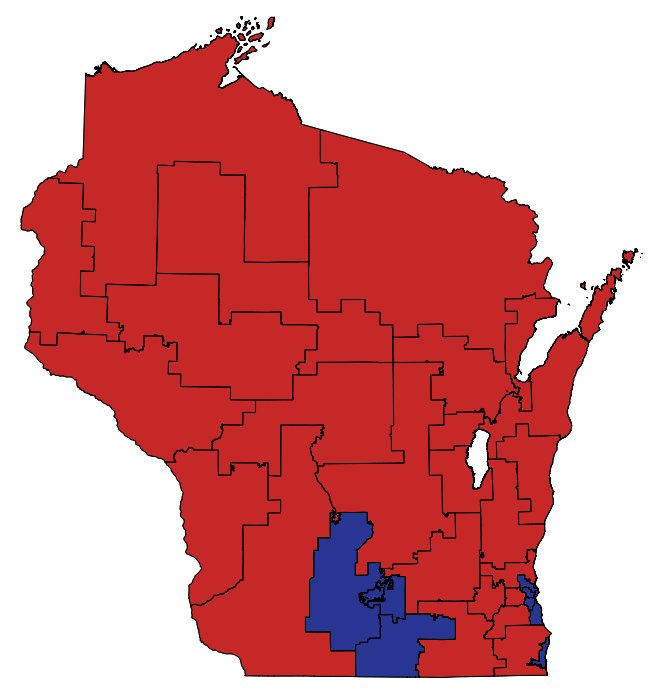}
         \caption{Campaigned Map}
         \label{fig: campaignedmapglobwinsR}
     \end{subfigure}
     \centering
    \begin{subfigure}[b]{0.24\textwidth}
         \centering
              \includegraphics[width=2.2cm]{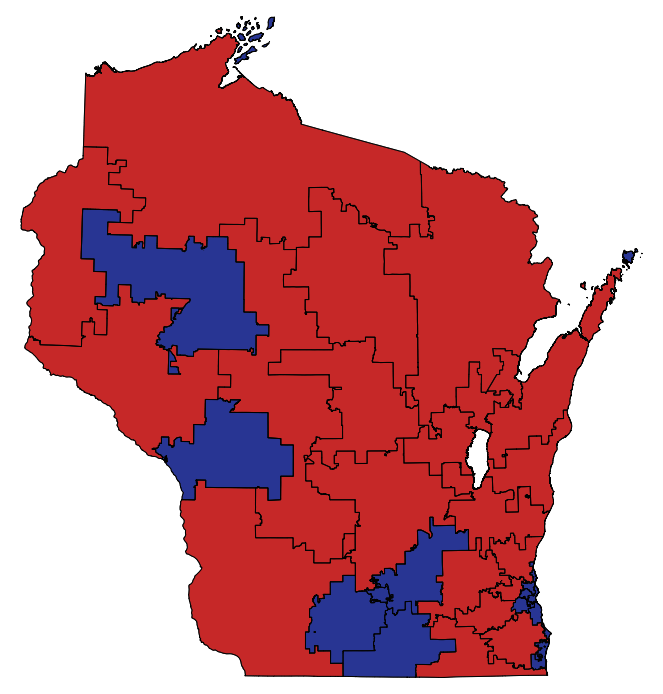}
         \caption{Votemandered Map}
         \label{fig: vmmapglobwinsR}
     \end{subfigure}
     \begin{subfigure}[b]{0.24\textwidth}
         \centering
      \includegraphics[width=2.2cm]{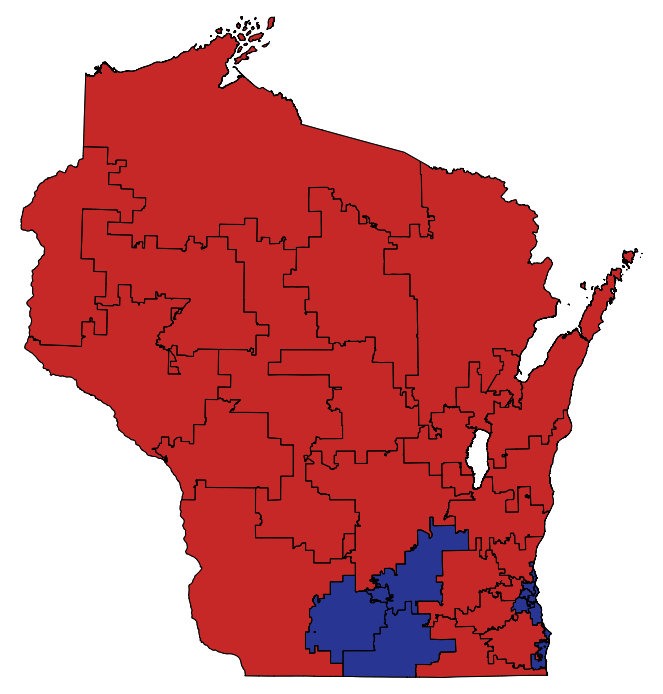}
         \caption{Target Map}
         \label{fig: targetmapglobwinsR}
     \end{subfigure}
   \caption{The Four Stages of Republican Global Votemandering, (with red and blue indicating the districts won by the Republican and Democratic parties, respectively)}
    \label{fig: WI2}
\end{figure} 
\begin{figure}[!ht]
    \centering
    \begin{subfigure}[b]{0.32\textwidth}
         \centering
            \includegraphics[width=3cm]{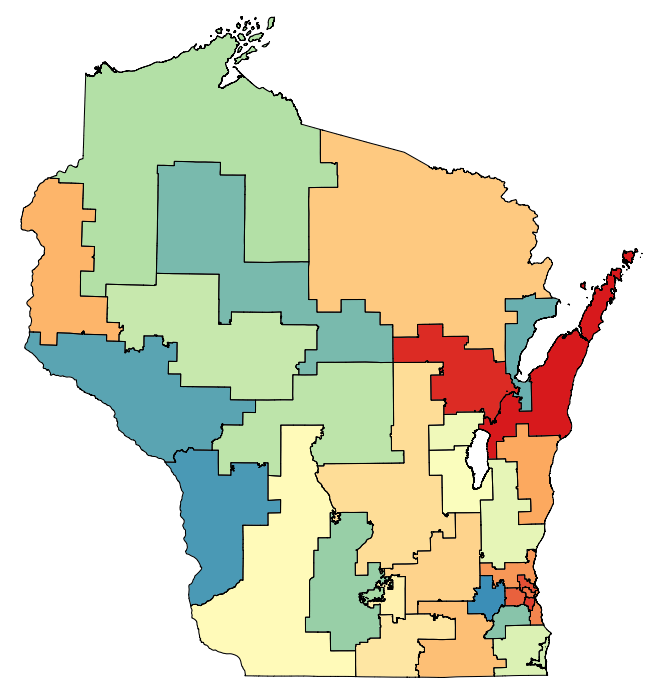}
         \caption{Initial Map}
         \label{fig: initmapglobD}
     \end{subfigure}
     \begin{subfigure}[b]{0.32\textwidth}
         \centering
    \includegraphics[width=3cm]{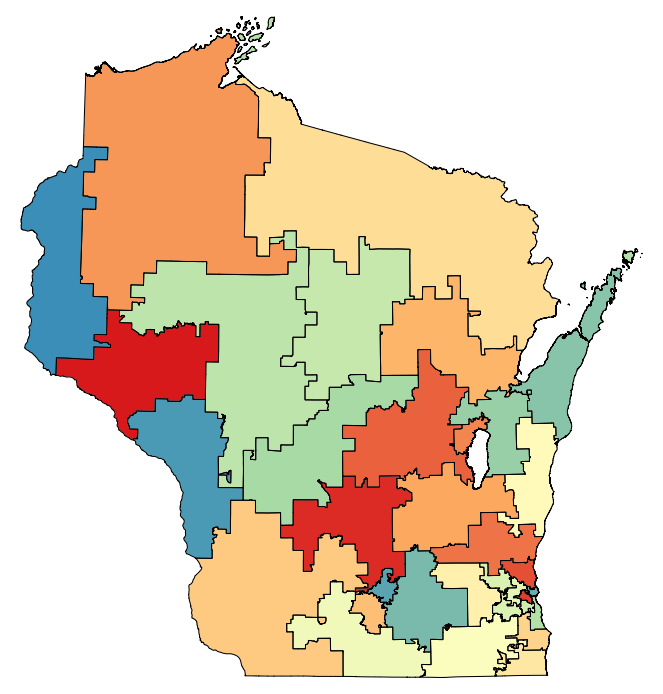}
         \caption{Target Map}
         \label{fig: targetmapglobD}
     \end{subfigure}
     \begin{subfigure}[b]{0.32\textwidth}
         \centering
        \includegraphics[width=3cm]{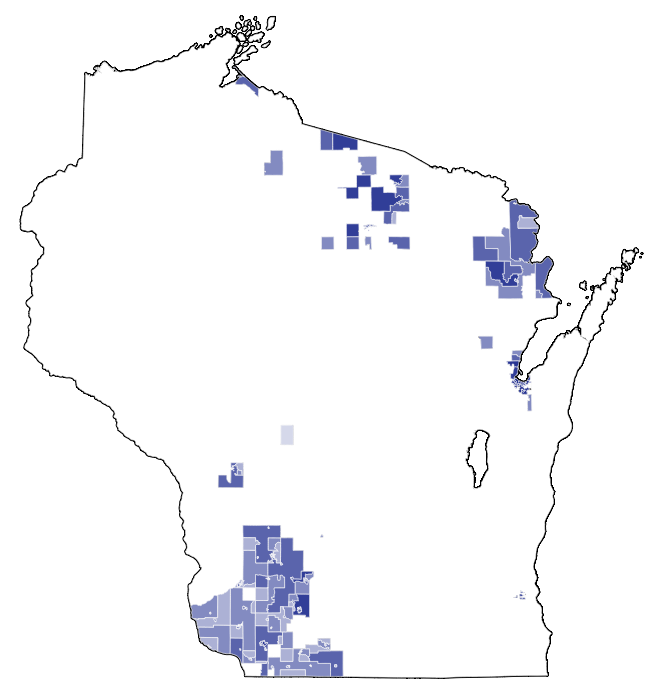}
         \caption{Investment}
         \label{fig: investmentglobD}
     \end{subfigure}
   \caption{Democratic Global Votemandering: The initial map, the chosen target map, and the strategic investment of budget (with intensity  indicated by the darker color)}
    \label{fig: WIglobal1D}
\end{figure}
\begin{figure}[!ht]
    \centering
    \begin{subfigure}[b]{0.24\textwidth}
         \centering
             \includegraphics[width=2.2cm]{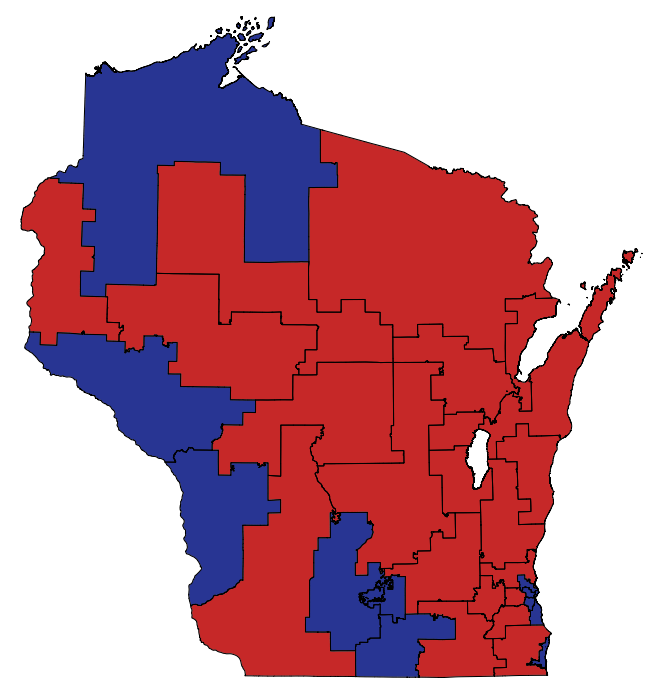}
         \caption{Initial Map}
         \label{fig: initmapglobwinsD}
     \end{subfigure}
     \begin{subfigure}[b]{0.24\textwidth}
         \centering
       \includegraphics[width=2.2cm]{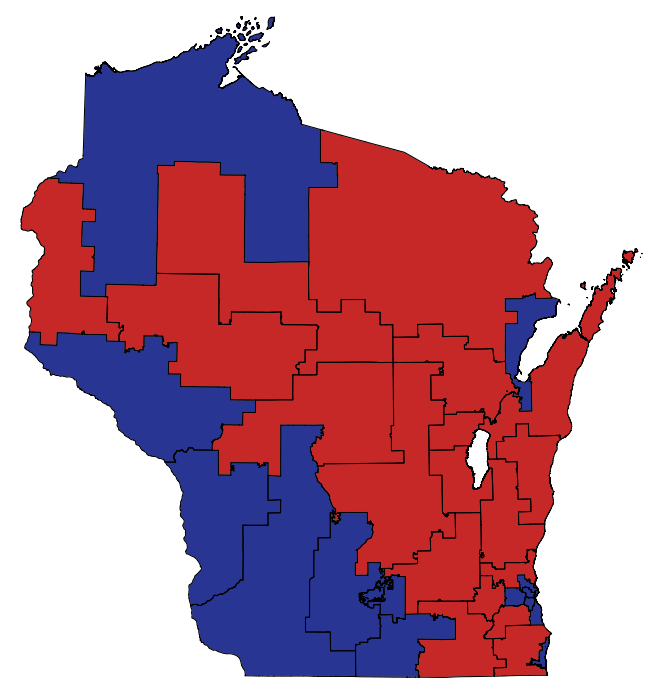}
         \caption{Campaigned Map}
         \label{fig: campaignedmapglobwinsD}
     \end{subfigure}
     \centering
    \begin{subfigure}[b]{0.24\textwidth}
         \centering
    \includegraphics[width=2.2cm]{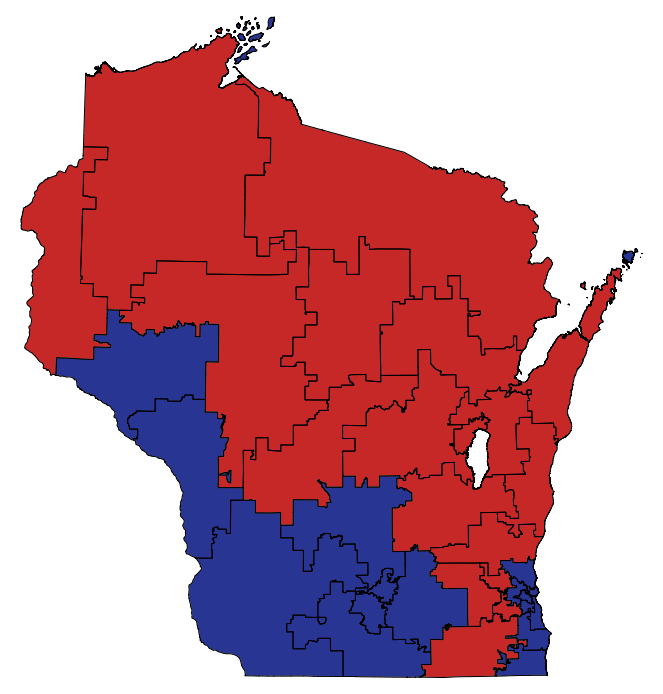}
         \caption{Votemandered Map}
         \label{fig: vmmapglobwinsD}
     \end{subfigure}
     \begin{subfigure}[b]{0.24\textwidth}
         \centering
      \includegraphics[width=2.2cm]{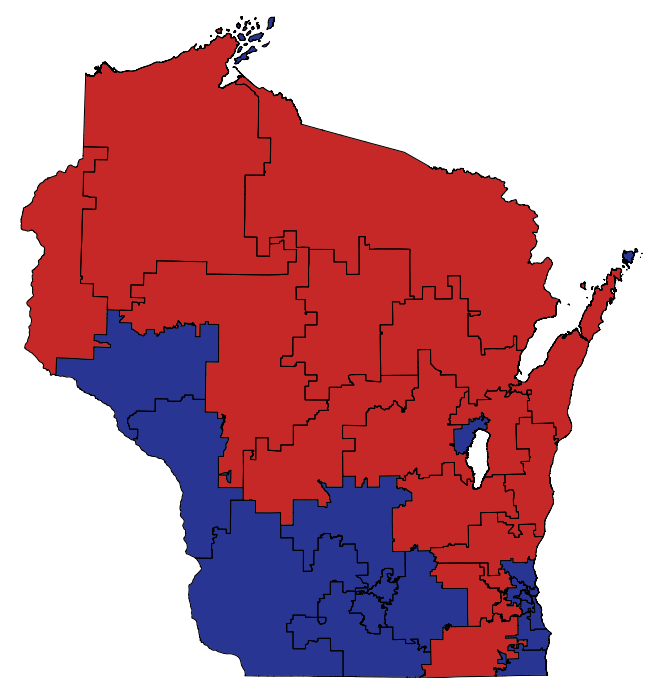}
         \caption{Target Map}
         \label{fig: targetmapglobwinsD}
     \end{subfigure}
   \caption{The Four Stages of Democratic Global Votemandering, (with red and blue indicating the districts won by the Republican and Democratic parties, respectively)}
    \label{fig: WIglobal2D}
\end{figure}

\subsection{Wisconsin Local Votemandering}
\label{subsec:wi_local}
When forced to decide the final state senate district plan, 
the Wisconsin Supreme Court announced that it would seek 
to make as few changes as possible to the existing map. 
In accordance with this goal, 
this section applies the bottom-up local votemandering heuristic 
to find strategies for both parties which produce target maps close to \govmap. 
\begin{table}[!hb]
\centering
\small
\caption{Local Votemandering Characteristics: Republican and Democratic Parties} \label{tab: casestudylocal}
\begin{tabular}{|l|cc|cc|}
\hline
\multicolumn{1}{|c|}{Majority party} & \multicolumn{2}{c|}{Republican Votemandering}        & \multicolumn{2}{c|}{Democratic Votemandering}        \\ \hline
\multicolumn{1}{|c|}{}               & \multicolumn{1}{c|}{Number of Wins} & Efficiency Gap & \multicolumn{1}{c|}{Number of Wins} & Efficiency Gap \\ \hline
Initial Map                          & \multicolumn{1}{c|}{21}             & 0.1409         & \multicolumn{1}{c|}{12}             & 0.1409         \\ \hline
Campaigned Map                       & \multicolumn{1}{c|}{21}             & 0.1478         & \multicolumn{1}{c|}{14}             & 0.1072           \\ \hline
Votemandered Map                     & \multicolumn{1}{c|}{21}             & 0.1442         & \multicolumn{1}{c|}{12}             & 0.1014      \\ \hline
Target Map                           & \multicolumn{1}{c|}{25}             & 0.2513        & \multicolumn{1}{c|}{20}             & -0.0888           \\ \hline
\end{tabular}
\normalsize
\end{table}

Both parties can gain advantages from local votemandering, although each party employs a significantly different approach. Figure \ref{fig: localgraphs} illustrates the \govmap\ district adjacency graph (\ref{fig: distadj}), the best Republican strategy discovered (\ref{fig: repmatch}), and the best Democratic strategy discovered (\ref{fig: demmatch}). Node colors represent the party with a higher vote-share in each district. Edges with nonzero weights indicate strategy-1 edges, while 0 weights correspond to strategy-2, and `FC' denotes fairness costs associated with strategy-3.

\begin{figure}[!hb]
    \centering
    \begin{subfigure}[b]{0.32\textwidth}
         \centering
    \includegraphics[width=3.5cm]{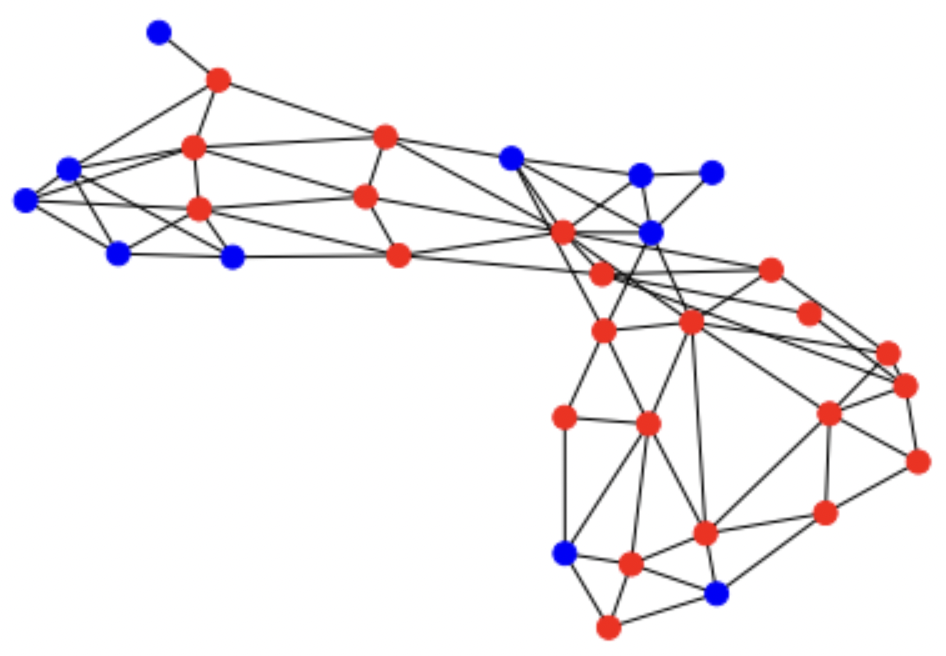}
         \caption{District Adjacency}
         \label{fig: distadj}
     \end{subfigure}
     \begin{subfigure}[b]{0.32\textwidth}
         \centering
    \includegraphics[width=3.5cm]{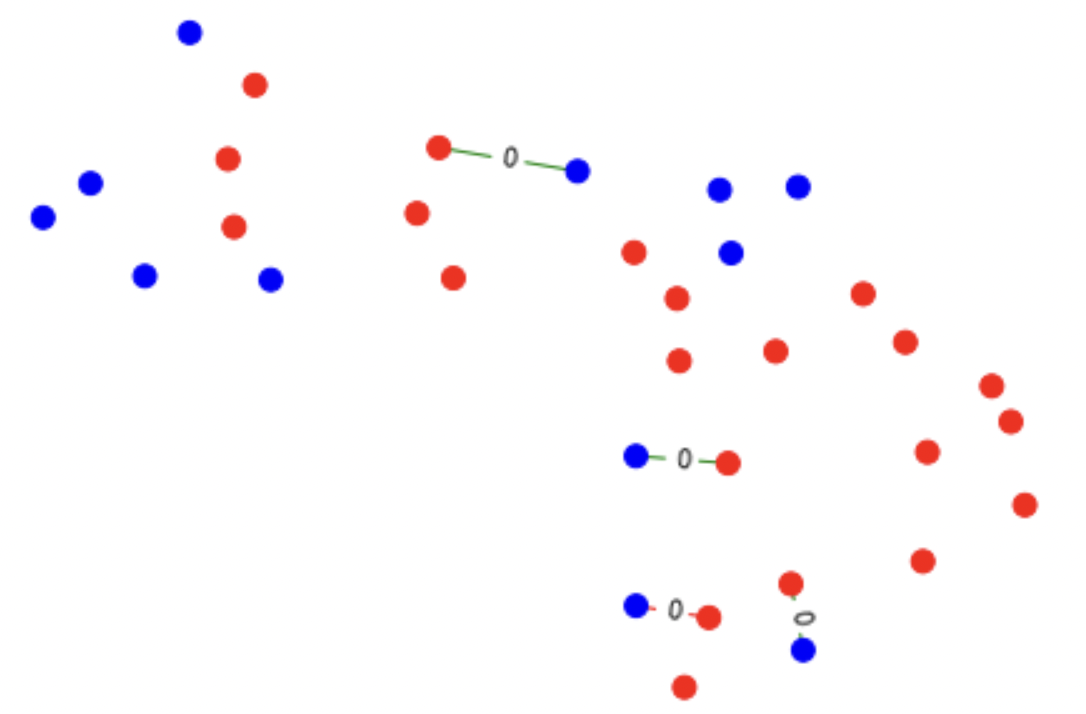}
         \caption{Republican Solution}
         \label{fig: repmatch}
     \end{subfigure}
     \begin{subfigure}[b]{0.32\textwidth}
         \centering
    \includegraphics[width=3.5cm]{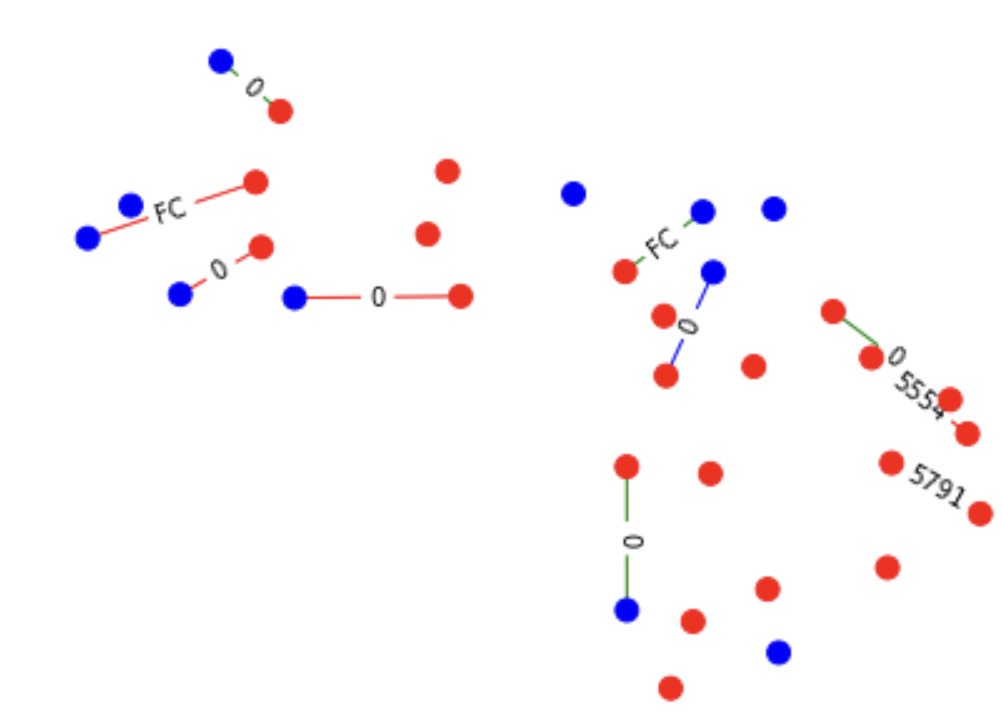}
         \caption{Democratic Solution}
         \label{fig: demmatch}
     \end{subfigure}
    \caption{Local Votemandering: The district adjacency graph and the maximum matching solutions}
    \label{fig: localgraphs}
\end{figure} 

The maximum matching solution for the Republican party (Figure \ref{fig: repmatch}) yields a votemandering bonus of four seats through strategy-2 improvements (i.e., without incurring monetary or fairness costs). Consequently, the Republicans' local votemandering bonus, which results in 46 seats across the two elections, is one seat fewer than their global votemandering bonus of 47 seats. The maximum matching solution for the Democratic party (Figure \ref{fig: demmatch}) attains a votemandering bonus of ten seats, surpassing their global votemandering bonus of eight, via two strategy-1, six strategy-2, and two strategy-3 edges. Strategy-1 edges consume a budget of 11,347 to boost Democratic voter turnout. Table \ref{tab: casestudylocal} details the map characteristics for each stage of local votemandering by each party, while Figures \ref{fig: WIlocal1R}-\ref{fig: WIlocal2D} in Appendix \ref{app: sec5} display the corresponding maps and campaign strategies.

Local votemandering strategies for both parties are not only computationally more efficient than their global votemandering counterparts but also generate target maps closer to \govmap. As described in Section \ref{sec:local_vm}, the closeness parameter is adjustable, allowing for the optimization of the bonus while producing maps as close to the initial map as desired. Notably, local votemandering relies extensively on 
strategy-2 improvements, which exploit the other party's campaign investments. Overall, the case study demonstrates 
the vulnerability of Wisconsin state senate redistricting 
to global and local votemandering.
\section{Conclusions and Future Directions}
\label{sec:conclusion}
In this study, we introduce the concept of \emph{votemandering}, a combination of strategic campaigning and gerrymandering employed to deceive fairness measures and increase seat-share across multiple elections. Focusing on the efficiency gap (EG) as a fairness metric, we establish sufficient conditions for a positive votemandering bonus (Theorem \ref{thm: existence}) and present an efficient heuristic for identifying votemandering strategies (Algorithm \ref{algo: globalvotemandering}, Proposition \ref{thm: algo}, Theorem \ref{thm: polytime}).

Through computational experiments, we investigate the impact of campaign budget, compactness, voter turnout, and spatial autocorrelation of voters on votemandering efficacy. Our findings indicate that enhancing voter turnout and compactness, parameters seemingly unrelated to partisan fairness, can potentially mitigate the influence of votemandering. A case study of Wisconsin state senate redistricting illustrates practical votemandering strategies for both parties, emphasizing its applicability beyond hypothetical scenarios and into real-world instances.

To further demonstrate the practicality of votemandering, we introduce \emph{local votemandering}, which allows the party controlling redistricting to make minor adjustments to a limited number of district boundaries. Our heuristic efficiently discovers profitable district maps with minimal district pair recombination. In the Wisconsin case study, local votemandering yields a higher votemandering bonus for one party compared to the global, pool-based heuristic.

Future research may explore votemandering in the context of alternative fairness measures by utilizing the general framework outlined here. Determining the most effective fairness measure, or a combination thereof, to prevent votemandering remains an open question. Our votemandering model currently allows for strategic campaigning only in the first election; expanding it to include the second election would increase realism but also complexity. Further extensions might examine votemandering across more than two election rounds, which would necessitate accounting for migration patterns and shifts in voter sentiment. Moreover, allowing both parties to strategically allocate their campaign budgets would introduce a generalized version of the Colonel Blotto game in the redistricting context, presenting challenges such as breaking ties and addressing increased computational complexity.

Finally, our work carries significant policy and practical implications. We demonstrate that strategic campaigning can substantially impact redistricting outcomes despite the presence of fairness constraints and that the sole use of EG as a fairness measure may be insufficient. Our results advocate for additional measures that account for strategic behavior and policies that curtail manipulative campaigning practices. Ultimately, this study underscores the importance of considering strategic behavior when designing and evaluating redistricting processes, and the need for ongoing research on the consequences of political manipulation in democratic systems.

\newpage
\bibliographystyle{ACM-Reference-Format}
\bibliography{biblio}

\clearpage
\appendix
\section{Notation}\label{app:notation}
\begin{table}[!ht]
\small
\begin{tabular}{|c|l|}
\hline
\multicolumn{1}{|l|}{Notation}                                     & Definition                                                                     \\ \hline
$\districtplans$                                                   & Set of district plans $\{D\}$                                                  \\ \hline
$\voterballots$                                                    & Set of voter ballots $\{V\}$                                                   \\ \hline
$D_0,V_0$                                                          & Original plan, original vote data                                              \\ \hline
$\tilde{D}, \tilde{V}$                                             & New plan, new vote data                                                        \\ \hline
$\election$                                                        & Election function $\districtplans \times \voterballots \rightarrow \mathbb{N}$ \\ \hline
$f$                                                                & Fairness function $\districtplans \times \voterballots \rightarrow \mathbb{R}$ \\ \hline
$\delta$                                                           & Fairness threshold; the map is fair if $f(D,V) \leq \delta$                               \\ \hline
$K$                                                                & Set of units                                                                   \\ \hline
$n$                                                                & Number of districts                                                            \\ \hline
$z_{ij}^r$                                                         & Indicates 1 if unit $j$ is assigned to the district with center $i$            \\ \hline
$v_{init, k}^A$, $v_{init, k}^B$                                   & Vote shares of party $A$ and $B$ in unit $k$                                               \\ \hline
$\alpha$                                                           & Fractional baseline voter turnout                                              \\ \hline
$\mathcal{B}^A$, $\mathcal{B}^B$                                   & Party $A$ and $B$'s GOTV budgets                                               \\ \hline
$b_k^A, b_k^B$                                                     & Budget allocation by $A$ and $B$ in unit $k$                                   \\ \hline
$\hat{s}_i^1 $,  $\hat{s}_i^{2}$                                   & Indicate 1 for $A$'s wins in campaigned, target maps                                 \\ \hline
$\pool$                                                            & Pool of district plans                                                         \\ \hline
$\mathcal{I} = \{I_1, ..I_{n}\}$, $\mathcal{J}= \{J_1, ..J_{n}\}$  & District assignments in the original, new plans                                \\ \hline
$\hat{x}_{I_i},  \hat{y}_{J_j}$                                    & Indicate 1 for $A$'s 
 wins in the campaigned, votemandered maps                       \\ \hline
($V_{init,I}^{A}$, $V_{init,I}^{B}$) and ($V_{I}^{A}$, $V_{I}^{B}$) & Pre-campaigning and post-campaigning votes in district $I$                     \\ \hline
$\mathcal{W}(B-A)_I$                                               & Difference between wasted votes in district assignment $I$                                        \\ \hline
$\Delta$                                                           & Votemandering bonus                                                            \\ \hline
\end{tabular}
\caption{Notation}
\label{tab: notation}
\end{table}
\normalsize

\section{Section \ref{sec:methodology} Details} \label{app:sec_methodology}

\subsection{Votemandering Illustration}
\label{app: example}
As an illustration of the votemandering phenomenon, consider a hypothetical state comprised of a $10 \times 10$ grid of equipopulous counties to be partitioned into $k=5$ districts. We start with the initial map with given unit vote shares as marked in Figure \ref{fig: imex}. Both parties have close state-wide total vote proportions ($51\%, 49\%$), spread uniformly throughout the map.  Out of the total vote shares in each unit, 50$\%$ votes are cast without any campaign influence, i.e., it is the baseline voter turnout, allowing the remaining to be added through GOTV campaigning. In the initial map, party $A$ wins three districts out of 5, as marked by the units with the symbol `A'. The resulting efficiency gap (EG) of the initial map is $0.072$. In this case of 5 districts, we call a map `fair' if its EG is less than or equal to 0.20 \citep{stephanopoulos2015partisan}. Thus, any map proposed by $A$ for the second round must have its EG less than 0.20.

\begin{figure}
    \centering
    \begin{subfigure}[b]{0.45\textwidth}
         \centering
       \includegraphics[width=5cm]{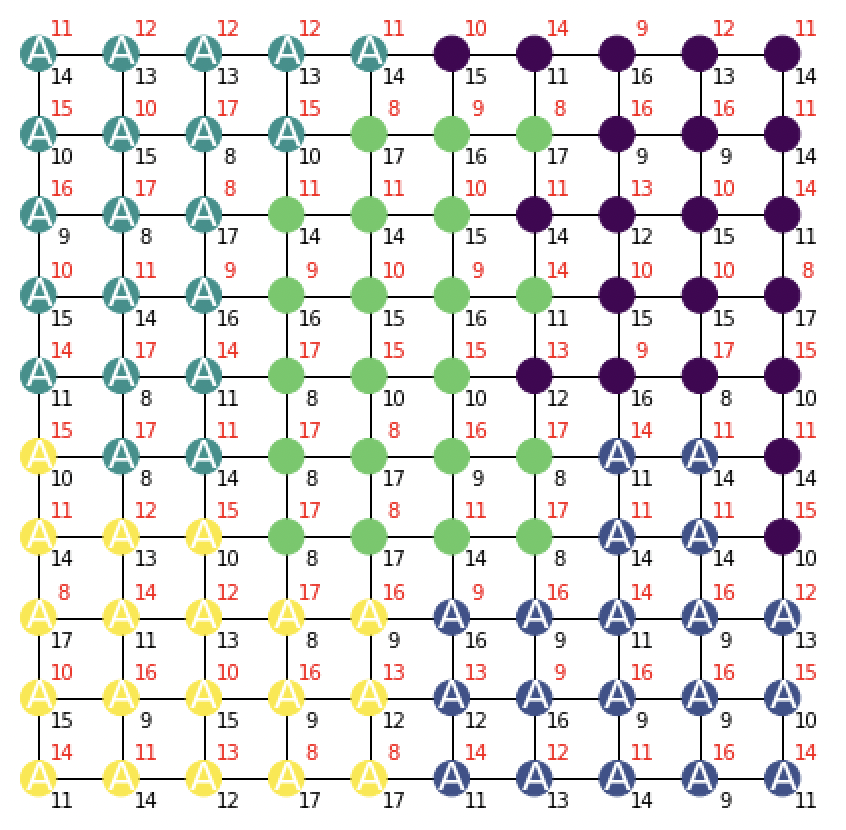}
         \caption{Initial Map}
         \label{fig: imex}
     \end{subfigure}
     \begin{subfigure}[b]{0.45\textwidth}
         \centering
       \includegraphics[width=5cm]{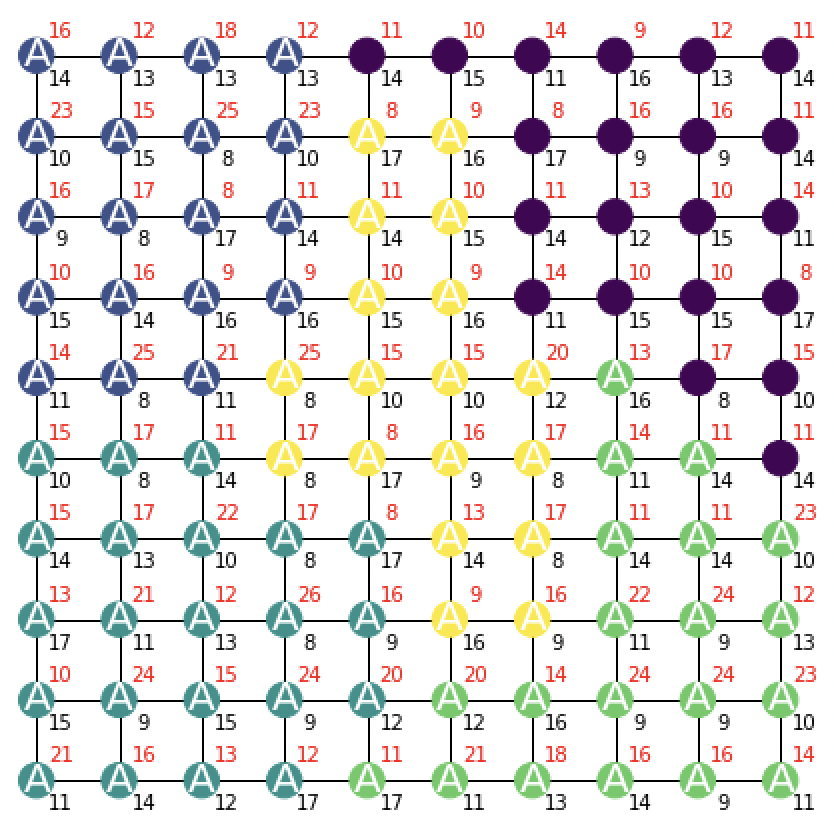}
         \caption{Votemandered Map}
         \label{fig: vmex}
     \end{subfigure}
  \caption{Initial map and the proposed map for round-2}
       \label{fig: original and target maps}
\end{figure}

\begin{figure}
    \centering
    \begin{subfigure}[b]{0.45\textwidth}
         \centering
    \includegraphics[width=6cm]{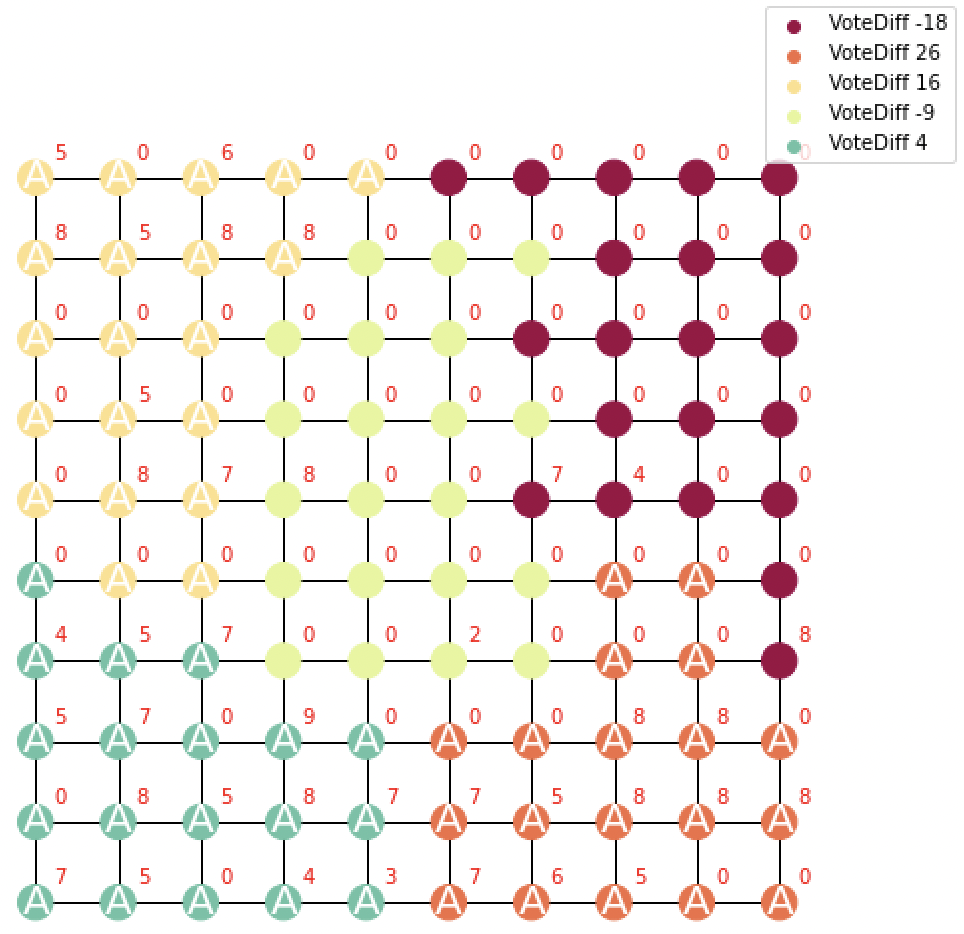}
         \caption{Initial Map}
         \label{fig: investimex}
     \end{subfigure}
     \begin{subfigure}[b]{0.45\textwidth}
         \centering
    \includegraphics[width=6cm]{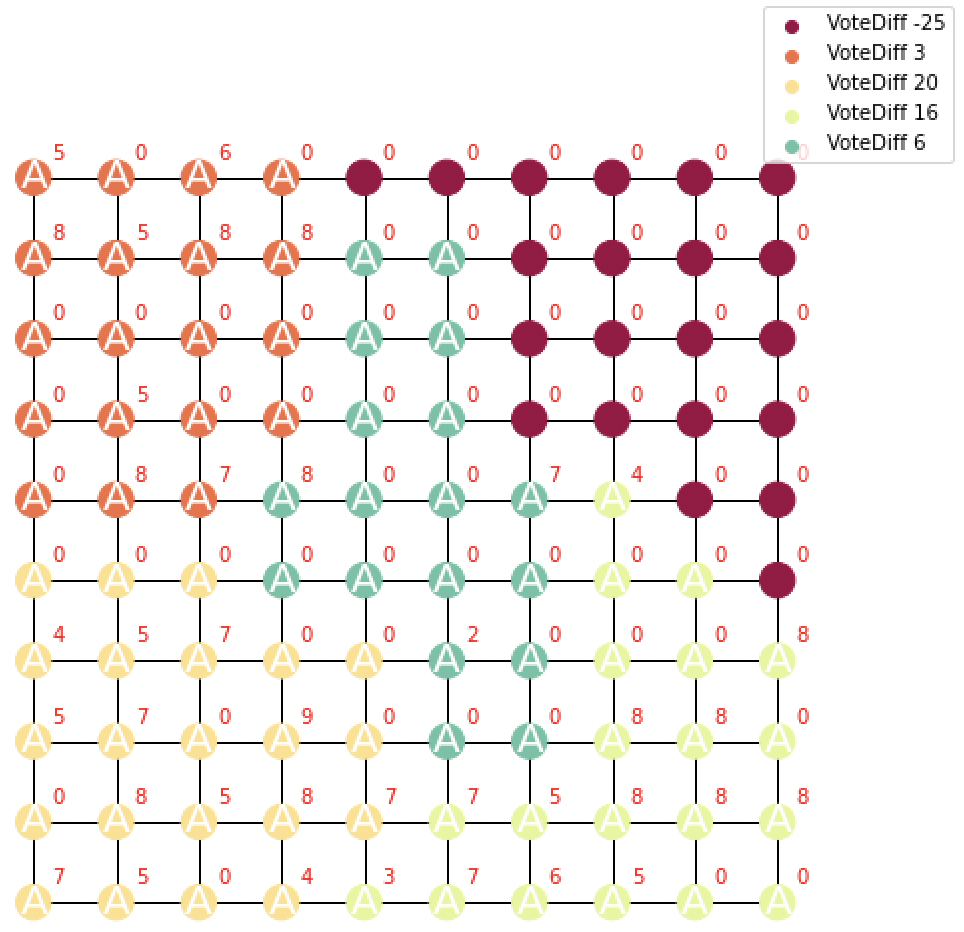}
         \caption{Target Map}
         \label{fig: investtmex}
     \end{subfigure}
 \caption{Investments as seen on both maps} \label{fig: investments}
\end{figure}

\begin{figure}
    \centering
    \begin{subfigure}[b]{0.45\textwidth}
         \centering
    \includegraphics[width=5cm]{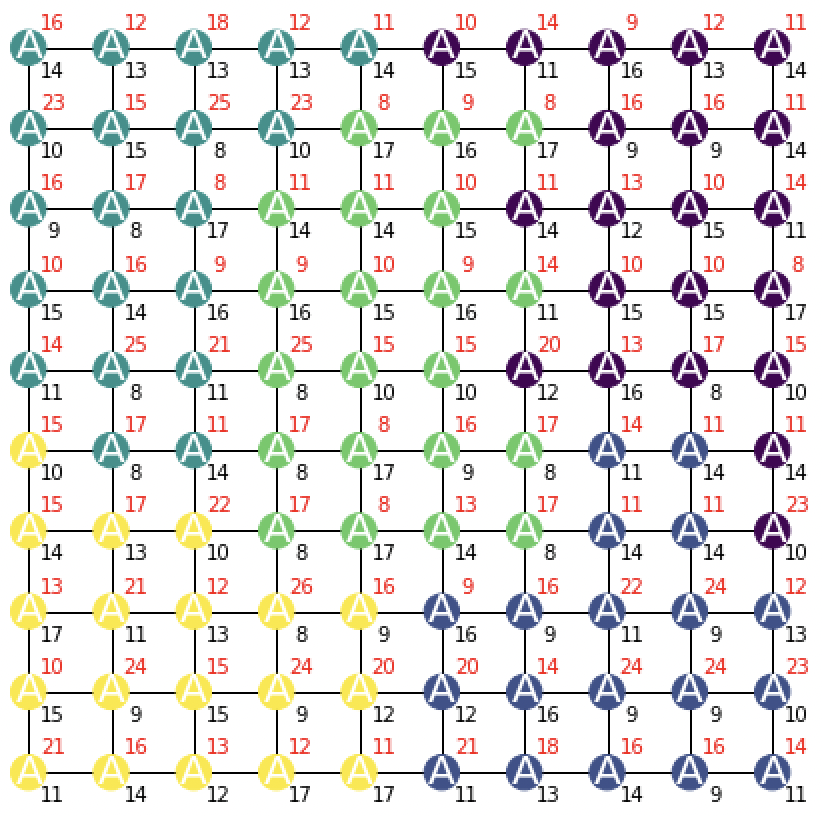}
         \caption{Round-1 (Campaigned map)}
         \label{fig: cmex}
     \end{subfigure}
     \begin{subfigure}[b]{0.45\textwidth}
         \centering
    \includegraphics[width=5cm]{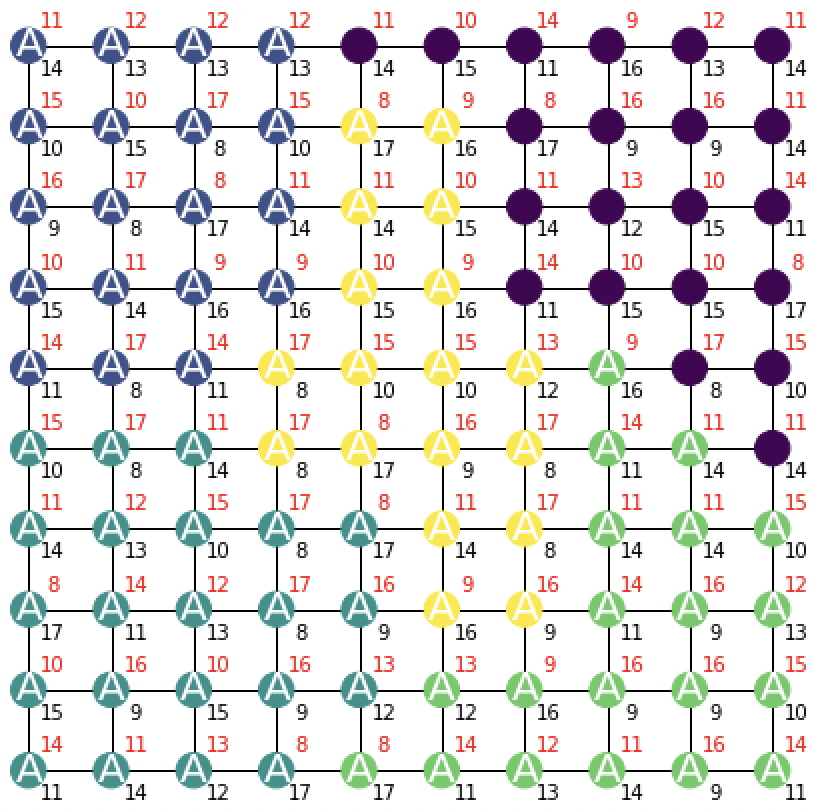}
         \caption{Round-2 (Target map)}
         \label{fig: trex}
     \end{subfigure}
   \caption{Final results}
    \label{fig: results}
\end{figure} 

In this example, party $B$ adds zero extra votes through campaigning. This is done to illustrate only the effect of A's campaigning on the election results and the proposed plan for round-2. Party B's campaigning has a different type of effect on votemandering as discussed in Section \ref{sec:results}. If party $A$ does not do any strategic campaigning or propose a different plan in round-2, it would win 3 seats in each round, making a total of 6 wins in two rounds. $A$ may choose to propose the plan in Figure \ref{fig: trex}, leading to 4 wins in round-2. However, the EG of this plan using the original vote shares, i.e., that of the target map is 0.28, making it an unfair proposal. We show that $A$ can successfully votemander in this case: via strategic campaigning, it manages to claim the fairness of this plan (winning 4 seats in round-2), while simultaneously winning all 5 seats in round-1 elections. Hence, as opposed to 6 without strategic campaigning, $A$ can win a total of 9 seats in two rounds.

We further illustrate the strategy implemented by $A$ through Figure \ref{fig: investments}.  The chosen plan for round-2 has multiple unit-to-district assignments same as that of the original plan, i.e., the two plans are not drastically different from each other. We see that $A$ has majorly invested in the districts it was already winning in the initial map. It also invests a sufficient budget in the districts it initially loses on, thereby winning all 5 in round-1 elections. This investment is clever: there is no investment in the units that are part of the losing district in the target map. The budget is allocated in such a way that we win that district (top-right positioned) in round-1 elections and we again lose it in the target map as well as the votemandered map (after redrawing its boundaries).  All investment then becomes a part of the winning districts in round-2, according to the new plan. 

To maintain the EG bound in the votemandered map,  $A$ needs to lose in at least one district. In the target map (using original vote shares), $A$ loses in one district and the map attains an EG score of 0.28, but the proposed votemandered map (using updated vote shares) gets an acceptable score of 0.195. The final results can be seen in Figure \ref{fig: results}. Using a budget of $~250$ for a state with a total voter population $5000$, i.e., by influencing just $5\%$ of the total voter population, $A$ can successfully votemander by winning 9/10 seats in two rounds. Moreover, if we extend this to 10 election rounds per redistricting cycle (as in the U.S.), the differences become starker, with 41 wins ($5 + 4\times9$) through votemandering as opposed to 30  wins ($3\times10$) without strategizing (out of 50).

\subsection{Proof of Proposition \ref{thm: algo}}
\label{app:sec_methodologyproof}
\begin{proof}
Algorithm \ref{algo: globalvotemandering} 
is a straightforward loop over all candidate plans in $\pool$. 
For each candidate plan $\districtplan_i$, 
Line \ref{line:fairness_step} computes the optimal objective value 
of \eqref{opt:votemander_mip} when $\districtplan_i$ is used. 
The \texttt{best\_plan} variable maintains the candidate plan 
with maximum objective value, \texttt{best\_obj}, 
among plans considered so far. 

It remains to prove the correctness of the termination condition 
in Line \ref{line:termination}. 
Observe the objective of \ref{opt:votemander_mip} 
decomposes into round-1 wins in the campaigned map, 
$s^1 \equiv \sum_{i \in K} \hat{s}_{i}^{1} 
    \equiv \election(\districtplan_0, \tilde{\voterballot})$, and 
round-2 wins in the target map, 
$s^2 \equiv \sum_{i \in K} \hat{s}_{i}^{2}
    \equiv \election(\tilde{\districtplan}, \voterballot_0)$. 
Fixing $\tilde{\districtplan} = \districtplan_i$ determines $s^2$ 
via the $z_{ik}^A$ variables, and 
$s^1$ becomes a function of 
the budget allocation variables $b_k^A$. 
The value of $s_{\text{max}}^{1}$ computed in Line \ref{line:find_s1_max} 
is found by removing the fairness constraints and second objective term from \eqref{opt:votemander_mip}, 
so $s_{\text{max}}^{1}$ is an upper bound on $s^1_i$ 
for any given plan $\districtplan_i$. 

Suppose the loop in Algorithm \ref{algo: globalvotemandering} 
terminates in iteration $j$ 
with $\texttt{best\_plan} = \districtplan_{i^*}$ 
for some $i^* < j$. 
For any $\ell \in \Set{j, j+1, \ldots, \poolsize}$, 
the maximum objective value using new plan $\districtplan_\ell$ is 
\begin{align*}
s^2 + \max_{\Set{b_k^A}_k : \text{ feasible in \eqref{opt:votemander_mip} with } \districtplan_\ell} \Paren{s^1}
&\le s^1_{\text{max}} + s^2\\
&= s^1_{\text{max}} + \election\Paren{\districtplan_\ell, \voterballot_0}\\
&\le s^1_{\text{max}} + \election\Paren{\districtplan_j, \voterballot_0}
&& \text{(by sorting of $\pool$)}\\
&< \texttt{best\_obj}
&& \text{(by Line \ref{line:termination})}\\
&= s^1_{i^*} + \election\Paren{\districtplan_{i^*}, \voterballot_0}.
\end{align*}
Hence none of the candidate plans omitted from consideration by loop termination 
could achieve a larger objective value for \eqref{opt:votemander_mip} than \texttt{best\_plan}. 
This completes the proof of Proposition \ref{thm: algo}. 
\end{proof}

\section{Section \ref{sec:results} Details} \label{app:sec_results}
\subsection{Proof of Lemma \ref{lem: tabvotes}} \label{app:sec_results-1}
\lem*
\begin{proof}
We discuss each action from Table \ref{tab: deltaWcases} and its impact on the difference between wasted votes below.
\begin{enumerate}
    \item  Wasting an additional vote on a winning district $I$: Before adding the extra vote, with votes $V_{I}^{A}$ and $V_{I}^{B}$ for both parties, the wasted votes are as follows 
    \begin{equation}
        w^A_I = V_{I}^{A}- \frac{V_{I}^{A}+V_{I}^{B}}{2}, \  \quad w_I^B = V_{I}^{B}
    \end{equation}
    Then, the initial difference between wasted votes is 
    \begin{equation}\mathcal{W}(B-A)_{init}(I)  = V_{I}^{B} - V_{I}^{A} +\frac{V_{I}^{A}+V_{I}^{B}}{2} \end{equation}
    After adding a vote to party $A$ , the wasted votes are updated as
     \begin{equation}w^A_I = V_{I}^{A}+1- \frac{V_{I}^{A}+1+V_{I}^{B}}{2}, \  \quad w^B_I = V_{I}^{B} \end{equation}
    Then,  \begin{equation}\mathcal{W}(B-A)_{final} (I) = V_{i}^{B}-V_{i}^{A}+ \frac{V_{i}^{A}+V_{i}^{B}}{2} - \frac{1}{2} = \mathcal{W}(B-A)_{init}(I)-\frac{1}{2} \end{equation}
    Resulting in $\Delta \mathcal{W}(B-A) = -\frac{1}{2} $. In words, if party $A$  adds a vote to a winning district, the difference between wasted votes changes by half a vote. 
    \item  Wasting an additional vote on a losing district: Before adding the extra vote, with votes $P_A$ and $P_B$ for both parties, the wasted votes are as follows:
     \begin{equation}w^A_I = V_{I}^{A}, \  \quad w^B_I = V_{I}^{B}- \frac{V_{I}^{A}+V_{I}^{B}}{2} \end{equation}
    Then,  \begin{equation}\mathcal{W}(B-A)_{init}(I) = V_{I}^{B}- \frac{V_{I}^{A}+V_{I}^{B}}{2}- V_{I}^{A} \end{equation}
    After adding a vote to party $A$ , the wasted votes are updated:
     \begin{equation}w^A_I = V_{I}^{A}+1, \  \quad w_I^B = V_{I}^{B}- \frac{V_{I}^{A}+V_{I}^{B}+1}{2} \end{equation}
    Then,  \begin{equation}\mathcal{W}(B-A)_{final}(I)  = V_{I}^{B}- \frac{V_{I}^{A}+V_{I}^{B}}{2}- V_{I}^{A} - \frac{3}{2} = \mathcal{W}(B-A)_{init}(I)-\frac{3}{2} \end{equation}
    Resulting in $\Delta \mathcal{W}(B-A) = -\frac{3}{2} $. For each vote added by $A$ in a losing district, the difference between wasted votes changes by $-3/2$ votes. 
    Clearly, if $A$ wants to decrease $W$ to satisfy the EG bound, it is more beneficial to waste votes in a losing district.
    \item Winning a district through campaigning: To win district $I$, $A$ just needs to add $V_{I}^{B}- V_{I}^{A}$ votes (assuming ties break in favor of $A$). Initially, the wasted votes difference can be computed as:
 \begin{equation}w^A_I = V_{I}^{A}, \ \quad w^B_I = V_{I}^{B} - \frac{V_{I}^{A}+V_{I}^{B}}{2} \end{equation}
 \begin{equation}\mathcal{W}(B-A)_{init}(I)  = V_{I}^{B} - \frac{V_{I}^{A}+V_{I}^{B}}{2} - V_{I}^{A}  \end{equation}
After party $A$ adds $V_{I}^{B} - V_{I}^{A}$ votes to win the district, the difference between wasted votes is updated as:
 \begin{equation}w^A_I = 0 \ \quad w^B_I = V_{I}^{B} \end{equation}
 \begin{equation}\mathcal{W}(B-A)_{final}(I)  = V_{I}^{B} \end{equation}
Then, the change in $W$ is computed as:
 \begin{equation}\Delta \mathcal{W}(B-A) = V_{I}^{B} - ( V_{I}^{B} - \frac{V_{I}^{A}+V_{I}^{B}}{2} - V^{A}_I) = \frac{3V_{I}^{A}+V_{I}^{B}}{2}  \end{equation}
Thus, given the initial vote count for district $I$, we can compute the change in the difference between wasted votes as $\frac{3V_{I}^{A}+V_{I}^{B}}{2}$.
\item Shift $x$ votes from a winning to a losing district: Using the analysis for $\Delta \mathcal{W}(B-A)$ when $A$ wastes votes on a winning district, we can compute $\mathcal{W}(B-A)$ when its reverse operation is performed. That is, when $A$ removes $x$ extra votes  from district $i$, the change in $\mathcal{W}(B-A)$ is:
 \begin{equation}\Delta \mathcal{W}(B-A) = \frac{x}{2} \end{equation}
Further, when $A$ distributes these votes in $B$'s winning district $J$,  $\Delta \mathcal{W}(B-A)$ is updated as:
 \begin{equation}\Delta \mathcal{W}(B-A) = \frac{x}{2} +\frac{-3}{2}\times x = -x \end{equation}
If $A$ chooses to redistribute these votes to another $A$'s winning district $l$, we get:
 \begin{equation}\Delta \mathcal{W}(B-A) = \frac{x}{2} + \frac{-1}{2}\times x = 0  \end{equation}
\item Shift $x$ votes from a losing to a winning district: Similar to the previous case, shift from a losing to a winning district results in a difference of $\Delta \mathcal{W}(B-A)= x$.
\end{enumerate}
This concludes the discussion of the strategy space of $A$ and its impact on the difference between wasted votes. \end{proof}
\subsection{Proof of Theorem \ref{thm: existence}} \label{app:sec_results-2}
\thmone*
\begin{proof} Since assignment $\mathcal{J}$ has a higher number of wins for the majority party compared to the initial map, the corresponding votemandering bonus is always positive. That is, even if we maintain the same number of wins in the first round after campaigning, the second round has an improvement. For a map with assignment $\mathcal{J}$, we then need to establish fairness by showing
\begin{align}
     \mathcal{W}(B-A)_{\mathcal{J}} = \mathcal{W}(B-A)_{\mathcal{I}} +   \Delta \mathcal{W} (B-A)_{\mathcal{I} \rightarrow \mathcal{J}} + \text{[campaign wasted votes]}\leq \text{ EG bound }
\end{align}
We know that the initial map is fair, i.e., $\mathcal{W}(B-A)_{\mathcal{I}} \leq$  EG bound. Since $\mathcal{J}$ has higher number of wins than $\mathcal{I}$, wlog, $\Delta \mathcal{W} (B-A)_{\mathcal{I} \rightarrow \mathcal{J}}$ is considered positive. If it's negative and makes $\mathcal{W}(B-A)_{\mathcal{J}} \leq -$ EG bound, we naturally establish that the map gives more advantage to the minority party as compared to the initial map (and is  acceptable), making a trivial case. If $\mathcal{W}(B-A)_{\mathcal{J}} \geq$  EG bound, then we  use Table \ref{tab: deltaWcases} to spend budget to satisfy the bound. However, this budget allocation is not trivial: the allocation of budget should not change the $W/L$ status of the districts in $\mathcal{J}$ and it needs to satisfy the individual unit voter-turnout constraints. Intuitively, the constraints may get violated under special cases like very high voter turnout or the election mandate hugely tilting towards a party. We next deduce sufficient conditions that allow such budget allocation to occur. 

For districts in $\mathcal{J}(W), \mathcal{J}(L)$, i.e., the winning and losing districts respectively, the capacities of budget allocation can be written as
\begin{align}
    \text{Total capacity of winning districts} & = \sum_{J\in \mathcal{J}(W)} \sum_{j \in J} (1-\alpha)v_i^A = c_1\\
    \text{Total capacity of losing districts} &  =\sum_{J\in \mathcal{J}(W)}  \min \{ \ \sum_{j \in J} (1-\alpha)v_i^A, \ V_J^B-V_J^A\} = c_2
\end{align}
Both $c_1$ and $c_2$ can be computed using the assignment $\mathcal{J}$. Using Table \ref{tab: deltaWcases}, this translates to a bound on the effect on wasted votes:
\[\text{Max. achievable difference in $\mathcal{W}(B-A)$ through budget allocation } = -\frac{1}{2}c_1- \frac{3}{2}c_2\]
Then, the sufficient condition for votemandering becomes 
\begin{align}
    \mathcal{W}(B-A)_{\mathcal{I}} +   \Delta \mathcal{W} (B-A)_{\mathcal{I} \rightarrow \mathcal{J}} -\frac{1}{2}c_1- \frac{3}{2}c_2 \leq \text{ EG bound }
\end{align}
We can simplify this condition further to deduce a (comparatively stringent) sufficient condition on $\alpha$ by asking if allocating only on $\mathcal{J}(W)$ can satisfy the bound:
\begin{align}
     0  & \geq \Delta \mathcal{W} (B-A)_{\mathcal{I} \rightarrow \mathcal{J}} -\frac{1}{2}c_1 \notag \\
   c_1 & = \sum_{J\in \mathcal{J}(W)} \sum_{j \in J} (1-\alpha)v_i^A \geq  2 \Delta \mathcal{W} (B-A)_{\mathcal{I} \rightarrow \mathcal{J}} \notag\\
     (1-\alpha) & \geq \left( \frac{2 \Delta \mathcal{W} (B-A)_{\mathcal{I} \rightarrow \mathcal{J}}}{\sum_{J\in \mathcal{J}(W)} \sum_{j \in J} v_i^A} \right) \notag\\
     \therefore \alpha & \leq 1- \left( \frac{2 \Delta \mathcal{W} (B-A)_{\mathcal{I} \rightarrow \mathcal{J}}}{\sum_{J\in \mathcal{J}(W)} \sum_{j \in J} v_i^A} \right)
\end{align}
Thus, we can see that the lesser the difference between the wasted votes of $\mathcal{I}$ and $\mathcal{J}$, i.e., $\Delta \mathcal{W} (B-A)_{\mathcal{I} \rightarrow \mathcal{J}}$, the higher the voter turnout votemandering strategies can handle.
\end{proof}

\subsection{Proof of Theorem \ref{thm: polytime}} \label{app:sec_results-3}
\thmtwo*
\begin{proof}
We elaborate on each of the steps outlined in the main body. Careful observation of the fairness MIP reveals that the first set of constraints \eqref{eq: budgetcapA}-\eqref{eq: budgetconstraint} track the effects of budget on updating vote shares, the second set of constraints \eqref{eq: seat1MIP}-\eqref{eq: seat2MIP} translates these effects into the  $W/L$ status of campaigned and votemandered map, the third set \eqref{eq: tau1}-\eqref{eq: tau2} further translates this into the fairness evaluation of the districts and finally the state in \eqref{eq: egMIP}. The objective only uses the variables $\hat{x}_I$ that indicates the $W/L$ status in the campaigned map. 

We use $b_k$ variables (over the space of $K$ units) for capturing the budget investment in each unit. Since all remaining constraints and objective use the win/lose variables ($\hat{x}_I, \hat{y}_J$) using the set notation, we can simplify the $b_k$ variables space to just capture the budget invested in \emph{pieces} formed by overlapping each district in $\mathcal{I}$  with each district in $\mathcal{J}$. We let variables $z_{IJ}$ denote the budget invested in the set of units that belong to district $I$ and $J$ in round-1 and 2, respectively. As outlined in the idea sketch, the spending capacity of each piece is defined as $c_{IJ}= \min$(division voter-turnout capacity, budget needed to win the round-2 district it is part of (if it's a losing district)).

Further, when we fix the $W/L$ status in the votemandered map, i.e., variables $\hat{y}_J$, we simplify the constraints by eliminating the $\tau$ variables. Finally, we see that the MIP is reduced to  the following linear formulation:
\begin{align}
      \max \  \sum_{I\in \mathcal{I}} \hat{x}_{I} \notag \\
        s.t. \ \ 
       &  1-M(1-\hat{x}_{I}) \leq  (V_I^{A}- V_I^{B}) \leq M \hat{x}_{I} & \forall I \in \mathcal{I}  \notag \\
       & \sum_{I\in \mathcal{I}} z_{IJ} \leq V_{J}^B-V_J^A & \forall J \in \mathcal{J}(L) \notag \\
       & \sum_{I, J \in \mathcal{I, J}} z_{IJ} \leq \mathcal{B}^A & \forall I \in \mathcal{I} \notag \\
       & \frac{1}{2}\sum_{J\in \mathcal{J}(W)} Z_{IJ} +\frac{3}{2}\sum_{J\in \mathcal{J}(L)} Z_{IJ} \geq \mathcal{W}(B-A)_{\mathcal{J}}-\text{ EG bound } \notag \\
       &  z_{IJ} \leq c_{IJ} & \forall I, J \in \mathcal{I, J} \notag \\
      & \hat{x}_{I} \in (0,1),  \ \ \ z_{IJ} \geq 0 & \forall I, J \in \mathcal{I,J} 
    \label{prog: linearfairness}
  \end{align}
After solving this linear program, we check if the optimal solution involves a tight $\sum_{I\in \mathcal{I}} z_{IJ'} \leq V_{J'}^B-V_{J'}^A$ for a $J' \in \mathcal{J}(L)$. If not, this implies that the optimal investment leading to maximum seats in round-1 does not need to change the $W/L$ status of any district in the votemandered map. Otherwise, a tight constraint implies a change in $W/L$ status of a district, i.e., $\hat{y_{J}}$ towards optimality which changes the EG significantly and non-continuously. We thus make the required change in constraints and solve the updated linear program again. Corresponding to district $J'$, the updates include replacing  $\sum_{I\in \mathcal{I}} z_{IJ'} \leq V_{J'}^B-V_{J'}^A $ with 
\[\sum_{I\in \mathcal{I}} z_{IJ'} \geq V_{J'}^B-V_{J'}^A \]
and replacing the fairness constraint with 
\[\frac{1}{2}\sum_{J\in \mathcal{J}(W) \cup J'} Z_{IJ} - \frac{1}{2} (V_{J'}^B-V_{J'}^A) +\frac{3}{2}\sum_{J\in \mathcal{J}(W)/J'} Z_{IJ} \geq \mathcal{W}(B-A)_{\mathcal{J}}-\text{ EG bound } +\frac{3V_{J'}^B+V_{J'}^A}{2}\]
or 
\[\frac{1}{2}\sum_{J\in \mathcal{J}(W) \cup J'} Z_{IJ} +\frac{3}{2}\sum_{J\in \mathcal{J}(W)/J'} Z_{IJ} \geq \mathcal{W}(B-A)_{\mathcal{J}}-\text{ EG bound } +2V_{J'}^B\]
In case we have two tight constraints at any step, we  can greedily choose the district $J'$ that has the less fairness cost in terms of $\frac{3V_{J'}^B+V_{J'}^A}{2}$. This way, we have to solve at most $n$ linear programs to reach an optimal solution to the MIP, which is the case when every update to the linear program produces an optimal solution that improves the objective. Thus, we can find an optimal solution using polynomial efforts.
\end{proof}

\section{Section 4 details} \label{app: 4}
\subsection{Local votemandering example}\label{app: 4a}
Here we illustrate the local votemandering heuristic through a grid graph example. We consider a $20 \times 20$ grid with fixed vote shares for parties $A$ and $B$, forming 10 districts, with other parameters the same as our settings in section 3. The initial map is given in Figure \ref{fig: localvmIM}, in which $A$ wins 4 seats.

\begin{figure}
    \centering
    \begin{subfigure}[b]{0.45\textwidth}
         \centering
          \includegraphics[width=5cm]{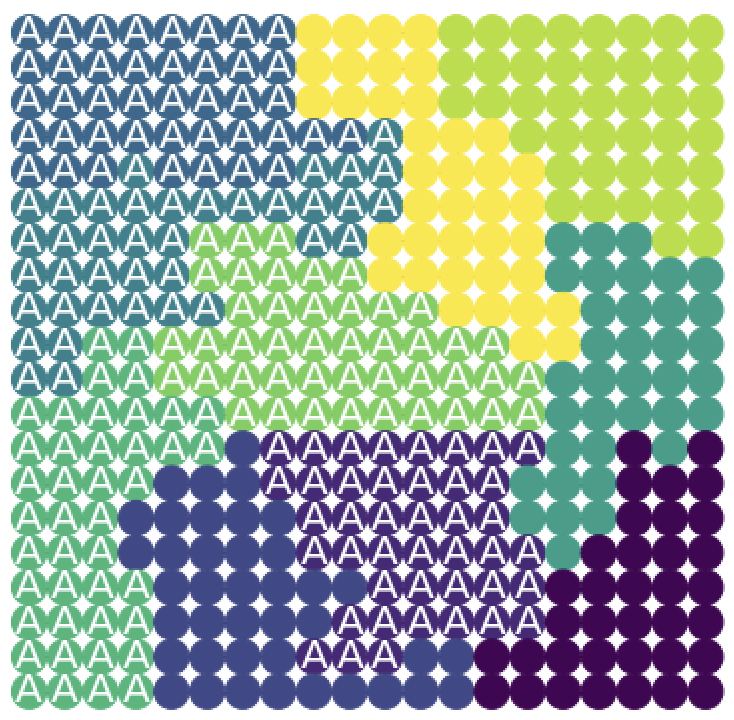}   
         \caption{Initial Map}
         \label{fig: localIM}
     \end{subfigure}
     \hfill
     \begin{subfigure}[b]{0.45\textwidth}
         \centering
        \includegraphics[width=5.1cm]{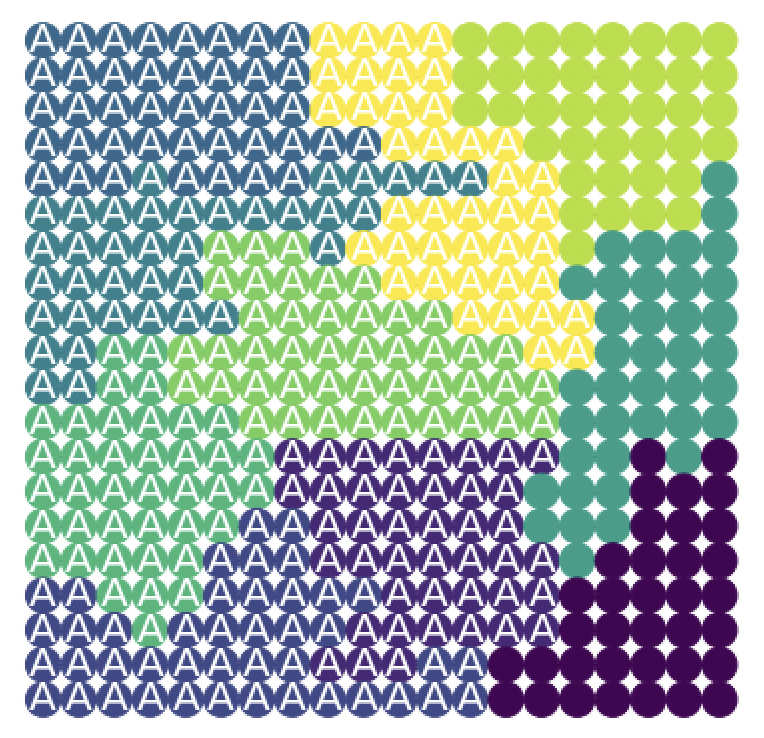}
         \caption{Target Map}
         \label{fig: localTM}
     \end{subfigure}
    \caption{Local Votemandering: Initial and Target Maps}
    \label{fig: localvmIM}
\end{figure} 

For each pair of neighbors in the initial map, we find strategy 1,2, and 3 edges and their corresponding weights. The edge weights are computed using recombination local search. Finally, the matching problem and its solution are shown in Figure \ref{fig: localvm}. The edges with nonzero weights signify strategy-1 edges, 0 weights are corresponding to strategy-2 and `FC' refers to fairness costs from strategy-3. The final target map is shown in Figure \ref{fig: localvmIM}. This local votemandering solution gives a votemandering bonus of 3, spending a budget of 174. 
\begin{figure}
    \centering
    \begin{subfigure}[b]{0.45\textwidth}
         \centering
          \includegraphics[width=4cm]{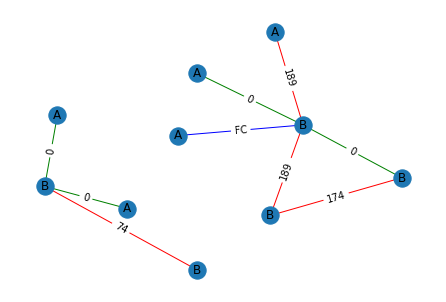}
         \caption{Problem Instance}
         \label{fig: localvmmatching}
     \end{subfigure}
     \hfill
     \begin{subfigure}[b]{0.45\textwidth}
         \centering
        \includegraphics[width=4cm]{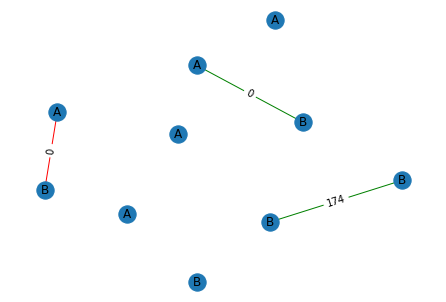}
         \caption{Maximum Matching Solution}
         \label{fig: localvmsol}
     \end{subfigure}
    \caption{Local Votemandering: Heuristic}
   \label{fig: localvm}
\end{figure}

\subsection{Local votemandering strategies}\label{app: 4b}
\textbf{Strategy-1:}\\
Consider a submap of a pair of neighboring districts $(D_1, D_2)$. In this strategy, we allocate sufficient budget in a currently losing district $D_1$ (change $D_1$ status from $L \rightarrow W$ in the first round elections i.e. in the campaigned submap) and propose perturbations with its neighbor $D_2$ so that it again loses in the second round (mark $D_1$ status as $L$ in the second round i.e. in the votemandered/target submap). We thus increase the votemandering bonus by 1, via a campaigned map win at $D_1$. 
The budget required here is simply the margin with which $D_1$ loses in the campaigned map ($margin(D_1$)) with no investment. Naturally, the perturbations for the new submap involve exchanging units across districts $D_1, D_2$ that nullify the effect of budget investment and regain the $L$ status of $D_1$ in the votemandered submap.  Let this imply a movement of $V_A^T$ party $A$ votes and $V_B^T$ party $B$ votes from $D_2$ to $D_1$. If $(D_1, D_2)$ have status $(L, W)$ in the initial submap, Lemma \ref{lem: tabvotes} implies:
\begin{equation}
     \Delta \mathcal{W}(B-A)_i \text{ for } (L, W) = -V_A^T - V_B^T -\frac{a_1}{2}-\frac{3a_2}{2}+\frac{b_2}{2}+\frac{3b_1}{2}
\end{equation}
where, $a_i, b_i$ is the budget allocation in district $i$ (in the initial submap) by party $A$ and $B$ respectively. Similarly, if the current status is $(L, L)$, we have:
\begin{equation}
    \Delta \mathcal{W}(B-A)_i \text{ for } (L, L) = -V_A^T - V_B^T -\frac{a_1}{2}-\frac{a_2}{2}+\frac{3b_2}{2}+\frac{3b_1}{2}
\end{equation}
 Strategy-1 edge weight is then given as ($margin(D_1), \Delta \mathcal{W}(B-A)_i$). Note the fairness cost from the initial map to the votemandered map is insignificant since the votemandered submap retains the $W/L$ status of both districts. 

\noindent
\textbf{Strategy-2:}\\
In this strategy, we let party $B$ win $D_1$ in campaigned and the votemandered submap and propose perturbations between $D_1, D_2$ such that $D_1$ wins in the second round (mark from $L\rightarrow W$  in the target submap). The target submap is chosen such that the votemandered submap leads to no change in the wins, but the target submap (using original vote shares) secures a votemandering bonus of 1. This is possible $margin(D_1)$ is smaller than $B$'s budget allocation in $D_1$. This also suggests that with an increase in $B$'s budget allocation, $A$ can specifically use strategy-2 to votemander more efficiently. 
The change in $\Delta \mathcal{W}$ for this setting is exactly the same as for strategy-1, giving strategy-2 edge weight as ($0, \Delta \mathcal{W}(B-A)_i$). Moreover, we don't expect $\Delta \mathcal{W}$ to change a lot here as well, since there is no change in the number of wins where the fairness constraint is concerned, i.e. in the votemandered submap.

\noindent
\textbf{Strategy-3:}\\
In this strategy, we let party $B$ win a district in the first round but propose perturbations such that $A$ wins in both the votemandered as well as the target submaps (mark from $L\rightarrow W$ in the second round elections). The difference with strategy-2 is that here we don't depend on $B$'s budget allocation to secure a win, but achieve only via the changes in the district boundaries i.e. the local perturbations. The target submap is chosen such that both the votemandered submap and the target submap have $D_1$ as a winning district. Naturally, here we expect $\Delta \mathcal{W}$ to change significantly that of the initial map, as $A$'s number of wins increases in the votemandered submap. 

If $(D_1, D_2)$ have status $(L, W)$ and $V_{A}^1, V_B^1$ are the original vote shares of district $D_1$, Lemma \ref{lem: tabvotes} implies:
\begin{equation}
    \Delta \mathcal{W}(B-A)_i \text{ for } (L, W) =    -V_A^1 - V_B^1 -\frac{a_1}{2}-\frac{3a_2}{2}+\frac{b_2}{2}+\frac{3b_1}{2}
\end{equation}
where, $a_i, b_i$ are the budget allocations in district $i$ by party $A$ and $B$ respectively. Similarly, if the status is $(L,L)$, we have:
\begin{equation}
    \Delta \mathcal{W}(B-A)_i \text{ for } (L, L) = -V_A^1 - V_B^1 -\frac{a_1}{2}-\frac{a_2}{2}+\frac{3b_2}{2}+\frac{3b_1}{2}
\end{equation}

\section{Section 5 Details} \label{app: sec5}
We include the visualization of local votemandering strategies for both Republican and Democratic parties in Figures \ref{fig: WIlocal1R}- \ref{fig: WIlocal2D}.
\begin{figure}
    \centering
    \begin{subfigure}[b]{0.32\textwidth}
         \centering
            \includegraphics[width=3cm]{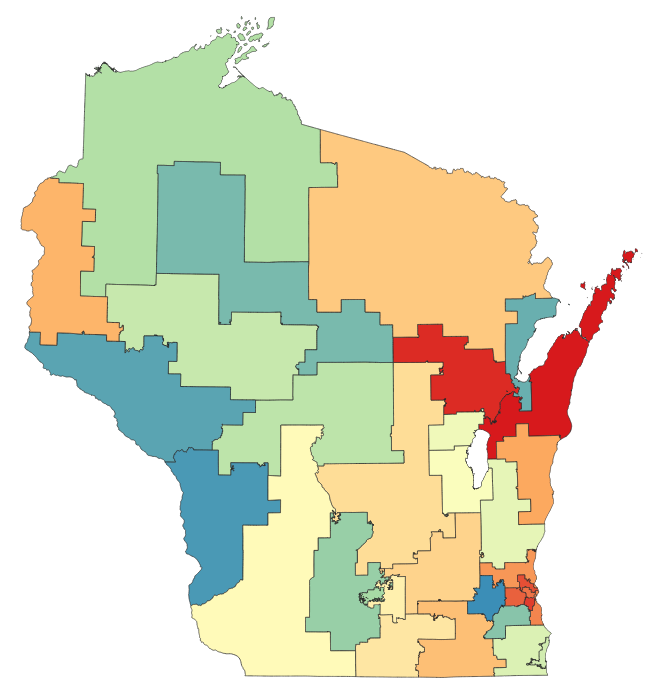}
         \caption{Initial Map}
         \label{fig: initialmaplocalr}
     \end{subfigure}
     \begin{subfigure}[b]{0.32\textwidth}
         \centering
    \includegraphics[width=3cm]{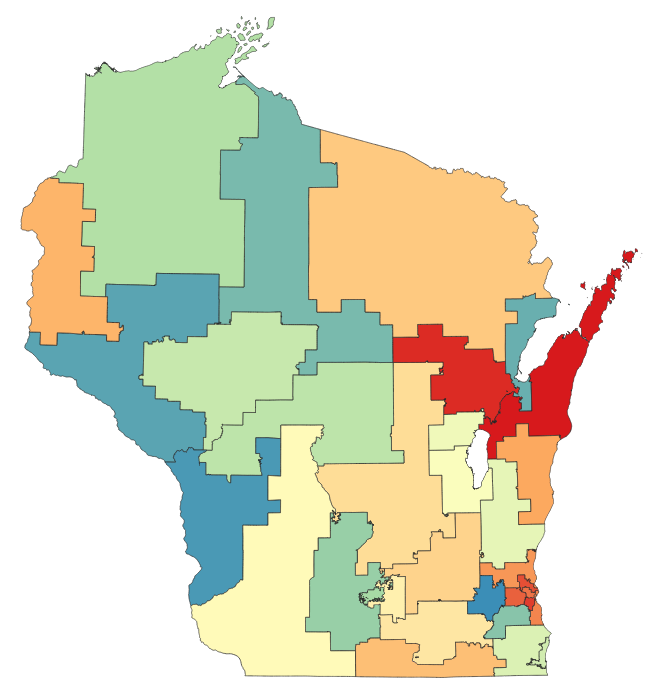}
         \caption{Target Map}
         \label{fig: targetmaplocalr}
     \end{subfigure}
     \begin{subfigure}[b]{0.32\textwidth}
         \centering
        \includegraphics[width=3cm]{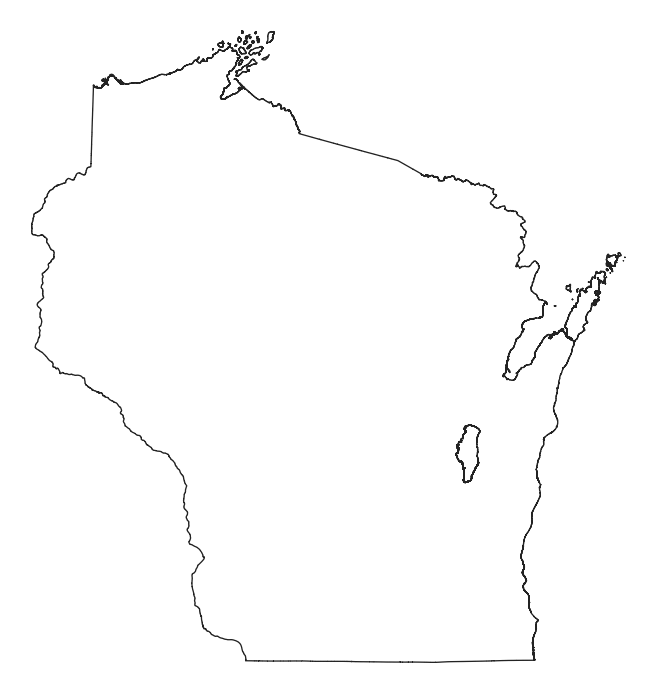}
         \caption{Investment}
         \label{fig: investmentlocalr}
     \end{subfigure}
   \caption{Republican Local Votemandering: The initial map, the created target map, and the strategic investment of budget (here, zero investment followed by all strategy-2 improvements)}
    \label{fig: WIlocal1R}
\end{figure} 
\begin{figure}
    \centering
    \begin{subfigure}[b]{0.24\textwidth}
         \centering
             \includegraphics[width=2.2cm]{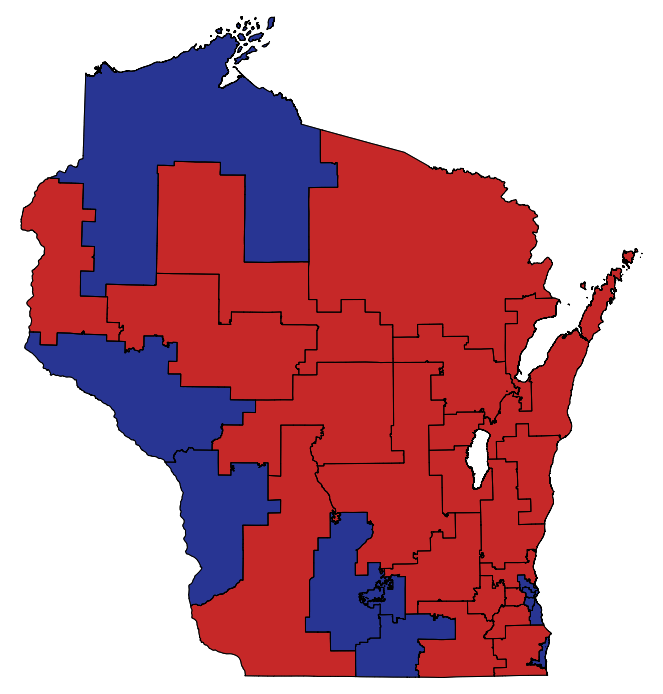}
         \caption{Initial Map}
         \label{fig: initmaplocalwinsr}
     \end{subfigure}
     \begin{subfigure}[b]{0.24\textwidth}
         \centering
       \includegraphics[width=2.2cm]{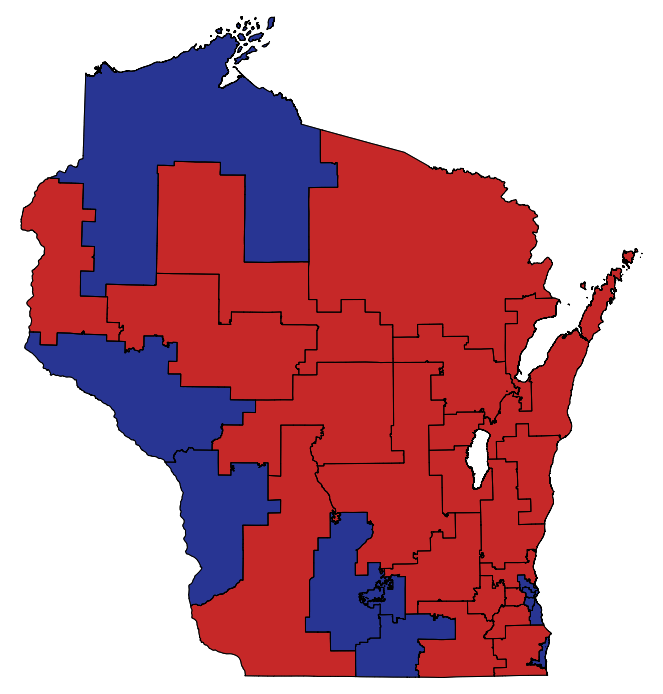}
         \caption{Campaigned Map}
         \label{fig: campaignedmaplocalwinsr}
     \end{subfigure}
     \centering
    \begin{subfigure}[b]{0.24\textwidth}
         \centering
    \includegraphics[width=2.2cm]{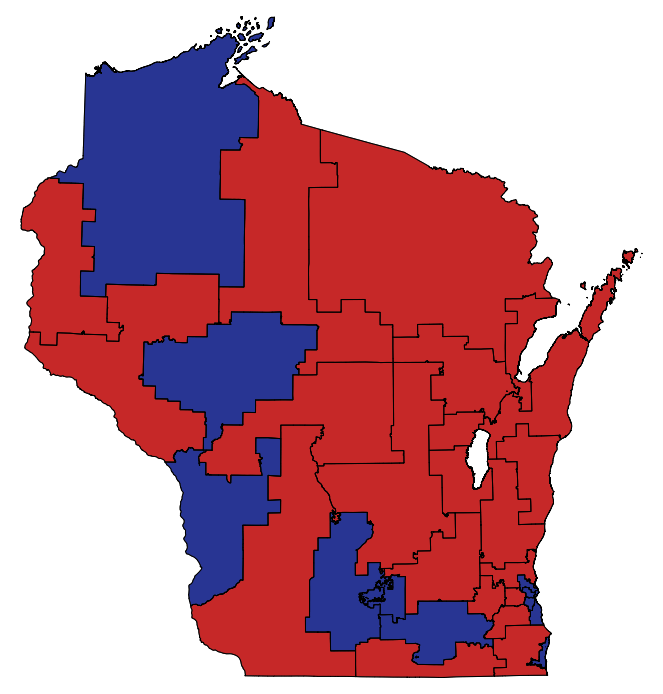}
         \caption{Votemandered Map}
         \label{fig: vmmaplocalwinsr}
     \end{subfigure}
     \begin{subfigure}[b]{0.24\textwidth}
         \centering
      \includegraphics[width=2.2cm]{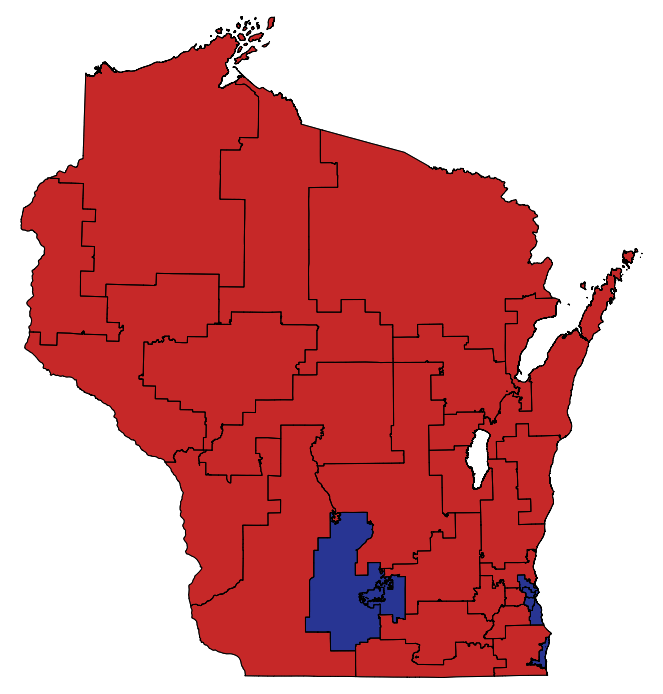}
         \caption{Target Map}
         \label{fig: targetmaplocalwinsr}
     \end{subfigure}
   \caption{The Four Stages of Republican Local Votemandering, (with red and blue indicating the districts won by the Republican and Democratic parties, respectively)}
    \label{fig: WIlocal2R}
\end{figure} 
\begin{figure}
    \centering
    \begin{subfigure}[b]{0.32\textwidth}
         \centering
            \includegraphics[width=3cm]{initmapglobD.png}
         \caption{Initial Map}
         \label{fig: initialmaplocald}
     \end{subfigure}
     \begin{subfigure}[b]{0.32\textwidth}
         \centering
    \includegraphics[width=3cm]{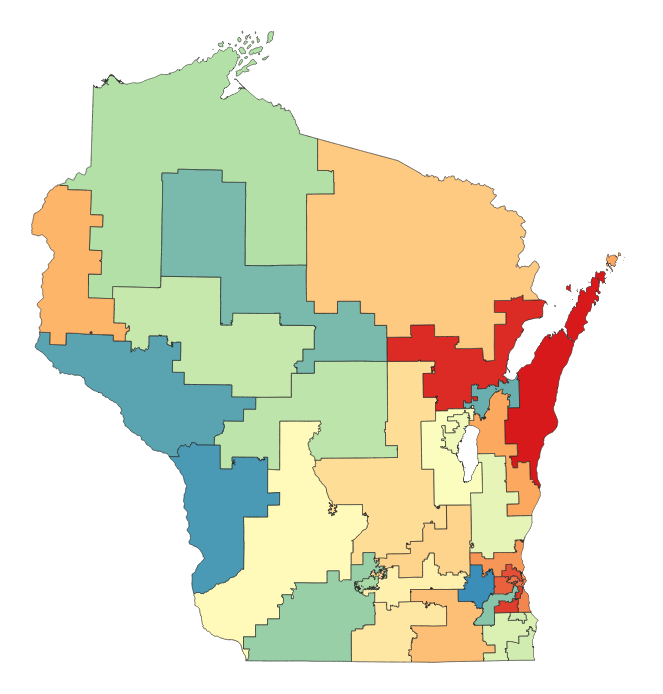}
         \caption{Target Map}
         \label{fig: targetmaplocald}
     \end{subfigure}
     \begin{subfigure}[b]{0.32\textwidth}
         \centering
        \includegraphics[width=3cm]{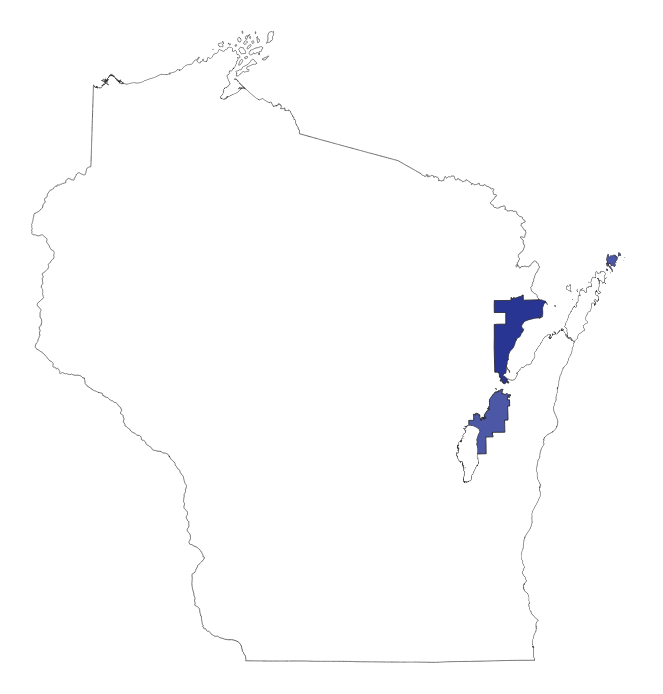}
         \caption{Investment}
         \label{fig: investmentlocald}
     \end{subfigure}
   \caption{Democratic Local Votemandering: The initial map, the created target map, and the strategic investment of budget (with intensity  indicated by the darker color)}
    \label{fig: WIlocal1D}
\end{figure} 
\begin{figure}
   \begin{subfigure}[b]{0.24\textwidth}
         \centering
             \includegraphics[width=2.2cm]{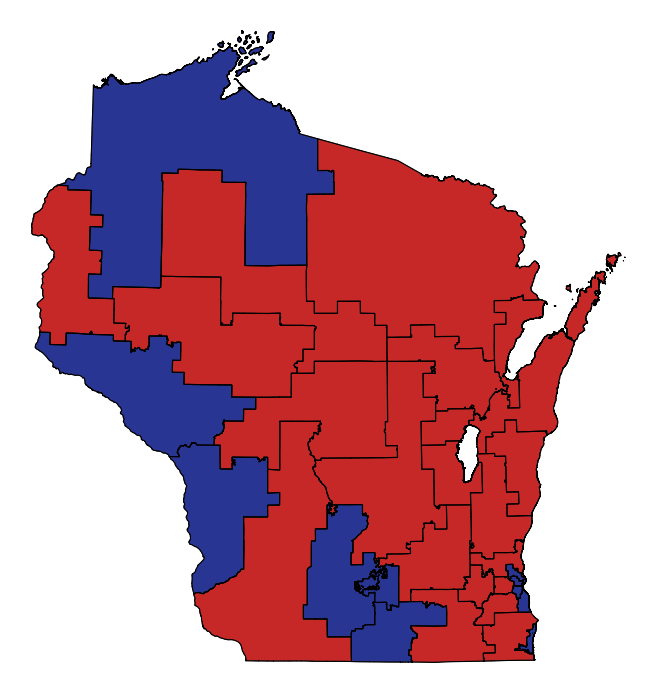}
         \caption{Initial Map}
         \label{fig: initmaplocalwinsd}
     \end{subfigure}
     \begin{subfigure}[b]{0.24\textwidth}
         \centering

       \includegraphics[width=2.2cm]{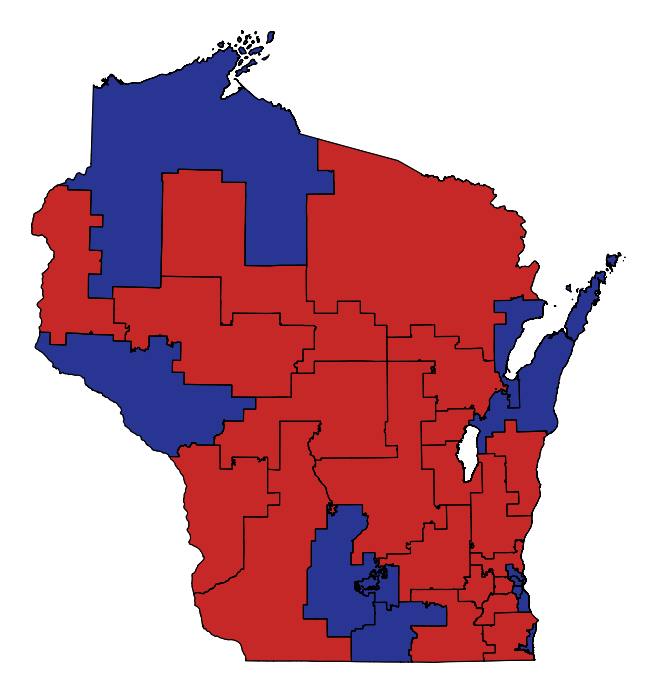}
         \caption{Campaigned Map}
         \label{fig: campaignedmaplocalwinsd}
     \end{subfigure}
     \centering
    \begin{subfigure}[b]{0.24\textwidth}
         \centering
    \includegraphics[width=2.2cm]{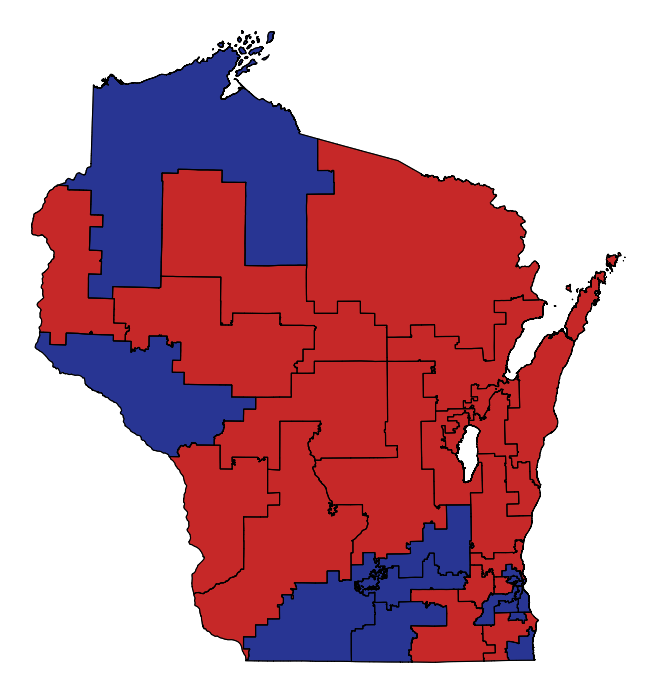}
         \caption{Votemandered Map}
         \label{fig: vmmaplocalwinsd}
     \end{subfigure}
     \begin{subfigure}[b]{0.24\textwidth}
         \centering
      \includegraphics[width=2.2cm]{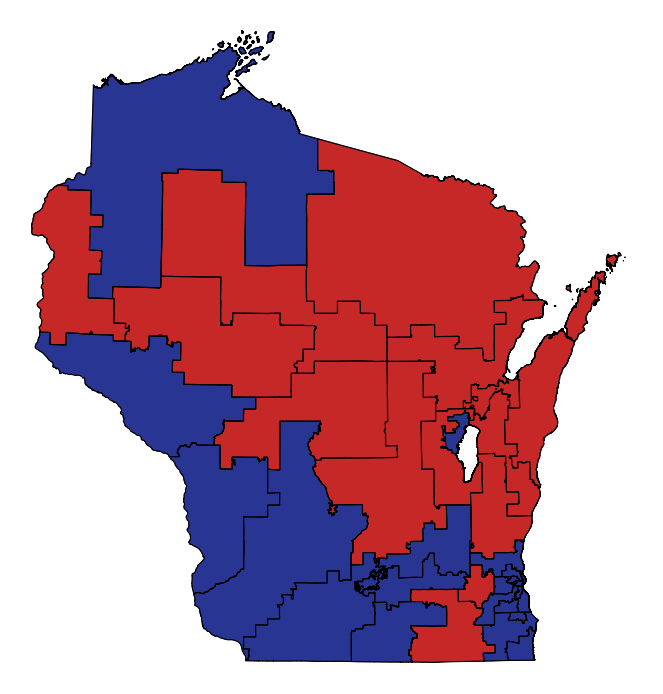}
         \caption{Target Map}
         \label{fig: targetmaplocalwinsd}
     \end{subfigure}
   \caption{The Four Stages of Democratic Local Votemandering, (with red and blue indicating the districts won by the Republican and Democratic parties, respectively)}
    \label{fig: WIlocal2D}
\end{figure} 

\end{document}